\documentclass[a4paper,USenglish,cleveref,autoref,thm-restate,final]{lipics-v2021}
\nolinenumbers
\hideLIPIcs

\usepackage[utf8]{inputenc}
\usepackage{booktabs}
\usepackage{dsfont}
\usepackage{amsmath}
\usepackage{amssymb}
\usepackage{graphicx}
\usepackage[textsize=tiny]{todonotes}
\usepackage{tikz}
\usepackage{thmtools} 
\usepackage{enumerate}

\DeclareMathOperator{\col}{col}
\DeclareMathOperator{\MOV}{MoV}
\DeclareMathOperator{\MaM}{MmM}

\newcommand{\smin}{p}
\newcommand{\smax}{q}

\newcommand{\oc}{\mathsf{Count}}
\newcommand{\ps}{+}
\newcommand{\ngt}{-}

\newcommand{\problem}[3]{
	\begin{center}
		\fbox{~\begin{minipage}{.97\textwidth}
			\vspace{2pt}
			\noindent
			\normalsize\textsc{#1}

			\vspace{4pt}
			\setlength{\tabcolsep}{3pt}
			\renewcommand{\arraystretch}{1.0}
			\begin{tabularx}{\textwidth}{@{}lX@{}}
				\normalsize\textbf{Input:}	& \normalsize#2 \\
				\normalsize\textbf{Question:}		 & \normalsize#3
			\end{tabularx}
		\end{minipage}}
	\end{center}
}

\usepackage{mathtools}

\makeatletter
\providecommand*{\cupdot}{%
	\mathbin{%
		\mathpalette\@cupdot{}%
	}%
}
\newcommand*{\@cupdot}[2]{%
	\ooalign{%
		$\m@th#1\cup$\cr
		\hidewidth$\m@th#1\cdot$\hidewidth
	}%
}

\DeclareMathOperator*{\argmax}{arg\,max}

\makeatother
\usepackage{etoolbox}

\newcommand{\lv}[1]{}

\newcommand{\appendixText}{}

\newcommand{\toappendix}[1]{\gappto{\appendixText}{{#1}}}

\newcommand{\appmark}{$\star$}

\usetikzlibrary{arrows,decorations.pathmorphing,decorations.pathreplacing,backgrounds,positioning,fit,matrix}
\usetikzlibrary{shapes,calc}
\usetikzlibrary{shapes.geometric}
\tikzstyle{vertex}=[draw, circle, fill, inner sep = 2.4pt]
\tikzstyle{square}=[draw, fill, inner sep = 3.6pt]
\tikzstyle{star4}=[draw, star, star points=4, fill, inner sep = 2.4pt]
\tikzstyle{triangle}=[draw, regular polygon, regular polygon sides=3, fill, inner sep = 2pt]
\tikzstyle{itriangle}=[draw, regular polygon, regular polygon sides=3,
rotate=180, fill, inner sep = 2pt]

\author{Niclas Boehmer}{Technische Universität Berlin, Algorithmics and Computational Complexity, Germany}{niclas.boehmer@tu-berlin.de}{https://orcid.org/0000-0001-5102-449X}{Supported by the Deutsche Forschungsgemeinschaft (DFG) projects MaMu (NI 369/19) and ComSoc-MPMS (NI 369/22)}

\author{Tomohiro Koana}{Technische Universität Berlin, Algorithmics and Computational Complexity, Germany}{tomohiro.koana@tu-berlin.de}{https://orcid.org/0000-0002-8684-0611}{Supported by the Deutsche Forschungsgemeinschaft (DFG) project DiPa (NI 369/21).}

\authorrunning{N. Boehmer and T. Koana} 

\Copyright{Niclas Boehmer and Tomohiro Koana} 



\acknowledgements{We are grateful to the anonymous \textit{ICALP 2022} reviewers for their 
thoughtful, constructive, and helpful comments. We thank Tom McCormick for helping us simplifying the proof for the touching separation theorem.}
 
\EventEditors{Miko{\l}aj Boja\'{n}czyk, Emanuela Merelli, and David P. Woodruff}
\EventNoEds{3}
\EventLongTitle{49th International Colloquium on Automata, Languages, and Programming (ICALP 2022)}
\EventShortTitle{ICALP 2022}
\EventAcronym{ICALP}
\EventYear{2022}
\EventDate{July 4--8, 2022}
\EventLocation{Paris, France}
\EventLogo{}
\SeriesVolume{229}
\ArticleNo{69}

\begin{document}

\title{The Complexity of Finding Fair Many-to-One Matchings}

\maketitle

\begin{abstract}
We analyze the (parameterized) computational complexity of ``fair'' variants of bipartite many-to-one matching, where each vertex from the ``left'' side is matched to exactly one vertex and each vertex from the ``right'' side may be matched to multiple vertices. 
We want to find a ``fair'' matching, in which each vertex from the right side is matched to a ``fair'' set of vertices.
Assuming that each vertex from the left side has one color modeling its ``attribute'', we study two fairness criteria.
For instance, in one of them, we deem a vertex set fair if for any two colors, the difference between the numbers of their occurrences does not exceed a given threshold.
Fairness is, for instance, relevant when finding many-to-one matchings between students and colleges, voters and constituencies, and applicants and firms.
Here colors may model sociodemographic attributes, party memberships, and qualifications, respectively.

We show that finding a fair many-to-one matching is NP-hard even for three colors and maximum degree five.
Our main contribution is the design of fixed-parameter tractable algorithms with respect to the number of vertices on the right side. 
Our algorithms make use of a variety of techniques including color coding.
At the core lie integer linear programs encoding Hall like conditions.
To establish the correctness of our integer programs, we prove a new separation result on (super)modular functions, inspired by Frank's separation theorem [Frank, Discrete Math. 1982], which may also be of independent interest.
We further obtain complete complexity dichotomies regarding the number of colors and the maximum degree of each side.
\end{abstract}
  
\section{Introduction}
A many-to-one matching in a bipartite graph $G=(U\cupdot V,E)$ is an edge subset $M\subseteq E$ such that each vertex in $U$ is incident to at most one edge in $M$.
We study the computational complexity of finding a ``fair'' many-to-one matching and call this problem \textsc{Fair Matching}: 
Given a bipartite graph $G=(U\cupdot V,E)$ in which every vertex in $U$ is colored, it asks for a many-to-one matching $M$ such that for each $v\in V$ the vertices matched to $v$ meet a fairness criterion.
In this work, we require that $M$ is ``left-perfect'', i.e., every vertex in $U$ is incident to exactly one edge in $M$.
Using a slightly different formulation, Stoica et al.~\cite{DBLP:conf/atal/StoicaCDG20} recently studied this problem in terms of a fairness requirement derived from \emph{margin of victory} (\textsc{MoV}).
Generally speaking, the margin of victory of a multiset is defined as the number of occurrences of the most frequently occurring element minus the number of occurrences of the second most frequently occurring element. (Given a set of colored vertices, we obtain a multiset from the occurrences of colors in it.)
By requiring that the margin of victory of a set of colored vertices shall not exceed a given threshold, we prevent one color from becoming a dominating majority
(see Stoica et al.~\cite{DBLP:conf/atal/StoicaCDG20} for a more extensive motivation of this concept).
As an alternative simple fairness measure, we consider \textsc{Max-Min}, which is defined as the difference between the number of occurrences of the most frequently occurring element and the number of occurrences of the least frequently occurring element in a multiset.
In a set of colored vertices with a small value of \textsc{Max-Min} all colors appear more or less equally often.

Which of \textsc{MoV} or \textsc{Max-Min} is more appropriate depends not only on the specific application (as discussed in the next two paragraphs) but also on the underlying data.
Suppose that we have $2n$ red, $2n$ blue, and one green vertex on the left side and two vertices on the right side.
Then, it would be natural to deem a subset consisting of $n$ red, $n$ blue, and one green vertex fair (as it is in some sense the best we can hope for).
Accordingly, \textsc{MoV} seems to be a better fit because the \textsc{MoV} of the described subset is zero whereas the \textsc{Max-Min} value is $n-1$.
In contrast, if there are $2n$ red, $2n$ blue, and $2n$ green vertices, then the same subset with $n$ red, $n$ blue, and one green vertex should be considered as unfair, rendering \textsc{Max-Min} more suitable for this color distribution than \textsc{MoV}.
In general, \textsc{Max-Min} seems to be a natural choice for homogeneous data.
The first example illustrates, however, that in some scenarios, \textsc{MoV} may serve as a viable relaxation of \textsc{Max-Min}.\footnote{Note that the \textsc{Max-Min} value is at least the \textsc{MoV} value for any multiset.}

A notable application of \textsc{Fair Matching} emerges in the context of district-based elections. 
In such elections, voters (modeled by vertices in $U$) are divided into constituencies (modeled by vertices in $V$), and then 
each constituency elects its own representative.
Here, colors can represent various attributes.
For instance, colors may represent political standings. 
A small margin of victory is particularly desirable in this case because it will lead to close elections, holding politicians accountable for their job. 
One could also strive for ``fair'' representation of different ethnic groups or age groups by modeling ethnicity or age with colors.\footnote{Notably in \cite{banerjee2007parochial} and \cite[pp. 251-252]{banerjee2011poor}, Banerjee and Duflo reported that, in particular in developing countries, districts that are dominated by one ethnicity are a serious problem, as candidates belonging to the dominating ethnicity often win independent of their merit. As a result, those candidates are, among others, more likely to be corrupt.}
Other applications include the assignment of school children to schools (where colors may model sociodemographic attributes) or the assignment of reviewers to academic papers (where colors may model the level of expertise or academic background of reviewers).

Similar fairness considerations also arise in modern online systems (see, e.g., \cite{DBLP:journals/corr/abs-2001-09784} for a survey).
For instance, fairness is a pressing issue to counter targeted advertising or to improve recommender systems. 
Here one task is to ensure that the content (each perhaps represented by multiple vertices in $U$) falling into different categories (colors) is assigned to users (vertices in $V$) in a way that each user is presented with a ``diverse'' selection of content.
Lastly, we mention that \textsc{(Max-Min) Fair Matching} has also applications outside of the ``fairness'' context:
Imagine a centralized job market for companies (vertices in $V$) and applicants (vertices in $U$), each having a specific skillset (color). 
Firms may wish to balance between applicants with different skillsets so that they are able to place employees with various skillsets in each team (to be formed).
For instance, it may be desirable for a software company to hire roughly the same number of frontend and backend developers.

\subparagraph*{Our Contributions.}
We perform a refined complexity analysis of the NP-hard \textsc{Max-Min Fair Matching} and \textsc{MoV Fair Matching} problems in terms of the size $k$ of $V$, the number $|C|$ of colors, and the maximum degree $\Delta_U$ and $\Delta_V$ among vertices in $U$ and $V$, respectively.
Our main contribution are arguably involved FPT algorithms for the parameterization $k$ (\Cref{sec:k}).
At the heart of the design of our algorithms lies an integer linear program (ILP) of bounded dimension.
We essentially determine whether Hall-like conditions that guarantee the existence of a fair matching are fulfilled by formulating these in a system of linear inequalities.
In order to establish the correctness of our ILP formulations, we prove what we call \emph{touching separation theorem}, getting inspiration from Frank's separation theorem on submodular and supermodular functions \cite{Fra82}.
For \textsc{MoV Fair Matching}, we apply our approach in conjunction with the color coding technique \cite{DBLP:journals/jacm/AlonYZ95}. 
To familiarize ourselves with the ideas underlying our ILPs, in \Cref{sec:kc}, we start with a warm-up where we present ILP-based fixed-parameter tractable algorithms for the larger parameter $k+|C|$. 
To sum up, as it is straightforward to see that \textsc{Max-Min/Mov Fair Matching} are FPT with respect to the size $n$ of $U$\footnote{We can enumerate all fair subsets of $U$ in $2^n \cdot (n + k)^{O(1)}$ time. Then, Knapsack-like dynamic programming solves \textsc{Fair Matching} in $3^n \cdot (n + k)^{O(1)}$ time.}, we establish the fixed-parameter tractability of both problems for the two natural parameters $n$ and $k$.\footnote{In most described applications $k$ is typically quite small and much smaller than $n$. For instance, Stoica et al.~\cite{DBLP:conf/atal/StoicaCDG20} performed some experiments  for \textsc{MoV Fair Matching} to assign voters to districts with $n=50,000$ and $k=10$, and to assign students to schools with $n=41,834$ and~$k=61$.}
We then in \Cref{sec:dic} study the computational complexity of \textsc{Max-Min/Mov Fair Matching} with respect to $\Delta_U$, $\Delta_V$, and $|C|$.
We show that \textsc{Max-Min/Mov Fair Matching} is polynomial-time solvable for $|C| = 2$ and that it becomes NP-hard for $|C| \ge 3$.
Moreover, we settle all questions concerning the problems' classical complexity in terms of $\Delta_U$ and $\Delta_V$, revealing a complete complexity landscape in this regard (see \Cref{fig:dichotomies}).
Finally, in \Cref{se:cliques}, we show that \textsc{Max-Min/Mov Fair Matching} are linear-time solvable  when every vertex in $U$ can be matched to any vertex in $V$.

Notably, all our algorithmic results hold even if we require that each vertex from $V$ is matched to at least one vertex from $U$.
This further constraint may appear when we need to divide the vertices into exactly $k$ non-empty fair subsets. 
Although this constraint is seemingly simple, sometimes (e.g., in our FPT algorithm for \textsc{Max-Min Fair Matching} for $k$) non-trivial adaptions are needed.

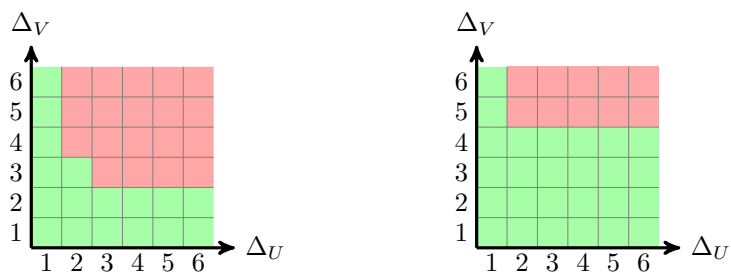
\begin{figure}[t]
	\centering
	\begin{tikzpicture}[scale=0.4]

		\pgfmathsetmacro{\xmax}{6}
		\pgfmathsetmacro{\xmaxd}{\xmax - 1}
		\pgfmathsetmacro{\xmaxdd}{\xmax - 2}
		\pgfmathsetmacro{\xmaxddd}{\xmax - 3}
		\pgfmathsetmacro{\ymax}{6}
		\pgfmathsetmacro{\ymaxd}{\ymax - 1}

		\foreach \i in { 1, ..., \xmax } \node at (\i - .5, -.5) {\i};
		\foreach \i in { 1, ..., \ymax } \node at (-.5, \i - .5) {\i};

		\fill[red!35] (1, 3) -- (2, 3) -- (2, 2) -- (\xmax, 2) -- (\xmax, \ymax) -- (1, \ymax) -- (1, 3);

		\tikzstyle{greenfill}=[green!35]
		\fill[greenfill] (0, 0) -- (0, \ymax) -- (1, \ymax) -- (1, 3) -- (2, 3) -- (2, 2) -- (\xmax, 2) -- (\xmax, 0) -- (0, 0);

		\foreach \i in { 1, ..., \xmax } {
		\draw[help lines] (\i - 1, 0) -- (\i - 1, \ymax);
		}
		\foreach \i in { 1, ..., \ymax } {
		\draw[help lines] (0, \i - 1) -- (\xmax, \i - 1);
		}
		\tikzstyle{axis}=[->, very thick, >=stealth']
		\draw[axis] (0, 0) -- (\xmax + .7, 0) node [right] {$\Delta_U$};
		\draw[axis] (0, 0) -- (0, \ymax + .7) node [above] {$\Delta_V$};
	\end{tikzpicture}
	\hspace{5em}
	\begin{tikzpicture}[scale=0.4]

		\pgfmathsetmacro{\xmax}{6}
		\pgfmathsetmacro{\xmaxd}{\xmax - 1}
		\pgfmathsetmacro{\xmaxdd}{\xmax - 2}
		\pgfmathsetmacro{\xmaxddd}{\xmax - 3}
		\pgfmathsetmacro{\ymax}{6}
		\pgfmathsetmacro{\ymaxd}{\ymax - 1}

		\foreach \i in { 1, ..., \xmax } \node at (\i - .5, -.5) {\i};
		\foreach \i in { 1, ..., \ymax } \node at (-.5, \i - .5) {\i};

		\fill[red!35] (1, 4) -- (1, \ymax) -- (\xmax, \ymax) -- (\xmax, 4) -- (1, 4);

		\tikzstyle{greenfill}=[green!35]
		\fill[greenfill] (0, 0) -- (0, \ymax) -- (1, \ymax) -- (1, 4) -- (\xmax, 4) -- (\xmax, 0) -- (0, 0);

		\foreach \i in { 1, ..., \xmax } {
		\draw[help lines] (\i - 1, 0) -- (\i - 1, \ymax);
		}
		\foreach \i in { 1, ..., \ymax } {
		\draw[help lines] (0, \i - 1) -- (\xmax, \i - 1);
		}
		\tikzstyle{axis}=[->, very thick, >=stealth']
		\draw[axis] (0, 0) -- (\xmax + .7, 0) node [right] {$\Delta_U$};
		\draw[axis] (0, 0) -- (0, \ymax + .7) node [above] {$\Delta_V$};
	\end{tikzpicture}
	\caption{The complexity landscape of \textsc{Max-Min Fair Matching} (left) and \textsc{MoV Fair Matching} (right) for the maximum degree $\Delta_U$ (resp., $\Delta_V$) in $U$ (resp., $V$). Green denotes polynomial-time solvability and red denotes NP-hardness. All NP-hardness results hold already for only three~colors.}\label{fig:dichotomies}
\end{figure}

\subparagraph*{Related Work.}
	Stoica et al.~\cite{DBLP:conf/atal/StoicaCDG20} introduced three problems where the task is to partition a set of colored vertices into subsets with a small margin of victory satisfying some global size constraints. 
	Among these three, the most general is \textsc{Fair Connected Regrouping}, where one is given a vertex-colored graph $G$, an integer $k$, and a function that determines for each vertex which subsets it can be part of. 
	The task is then to find a partitioning of $G$ into $k$ fair districts (i.e., connected components).
	In a follow-up work, Boehmer et al.~\cite{DBLP:journals/corr/abs-2102-11864} analyzed how the structure of $G$ influences the (parameterized) computational complexity of \textsc{Fair Connected Regrouping}.
	The other two problems Stoica et al.~\cite{DBLP:conf/atal/StoicaCDG20} considered are special cases of \textsc{Fair Connected Regrouping}:
	One is \textsc{Fair Regrouping} where the connectivity constraints are dropped (corresponding to \textsc{MoV Fair Matching}).
	The other is \textsc{Fair Regrouping\_X}, which is a special case of \textsc{Fair Regrouping}, where any vertex can belong to any district (corresponding to \textsc{MoV Fair Matching} on complete bipartite graphs; we study this special case in \Cref{se:cliques}). 
	They proved that \textsc{Fair Regrouping} is NP-hard for three colors (without any constraints on the degree of the graph) and is XP with respect to the number $k$ of districts (i.e., polynomial-time solvable for constant $k$).
	They also showed that \textsc{Fair Regrouping\_X} is XP with respect to the number of colors.

	Coming back to \textsc{Fair Matching} as a matching problem, Ahmed et al.~\cite{DBLP:conf/ijcai/AhmedDF17} proposed a global supermodular objective to model the fairness (which they call diversity)
	of a bipartite weighted many-to-many matching and developed a polynomial-time greedy heuristic for it.
	Ahmadi et al.~\cite{DBLP:conf/ijcai/AhmadiADFK20} extended the work of Ahmed et al.~\cite{DBLP:conf/ijcai/AhmedDF17} by generalizing the problem to the case where vertices can have multiple different colors and presented a pseudo-polynomial-time algorithm for it. 
	Moreover, Dickerson et al.~\cite{DBLP:conf/aaai/DickersonSSX19} applied this formulation of fairness to an online setting where vertices from the left side arrive over time and Ahmed et al.~\cite{ahmed2020forming} applied it to the task of forming~teams.

	Fairness is also a popular topic when finding a stable many-to-one matching of vertices that have preferences over each other. Here, fairness constraints are typically modeled by imposing for each vertex from the right side certain lower and upper bounds on the number of vertices of each color that can be matched to it~\cite{DBLP:conf/atal/0001GSW19,bo2016fair,DBLP:conf/ijcai/ChenGH20,hafalir2013effective,DBLP:conf/soda/Huang10}.  
 
	More broadly speaking, fairness has recently also been frequently applied to a variety of different problems from the area of combinatorial optimization. 
	For instance, in the context of the \textsc{Knapsack} \cite{DBLP:conf/atal/Patel0L21}, \textsc{Maximum Coverage} \cite{DBLP:journals/corr/abs-2007-08069} or \textsc{$s$-$t$ Path} \cite{bentert2022finding} problem, fairness means that all types are represented equally in the selected solution. 
	Another example is \textsc{Clustering}, where each cluster is considered fair when each type accounts for a certain fraction of vertices in it \cite{DBLP:conf/aistats/AhmadianE0M20,DBLP:conf/wads/FriggstadM21a,froese2021modification}.
 
The proof (or their completion) of all results marked by (\appmark) can be found in the appendix. 

\section{Preliminaries}
For two integers $i<j\in \mathbb{N}$, let $[i,j]=\{i,i+1,\dots,j-1,j\}$ and let $[i] = [1,i]$.
For a set $S$ and an element $x \in S$, we sometimes write $S - x$ to denote $S \setminus \{ x \}$.

Let $G = (U \cupdot V, E)$ be a bipartite graph, where $U$ is the \emph{left side} and $V$ is the \emph{right side} of $G$.
Let $n := |U|$ and $k := |V|$ be the number of vertices in the left side and right side, respectively.
For a vertex $w \in U \cupdot V$ and an edge set $M \subseteq E$, let $M(w)$ be the set of vertices \emph{matched to} $w$ in $M$, i.e., $M(w) = \{ w' \in U \cupdot V \mid \{ w, w' \} \in M \}$.
We say that $M \subseteq E$ is a \emph{many-to-one matching} in $G$ if $|M(u)| \le 1$ for every $u \in U$.
A many-to-one matching $M$ is \emph{left-perfect} if $|M(u)| = 1$ for every $u \in U$.
Note that we require $M$ to be left-perfect as otherwise an empty set would constitute a trivial solution for our problem.

When clear from context, we refer to a left-perfect many-to-one matching as a matching.
For a vertex $w \in U \cupdot V$, let $N_G(w)$ be the set of its neighbors in $G$, i.e., $N_G(w) = \{ w' \in U \cupdot V \mid \{ w, w' \} \in E \}$.
For $W \subseteq U \cupdot V$, let $N_G(W) = \bigcup_{w \in W} N_G(w)$ be the joint neighborhood of vertices from $W$ and let $\nu_G(W) = \{ w' \in U \cupdot V \mid N_G(w') \subseteq W \}$ be the set of vertices which are only adjacent to vertices in $W$.
We drop the subscript $\cdot_G$ when it is clear from context.

Let $C$ be the set of colors and let $\col \colon U \to C$ be a function that assigns a color to every vertex of $U$.
For $U' \subseteq U$, let $U_c' \subseteq U'$ be the set of vertices $u \in U'$ of color $c$.
For instance, given a matching $M$ and a vertex $v \in V$, $M(v)_c$ denotes the set of vertices matched to $v$ in $M$ that have color $c$.
We denote by $G_c = G[U_c \cup V]$ the graph $G$ restricted to vertices from $U_c\cup V$.
We use the shorthand $N_c(W)$ (resp., $\nu_c(W)$) for $N_{G_c}(W)$ (resp., $\nu_{G_c}(W)$).

Throughout the paper, we assume that the set $C$ of colors is equipped with some linear order $\le_C$, which serves as a tie breaker.
So $\arg \max$ over $C$ is well-defined.
We write $\max^1$ for $\max$ and $\max^2$ for the second largest element.

We now define our two fairness measures. 
For a subset of vertices $U'\subseteq U$, let $\MOV(U'):=\max^1_{c\in C} |U'_c|-\max^2_{c\in C} |U'_c|$ be the difference between the number of occurrences of the most and second most frequent color in $U'$.
Similarly, for a subset of vertices $U'\subseteq U$, let $\MaM(U'):=\max_{c\in C} |U'_c|-\min_{c\in C} |U'_c|$ be the difference between the number of occurrences of the most and least frequent color in $U'$.
A subset of vertices $U'\subseteq U$ is $\ell$-fair according to $\textsc{MoV}$ (resp., \textsc{Max-Min}) if $\MOV(U')\leq \ell$ (resp., $\MaM(U')\leq \ell$).\footnote{Notably, the definition of margin of victory of Stoica et al.~\cite{DBLP:conf/atal/StoicaCDG20} differs slightly from ours in that in their definition sets of vertices where the two most frequent colors have the same number of occurrences have a margin of victory of one (and not of zero). We chose our definition in accordance with Boehmer et al.~\cite{DBLP:journals/corr/abs-2102-11864} to be able to distinguish a tie between two colors from one color being one vertex ahead of another.} 
A many-to-one matching $M$ in $G$ is $\ell$-fair according to \textsc{MoV} (resp., \textsc{Max-Min}) if $M(v)$ is $\ell$-fair according to \textsc{MoV} (resp., \textsc{Max-Min}) for all $v\in V$. 
The considered fairness notion will always be clear from context.
We now define our central problem \textsc{$\Pi$ Fair Matching} for some fairness measure $\Pi$:
\problem{$\Pi$ Fair Matching}
{A bipartite graph $G = (U \cupdot V, E)$, a set $C$ of colors, a function $\col \colon U \to C$, and an integer $\ell \in \mathbb{N}$.}
{Is there a left-perfect many-to-one matching $M \subseteq E$ which is $\ell$-fair according to the fairness measure~$\Pi$?} 
We also sometimes consider \textsc{$\Pi$ Fair Matching} with size constraints where additionally given two integers $\smin$ and $\smax$, we require that the matching $M$ to be found satisfies $\smin \leq |M(v)|\leq \smax$ for all $v\in V$. 
We refer to the case with $\smin=1$ and $\smax=n$ as the non-emptiness constraint.
The non-emptiness constraint is arguably crucial for some applications when we want to partition the vertices in the left side into exactly $k$ non-empty subsets.

Let $\mathcal{I}$ be an instance of some problem and let $\mathcal{P}(\mathcal{I})$ be an integer linear program (ILP) constructed from $\mathcal{I}$.
We say that $\mathcal{P}$ is \emph{complete} if $\mathcal{P}(\mathcal{I})$ is feasible whenever $\mathcal{I}$ is a yes-instance.
Conversely, we say that $\mathcal{P}$ is \emph{sound} if $\mathcal{I}$ is a yes-instance whenever $\mathcal{P}(\mathcal{I})$ is feasible.
In this work, we will make use of Lenstra's algorithm \cite{Kan87,Len83} that decides whether an ILP of size $L$ with $p$ variables is feasible using $O(p^{2.5p + o(p)} \cdot |L|)$ arithmetic operations.

We assume that the reader is familiar with basic concepts in parameterized complexity (see for instance \cite{DBLP:books/sp/CyganFKLMPPS15}).
As a reminder, an FPT algorithm for a parameter $k$ is an algorithm whose running time on input $\mathcal{I}$ is $f(k) \cdot |\mathcal{I}|^{O(1)}$ for some computable function $f$.

\section{Warmup: FPT Algorithms for $k + |C|$}
\label{sec:kc}
We prove that both \textsc{Fair Matching} problems are fixed-parameter tractable with respect to $k+|C|$:
We present an integer linear programming (ILP) formulation of these problems whose number of variables is bounded in a function of the parameter $k+|C|$ and subsequently employ Lenstra's algorithm \cite{Kan87,Len83}.
Notably, one can upper-bound the number of ``types'' (according to their neighborhoods and colors) of vertices in $U$ by $2^k \cdot |C|$. 
Using this observation, it is straightforward to give an ILP formulation of \textsc{Max-Min/MoV Fair Matching} using $O(2^k\cdot |C|)$ variables. 
Instead, we follow a theoretically more involved but more efficient approach. 
For this, we use a structural property of our problem related to Hall's theorem, which decreases the number of variables in our ILP to $O(|C| \cdot k)$. 
We reuse some of the results from this section in \Cref{sec:k}, where we prove that \textsc{Max-Min/MoV Fair Matching} are actually fixed-parameter tractable with respect to $k$.  

To prove that the ILP we present in the following is complete, we use the following: 

\begin{restatable}{lemma}{kfor}
\label{lemma:k:forward}
	Let $G = (U \cupdot V, E)$ be a bipartite graph and let $M$ be a left-perfect many-to-one  matching.
	Then, $|\nu(W)| \le \sum_{v \in W} |M(v)| \le |N(W)|$ for every $W \subseteq V$.
\end{restatable}
\begin{proof}
	For every vertex $u \in \nu(W)$, we have $N(u) \subseteq W$ by definition.
	Since $u$ is matched to some vertex $v \in W$ in $M$ (that is, $u \in M(v)$), the first inequality holds.
	For the second inequality, observe that every vertex $u \in M(v)$ for $v \in W$ satisfies $u \in N(W)$.
\end{proof}

For proving that our ILP is sound, we use the following slightly more complicated Hall-like lemma:

\begin{restatable}{lemma}{kback}
	\label{lemma:k:backward}
	Let $G = (U \cupdot V, E)$ be a bipartite graph and let $\{ z_v \in \mathbb{N} \mid v \in V\}$ be a set of integers.
	Suppose that $\sum_{v \in W} z_v \ge |\nu_G(W)|$ for every $W \subseteq V$ and that $\sum_{v \in V} z_v = |U|$.
	Then, there is a left-perfect many-to-one matching $M$ such that $|M(v)| = z_v$ for every $v \in V$.
\end{restatable}
	\begin{proof}
    Assume that the conditions stated in the lemma hold.
	We prove the existence of such a matching $M$ by making use of Hall's theorem \cite{hall1987representatives}.
	To do so, we introduce an auxiliary bipartite graph $G'$ as follows:
	In $G'$, the vertices on one of the two sides are the vertices from $U$.
	The vertices on the other side are $V' := \bigcup_{v \in V} Z_v$, where $Z_v$ is a set of $z_v$ vertices.
	There is an edge between $u \in U$ and $v' \in Z_v \subseteq V'$ if and only if $\{ u, v \} \in E$.
	In order to apply Hall's theorem, we show that $|U'| \le |N_{G'}(U')|$ for every $U' \subseteq U$.

	Fix some $U'\subseteq U$ and let $W' = N_G(U')$.
	The construction of $G'$ gives us $|N_{G'}(U')| = \sum_{v \in W'} z_v$.
	By the assumption of the lemma, we have $\sum_{v \in W'} z_v \ge |\nu_G(W')|$.
	Moreover, we have $\nu_G(W') = \nu_G(N_G(U')) \supseteq U'$.
	Putting this together, we obtain $|N_{G'}(U')| = \sum_{v \in W'} z_v \ge |\nu_G(W')| \ge |U'|$.
	Hall's theorem then implies that $G'$ admits a one-to-one matching $M'$ which matches all vertices from $U$.
	In fact, $M'$ is a perfect matching since $|V'| = \sum_{v \in V} z_v = |U|$ by our assumption.
	Now consider the matching $M$ in $G$ where a vertex $u \in U$ is matched to $v\in V$ if  $u$ is matched to a vertex from $Z_v$ in $M'$.
	Then, $M$ is a left-perfect many-to-one matching with $|M(v)| = z_v$ for every~$v \in V$.
\end{proof}

Using \Cref{lemma:k:forward,lemma:k:backward}, we  give an ILP formulation of \textsc{Fair Matching} with $O(|C|\cdot k)$~variables. 

\subparagraph*{ILP formulation.}

Introduce a variable $z_v^c \in \mathbb{N}$ for every $v \in V$ and every $c \in C$.
The variable $z_v^c$ represents the number of vertices in $U$ of color $c$ that are matched to $v$.
Suppose that the given instance admits a left-perfect $\ell$-fair many-to-one matching $M$ respecting the values of $z_v^c$.
Then, for each $c\in C$, there is a matching $M_c$ in $G_c$ with $|M_c(v)| = z_v^c$ for $v\in V$ and $|M_c(u)|=1$ for $u\in U_c$.
As shown in \Cref{lemma:k:forward}, from this one can conclude that the following constraints must be fulfilled:
\begin{align}
 |\nu_c(W)| \le \sum_{v \in W} z_v^c \le |N_c(W)| \text{ for all } W \subseteq V, c \in C. \label{eq:ILPkC:c1}
\end{align}
Next, we encode the fairness requirement.
For \textsc{Max-Min}, we need to have that for every pair of colors the number of vertices of these two colors assigned to some vertex $v\in V$ differ by at most $\ell$.
Thus, we add the constraint 
\begin{align*}
 z_v^{c'} - z_v^c \le \ell \text{ for all } v \in V  \text{ and } c, c' \in C.
\end{align*}
To model $\ell$-fairness for \textsc{MoV Fair Matching}, we introduce two new binary variables $a_v^c, b_v^c \in \{ 0, 1 \}$ for all $v\in V$ and $c\in C$.
Informally speaking, the intended meaning of these variables is that $a_v^c = 1$ (resp., $b_v^c = 1$) when $c$ is the most (resp., second most) frequent color among vertices matched to $v$.
We ensure that the values of these variables are set accordingly as follows:
{\small
\begin{align*}
	&z_v^c - z_v^{c'} \ge n(a_v^c + b_v^{c'} - 2), \quad
	z_v^c - z_v^{c'} \ge n(b_v^{c} - a_v^{c'} - 1), \text{ and }
	a_v^c + b_v^c \le 1
	\quad \forall v \in V, c, c' \in C; \\
	&\sum_{c \in C} a_v^c = \sum_{c \in C} b_v^c = 1 \quad \forall v \in V.
\end{align*}}
For the first constraint, note that it becomes $z_v^c \ge z_v^{c'}$ if $a_v^c = b_v^{c'} = 1$ and that it is always fulfilled otherwise.
Similarly for the second constraint, observe that it becomes $z_v^c \ge z_v^{c'}$ if $a_v^{c'} = 0$ and $b_v^c = 1$ and that it is always fulfilled otherwise.

Finally, using $a_v^c$ and $b_v^c$ (and their meaning as proven above), we add the following constraint that encodes the $\ell$-fairness in terms of margin of victory of $M(v)$ for all $v\in V$:
\begin{align*}
	z_v^c - z_v^{c'} - n(2 - a_v^c - b_v^{c'}) \le \ell \qquad \forall v\in V.
\end{align*}

Lastly, we can also add linear constraints ensuring that the number of vertices from $U$ matched to each vertex $v\in V$ is between $\smin$ and $\smax$: $\smin \leq \sum_{c\in C} z_v^c \leq \smax$ for all $v \in V.$

\begin{theorem} \label{th:fptkc}
	\textsc{Max-Min/MoV Fair Matching} with arbitrary size constraints can be solved in $O^{\star}((|C| \cdot k)^{O(|C| \cdot k)})$ time.
\end{theorem}
\begin{proof}
	We show that an instance $(G=(U\cupdot V,E),C,\col,\ell,\smin,\smax)$ of \textsc{Max-Min/MoV Fair Matching} with size constraints admits an $\ell$-fair left-perfect many-to-one matching if and only if the constructed ILP constructed is feasible.
	As described above, by \Cref{lemma:k:forward}, if the given instance is a yes-instance, there is an assignment to $z_v^c$ satisfying all the constraints.
	Conversely, suppose that the ILP admits a solution $\{ z_v^c \mid v \in V, c \in C \}$.
	Then, from Constraints \ref{eq:ILPkC:c1}, we have for every $c \in C$ that $\sum_{v \in W} z_v^c \ge |\nu_{c}(W)|$ for every $W \subseteq V$ and $\sum_{v \in V} z_v^c = |U_c|$ (the later part follows from our first set of constraints for $W=V$).
	By \Cref{lemma:k:backward}, it follows that, for each color $c\in C$, $G_c$ has a matching $M_c$ in which every vertex in $U_c$ is matched and the values of $z_v^c$ for $v \in V$, are respected.
	Aggregating $M_c$ for all colors $c$ yields a left-perfect $\ell$-fair many-to-one matching for $G$ respecting the given size constraints.
	Using Lenstra's algorithm~\cite{Kan87,Len83}, the feasibility of the ILP can be determined in the claimed~time.
\end{proof}

\section{FPT Algorithms for $k$}
\label{sec:k}
In this section, we develop FPT algorithms for \textsc{Max-Min/MoV Fair Matching} for the parameterization $|V|=k$. We start with a discussion of the challenges for our algorithms. Afterwards, we obtain a new result on submodular and supermodular functions. Using this, in \Cref{sec:k:mmm} we present the algorithm for \textsc{Max-Min Fair Matching}, and in \Cref{sec:k:mov} the slightly more involved algorithm for \textsc{MoV Fair Matching}.

The crux of our algorithms is an ILP as in \Cref{sec:kc}.
However, since we look into the parameterization without $|C|$, it would be too costly to introduce variables for each color.
To illustrate our idea to work around this issue, take \textsc{Max-Min Fair Matching} as an example.
One of the straightforward ideas how to formulate this problem as an ILP with $\mathcal{O}(k)$ variables would be to introduce two variables $x_v \le y_v$ for every vertex $v\in V$, where $x_v$ (resp., $y_v$) encodes the minimum (resp., maximum) number of vertices of some color $c\in C$ matched to $v$. 
Informally speaking, to encode the ILP constraints from \Cref{sec:kc}, we could now replace every occurrence of $z_v^c$ with $x_v$ (resp., $y_v$) if the constraint in which $z_v^c$ occurs imposes an upper (resp., lower) bound on $z_v^c$: Constraints \ref{eq:ILPkC:c1} from \Cref{sec:kc} then translate to $\sum_{v \in W} y_v \ge \max_{c \in C} |\nu_c(W)| \text{ and } \sum_{v \in W} x_v \le \min_{c \in C} |N_c(W)|$ for every $W \subseteq V$.
Moreover, we add the constraint $y_v-x_v\leq \ell$ for all $v\in V$.
Let $\mathcal P$ denote the ILP obtained this way.
(See \Cref{sec:k:mmm} for the formal construction of $\mathcal P$.)
Although it is easy to see that $\mathcal P$ is complete, it turns out to be nontrivial to prove its soundness.
To highlight this challenge, suppose that $\mathcal P$ is feasible for $\{ x_v, y_v \mid v \in V \}$.
To show that this implies that there is an $\ell$-fair matching $M$, we have to show that for every color $c\in C$, the graph $G_c$ has a matching $M_c$ such that $x_v \le |M_c(v)| \le y_v$ for all $v\in V$.
Unfortunately, we cannot directly apply \Cref{lemma:k:backward} (as we were able to do in \Cref{sec:kc}, since we do not know the exact number of vertices of each color matched to a vertex from $V$ but only upper and lower bounds):
To construct an $\ell$-fair matching from $\{ x_v, y_v \mid v \in V \}$ using  \Cref{lemma:k:backward}, we have to show that for every color $c \in C$, there always exists a set of integers $\{z_v^c\mid v\in V\}$ with the following properties: 

\begin{description}
 \item[(i)] $\sum_{v \in W} z_v^c \ge |\nu_c(W)|,$ $\forall W \subseteq V$
\item[(ii)] $\sum_{v \in V} z_v^c = |U_c|$
\item[(iii)] $x_v \le z_v^c \le y_v, $ $\forall v \in V$.
\end{description}

\noindent Finding an assignment of variables $z_v^c$ that fulfill (i) and (iii) is trivial; setting
$z_v^c = y_v$ suffices because $\mathcal P$ dictates that $\sum_{v \in W} y_v \ge \max_{c' \in C} |\nu_{c'}(W)| \ge |\nu_c(W)|$ for every $W \subseteq V$.
However, integers satisfying (i), (ii), and (iii) simultaneously are not trivially guaranteed to exist.
To nevertheless prove their existence, we prove what we call the \emph{touching separation theorem}, inspired by Frank's separation theorem \cite{Fra82} on submodular and supermodular functions.
Our theorem implies that if there is a solution to $\mathcal{P}$, then there is an assignment of variables $z_v^c$ satisfying (i), (ii), and (iii).
Submodular and supermodular functions are defined as~follows:

\begin{definition}
	\label{def:k:mod}
	Let $f \colon 2^S \to \mathbb{N}$ be a set function over a set $S$.
	We say that $f$ is submodular if $f(X) + f(Y) \ge f(X \cup Y) + f(X \cap Y)$ for every $X, Y \subseteq S$ and supermodular if $f(X) + f(Y) \le f(X \cup Y) + f(X \cap Y)$ for every $X, Y \subseteq S$.
	Moreover, $f$ is modular if $f$ is both submodular and supermodular.
\end{definition}

Modular functions admit a simpler characterization, which will be useful in several proofs:

\begin{restatable}[folklore, \appmark]{lemma}{mod}
	\label{lemma:mod}
	Let $f \colon 2^S \to \mathbb{N}$ be a set function.
	Then, $f$ is modular if and only if $f(X) = f(\emptyset) + \sum_{x \in X} (f(x) - f(\emptyset))$ for every $X \subseteq S$.
\end{restatable}
\toappendix{
\section{Proof of Folklore \Cref{lemma:mod}}
\mod*
\begin{proof}
	Suppose that $f(X) = f(\emptyset) + \sum_{x \in X} (f(x) - f(\emptyset))$ for every $X \subseteq S$.
	Observe that for every $X, Y \subseteq S$,
	\begin{align*}
		f(X \cup Y) + f(X \cap Y)
		&= 2 f(\emptyset) + \sum_{x \in X \cup Y} (f(x) - f(\emptyset)) + \sum_{x \in X \cap Y} (f(x) - f(\emptyset)) \\
		&= 2 f(\emptyset) + \sum_{x \in X} (f(x) - f(\emptyset)) + \sum_{x \in Y} (f(x) - f(\emptyset)) = f(X) + f(Y).
	\end{align*} 
	Thus, $f$ is modular according to \Cref{def:k:mod}.

	Conversely, we prove that a modular function satisfies  $f(X) = f(\emptyset) + \sum_{x \in X} (f(x) - f(\emptyset))$ for every $X \subseteq S$ by induction on $|X|$.
	For $X = \{ x \}$, we have $f(\emptyset) + \sum_{x' \in X} (f(x') - f(\emptyset)) = f(x)$.
	For $|X| > 1$, let $x \in X$ be an element of $X$.
	By the modularity of $f$, we have
	\begin{align*}
		f(X) &= f(X - x) + f(x) - f(\emptyset) \\
		&= f(x) + \sum_{x' \in X \setminus \{ x \}} (f(x') - f(\emptyset)) = f(\emptyset) + \sum_{x' \in X} (f(x) - f(\emptyset)).
	\end{align*}
	Here, the equality holds as $f$ is modular and the second equality is due to induction hypothesis.
\end{proof}}

Frank's famous separation theorem \cite{Fra82} can be stated as follows:
If $f \colon 2^S \to \mathbb{N}$ is submodular and $g \colon 2^S \to \mathbb{N}$ is supermodular, and $g(X) \le f(X)$ for every $X \subseteq S$, then there exists a modular set function $h \colon 2^S \to \mathbb{N}$ such that $g(X) \le h(X) \le f(X)$ for every $X \subseteq S$.
We prove an analogous separation theorem where the ``lower-bound'' function is the maximum of two functions---one being modular and the other being supermodular---and the ``upper-bound'' function is a modular function. (In general, the ``lower-bound'' function arising this way is neither submodular nor supermodular.)
We additionally require that the modular function separating the two bounds ``touches'' the lower-bound function.
Our proof uses a rather involved induction on the size of $S$.

\begin{restatable}[Touching separation theorem]{theorem}{tst}
	\label{lemma:k:separation}
	Let $f, f' \colon 2^S \to \mathbb{N}$ be two modular set functions and let $g \colon 2^S \to \mathbb{N}$ be a supermodular set function.
	Suppose that $f(\emptyset) = f'(\emptyset) = g(\emptyset) = 0$ and $\max(f(X), g(X)) \le f'(X)$ for every $X \subseteq S$.
	Then, there is a modular set function $h \colon 2^S \to \mathbb{N}$ such that $\max(f(X), g(X)) \le h(X) \le f'(X)$ for every $X \subseteq S$ and $h(S) = \max_{T \subseteq S} f(T) + g(S \setminus T)$.
\end{restatable}
\begin{proof}
	Our proof is based on the following claim which is a simple corollary of  Lemma 2.4 from McCormick \cite{surveymodular} (where we set $T=\emptyset$, $f=-\beta$, $z=-\alpha$, and $y=-\gamma$ in the statement of Lemma 2.4):
	\begin{claim} \label{claim:tst}
		Let $\alpha\colon 2^S \to \mathbb{N}$ be a modular function and $\beta\colon 2^S \to \mathbb{N}$ be a supermodular function where $\beta(X)\leq \alpha(X)$ for all $X\subseteq S$. 
		 Then, there is a modular function $\gamma \colon 2^S \to \mathbb{N}$ such that $\beta(X)\leq \gamma(X) \leq \alpha (X)$ for every $X\subseteq S$ and $\gamma(S)=\beta(S)$.
	\end{claim}
	Using \Cref{claim:tst}, we can now establish our theorem. 
	Let $g'\colon 2^S \to \mathbb{N}$ be the following set function $g'(X)=\max_{T\subseteq X} g(T)+f(X\setminus T)$ for all $X\subseteq S$.
	 
	We first prove that $g'$ is supermodular. 
	For some subsets $X,Y\subseteq S$, let $X'\in \argmax_{T\subseteq X} g(T)+f(X\setminus T)$ and $Y'\in \argmax_{T\subseteq X} g(T)+f(X\setminus T)$. 
	Using this, as $f$ is modular and $g$ is supermodular, we directly get that $g'$ is supermodular: 
	\begin{align*}
		g'(X) + g'(Y) &= g(X')+f(X\setminus X')+g(Y')+f(Y\setminus Y')\\
		& = g(X')+g(Y')+f(X)+f(Y)-f(X')-f(Y')\\
		& = g(X')+g(Y')+f(X\cup Y)+f(X\cap Y)-f(X'\cup Y')-f(X'\cap Y')\\
		& = g(X')+g(Y')+f((X\cup Y)\setminus (X'\cup Y'))+f((X\cap Y)\setminus(X'\cap Y'))\\
		& \leq g(X'\cup Y')+f((X\cup Y)\setminus (X'\cup Y'))+g(X'\cap Y')+f((X\cap Y)\setminus(X'\cap Y'))\\
		& \leq \max_{T\subseteq X\cup Y} g(T)+f(X\cup Y\setminus T) + \max_{T\subseteq X\cap Y} g(T)+f(X \cap Y\setminus T) \\
		& = g'(X\cup Y)+g'(X\cap Y), 
	\end{align*}
	We now turn to proving  that $g'(X)\leq f'(X)$ for every $X\subseteq S$. Fix some $X\subseteq S$ and let $X'\in \argmax_{T\subseteq X} g(T)+f(X\setminus T)$. Then, as $\max(f(X),g(X))\leq f'(X)$ and $f'$ is modular we get: 
	$$
	g'(X)=g(X')+f(X\setminus X')\leq f'(X')+f'(X\setminus X')=f'(X).$$
	Thus, applying  \Cref{claim:tst} (with $\alpha=f'$ and $\beta=g'$), we get a modular function $h(X)$ with $g'(X)\leq h(X)\leq f'(X)$ and $h(S)=g'(S)= \max_{T \subseteq S} f(T) + g(S \setminus T)$. 
	Note that we clearly have that $\max(f(X),g(X))\leq g'(X)$ for all $X\subseteq S$ and thus $\max(f(X),g(X))\leq g'(X)\leq h(X)\leq f'(X)$. 
	Thereby, the theorem follows. 
\end{proof}

To apply \Cref{lemma:k:separation}, we use a supermodular function arising from a bipartite~graph:

\begin{restatable}{lemma}{lemmanu}
	\label{lemma:nu}
	Let $G = (U \cupdot V, E)$ be a bipartite graph and let $g \colon 2^V \to \mathbb{N}$ be a set function such that $g(W) = |\nu(W)|$ for each $W \subseteq V$.
	Then, $g$ is supermodular.
\end{restatable}
\begin{proof}
	Let $W, W' \subseteq V$.
	Observe that
	\begin{align*}
	g(W \cup W') &= |\{ u \in U \mid N(u) \subseteq W \cup W' \}| \\
	&\ge |\{ u \in U \mid N(u) \subseteq W \}|  + |\{ u \in U \mid N(u) \subseteq W' \}|\\ & - |\{ u \in U \mid N(u) \subseteq W \cap W' \}| \\
	&= g(W) + g(W') - g(W \cap W'),
	\end{align*}
	which shows that $g$ is supermodular.
\end{proof}

The rest of this section is organized as follows. 
In \Cref{sec:k:mmm}, we consider \textsc{Max-Min Fair Matching}. 
We first construct an ILP with $\mathcal{O}(k)$ variables and establish its soundness with the help of the touching separation theorem. 
Subsequently, we extend our ILP to incorporate the non-emptiness constraint, which turns out to be surprisingly challenging. 
Afterwards, in \Cref{sec:k:mov}, we present our approach for \textsc{MoV Fair Matching} (the same for both the case with or without the non-emptiness constraint). 
Notably, we do not derive an ILP formulation for \textsc{MoV Fair Matching} directly, but instead introduce an auxiliary variant of the problem for which we construct an ILP with $\mathcal{O}(k^2)$ variables. 
As a last step, we show how \textsc{MoV Fair Matching} can be reduced to the auxiliary problem using the color coding technique together with some structural observations about fair matchings.

\subsection{Max-Min Fair Matching}
\label{sec:k:mmm} 

We now present an FPT algorithm for \textsc{Max-Min Fair Matching}.
As mentioned, our algorithm, which follows a similar approach as used in \Cref{sec:kc}, builds an ILP.
The difference is that the ILP constructed here involves $O(k)$ and not $\Theta(k\cdot |C|)$ variables.
We prove the correctness of the ILP in \Cref{th:max-min-k}, crucially relying on~\Cref{lemma:k:separation}.
\subparagraph*{ILP formulation.}
We add two variables $x_v \le y_v \in \mathbb{N}$ for every $v \in V$.
The intended meaning for these variables is that for every color $c \in C$, the number of vertices from $U_c$ matched to $v$ is between $x_v$ and $y_v$.
We obtain the following constraints from \Cref{lemma:k:forward}:
\begin{align} \label{eq:maxMinILP}
	\sum_{v \in W} y_v \ge \max_{c \in C} |\nu_c(W)| \text{ and } \sum_{v \in W} x_v \le \min_{c \in C} |N_c(W)| \quad \forall \, W \subseteq V.
\end{align}
To encode $\ell$-fairness, we add: $y_v - x_v \le \ell$ for all $v \in V$.

\begin{theorem}\label{th:max-min-k}
	\textsc{Max-Min Fair Matching} can be solved in $O^{\star}(k^{O(k)})$ time.
\end{theorem}
\begin{proof}
	Our ILP uses $2k$ variables and $O(2^k)$ constraints; the construction takes $O^{\star}(2^k)$ time.
	Using Lenstra's algorithm \cite{Kan87,Len83}, it is possible to check whether the constructed ILP is feasible in $O^{\star}(k^{O(k)})$ time.
	It remains to prove the correctness of our ILP.

	The completeness of our ILP follows from \Cref{lemma:k:forward}.
	For the soundness, suppose that $\{ x_v, y_v \mid v \in V \}$ is a feasible solution for the ILP.
	For every color $c\in C$, we show the following:
	For the set $U_c \subseteq U$ of vertices of color $c$, there is a left-perfect matching $M_c$ in the bipartite graph $G_c = G[U_c \cup V]$ such that $x_v \le |M_c(v)| \le y_v$ for each $v \in V$, from which the soundness of the ILP directly follows.
	To show the existence of such a matching $M_c$, we will rely on \Cref{lemma:k:backward}.
	Note, however, that \Cref{lemma:k:backward} asks for integers $\{ z_v \mid v \in V \}$ meeting certain constraints, which we cannot choose arbitrarily as we have to respect the constraint $x_v \le z_v \le y_v$ for every $v \in V$.
	Let us fix some color $c\in C$. We find integers $\{z_v \mid v\in V\}$ via~\Cref{lemma:k:separation}:

	Let $f, f', g \colon 2^V \to \mathbb{N}$ be set functions such that $f(W) = \sum_{v \in W} x_v$ and $f'(W) = \sum_{v \in W} y_v$, and $g(W) = |\nu_c(W)|$ for each $W \subseteq V$.
	Note that $f$ and $f'$ are modular by \Cref{lemma:mod} and that $g$ is supermodular by \Cref{lemma:nu}.
	The constraints in the ILP ensure that $\max(f(W), g(W)) \le f'(W)$ for every $W \subseteq V$.
	Consequently, \Cref{lemma:k:separation} yields a set modular function $h \colon 2^V \to \mathbb{N}$ such that $\max(f(W), g(W)) \le h(W) \le f'(W)$ for each $W \subseteq V$ and $h(V) = \max_{W \subseteq V} f(W) + g(V \setminus W)$.
	Let $z_v = h(v)$ for every $v \in V$.
	Note that as $h$ is modular and fulfills the constraints stated above, $\sum_{v \in W} z_v = h(W) \ge g(W) = |\nu_c(W)|$ for each $W\subseteq V$.
	Hence, in order to apply \Cref{lemma:k:backward} on the integers $\{z_v\mid v\in V\}$, it remains to show that $\sum_{v \in V} z_v = h(V) = |U_c|$.
	Since our ILP requires that $f(W) = \sum_{v \in W} x_v \le |N_c(W)|$ for $W \subseteq V$, we have	that
	\begin{align*}
	f(W) + g(V \setminus W)
	&\le |N_c(W)| + |\nu_c(V \setminus W)| \\
	&= |\{ u \in U_c \mid N_c(u) \cap W \ne \emptyset \}| + |\{ u \in U_c \mid N_c(u) \subseteq V \setminus W \}| = |U_c|.
	\end{align*}
	Moreover, we have $f(\emptyset) + g(V) = |U_c|$, resulting in $h(V) = \max_{W \subseteq V} f(W) + g(V \setminus W) = |U_c|$ by \Cref{lemma:k:separation}.
	Therefore, we have $\sum_{v\in V} z_v=|U_c|$ and by \Cref{lemma:k:backward} the graph $G_c$ admits a matching $M_c$ with $x_v \le |M_c(v)| \le y_v$ for each $v \in V$.
	Combining these matchings yields an $\ell$-fair matching~in~$G$.
\end{proof}

\subparagraph*{Non-emptiness constraint.}
For a matching $M$ found by our algorithm, we may have $M(v) = \emptyset$ for some $v \in V$. 
So the non-emptiness constraint may be violated.
Unfortunately, there is seemingly no simple linear constraint to ensure that $M(v) \ne \emptyset$ for each $v \in V$.\footnote{Adding $x_v>0$ ensures that there is at least one vertex of each color matched to $v$. Adding $y_v>0$ is  not enough, as the ILP only ensures that for each color \emph{at most} $y_v$ vertices are matched to $v$.} 
Overcoming this challenge, we now develop an FPT algorithm for \textsc{Max-Min Fair Matching} with the non-emptiness constraint.
We will build upon the ILP formulation for \Cref{th:max-min-k}.
However, the situation becomes substantially more complicated.
First, observe that for $\ell = 0$, it suffices to add the constraint $y_v > 0$ for all $v \in V$ because this ensures that every vertex $v \in V$ is matched to at least $x_v = y_v > 0$ vertices of each color.

For $\ell > 0$, we develop a more involved algorithm. 
We will start by describing its main ideas before presenting the formal proof.
Suppose that there is an $\ell$-fair many-to-one matching $M$ such that $M(v) \ne \emptyset$ for each $v \in V$.
By choosing an arbitrary element $u_v$ of $M(v)$ for each $v \in V$, we obtain a matching $M^{\star} := \{ \{ u_v, v \} \mid v \in V \} \subseteq M$ with $|M^{\star}(v)| = 1$ for each $v \in V$.
Since there are possibly $n^{\Omega(k)}$ choices for $M^{\star}$, we cannot assume that  $M^{\star}$ is given.
Instead, our algorithm only ``guesses'' some structural properties of $M^{\star}$, which we will incorporate into the above ILP for \textsc{Max-Min Fair Matching} by adding  further constraints. 
(Our algorithm avoids guessing objects marked with a star.)

For each $v \in V$, let $\chi(v)$ be the color of $M^{\star}(v)$.
By iterating over all possible partitions $\mathcal{V}$ of $V$, we first guess a partition of $V$ according to $\chi(v)$, i.e., two vertices $v, v' \in V$ belong to the same subset of $\mathcal{V}$ if and only if $\chi(v)=\chi(v')$.
For each $S\in \mathcal{V}$, let $\chi(S)$ be the color of $\chi(v)$ for all vertices $v \in S$.
Let us fix some $S\in \mathcal{V}$. 
We will formulate constraints that must be fulfilled when each vertex $v$ in $S$ is matched to its partner $M^{\star}(v)$ (which has color $\chi(S)$).
For this, let $U_S^{\star} \subseteq U_{\chi(S)}$ be the set of vertices from $U$ of color $\chi(S)$ incident to an edge in $M^{\star}$ and let $G_S^{\star}$ be the graph obtained from $G_{\chi(S)}$ by deleting all edges incident to $U_S^{\star}$ and then adding the edges of $M^{\star}$ whose endpoint on the left side has color $\chi(S)$.
Since every edge in $M$ with an endpoint of color $\chi(S)$ is present in $G_S^{\star}$, these edges form a left-perfect many-to-one matching in $G_S^{\star}$.
Thus, by \Cref{lemma:k:forward}, we should have $\sum_{v \in W} y_v \ge |\nu_{G_S^{\star}}(W)|$ and $\sum_{v \in W} x_v \le |N_{G_S^{\star}}(W)|$ for each $W \subseteq V$ if there is a left-perfect matching in $G_S^{\star}$.
We need to evaluate $|\nu_{G_S^{\star}}(W)|$ and $|N_{G_S^{\star}}(W)|$ to include these constraints into our ILP.
Note, however, that we cannot compute these values without $M^{\star}$ given.
In the following, we explain how we can nevertheless incorporate these constraints by guessing further structural aspects of $M^{\star}$.

First, we rewrite $|\nu_{G_S^{\star}}(W)|$ and $|N_{G_S^{\star}}(W)|$ as follows: 
\begin{align}
	|\nu_{G_S^{\star}}(W)| &= |\nu_{G_{\chi(S)}}(W)| + |\{ v \in S \cap W \mid N_{G_{\chi(S)}}(M^{\star}(v)) \not \subseteq W \}| \text{ and } \nonumber \\
	|N_{G_S^{\star}}(W)|
	&= |N_{G_{\chi(S)}}(W)| - |\{ v \in S \setminus W \mid N_{G_{\chi(S)}}(M^{\star}(v)) \cap W \ne \emptyset \}|. \label{eq:GSN}
\end{align}
To see why the first equation holds, observe that we have only deleted edges  when constructing $G_S^{\star}$ from $G_{\chi(S)}$. 
Thus, we have $\nu_{G_{\chi(S)}}(W) \subseteq \nu_{G_S^{\star}}(W)$.
Moreover, we have $u \in \nu_{G_S^{\star}}(W)\setminus \nu_{G_{\chi(S)}}(W)$ if and only if $u=M^{\star}(v)$ for some $v\in S\cap W$ (which implies that $u$ is only adjacent to $v$ in $G^{\star}_S$) and $u$ has a neighbor outside of $W$ in $G_{\chi(S)}$. 
The second equation follows similarly:
As we only deleted edges, we have $N_{G_S^{\star}}(W)\subseteq N_{G_{\chi(S)}}(W)$.
We also have $u\in N_{G_{\chi(S)}}(W)\setminus N_{G_S^{\star}}(W)$ if and only if $u=M^{\star}(v)$ for some $v\in S\setminus W$ and $u$ has a neighbor in $W$ in $G_{\chi(S)}$.
To evaluate the second term in each of these equations, we guess a function $\mu \colon V \to 2^V$ such that $\mu(v) = N_{G}(M^*(v))$ for each $v \in V$.
For $\alpha(S, W) := |\{ v \in S \cap W \mid \mu(v) \not \subseteq W \}|$ and $\beta(S, W) := |\{ v \in S \setminus W \mid \mu(v) \cap W \ne \emptyset \}|$, we then have from \Cref{eq:GSN}: 
\begin{align}
	|\nu_{G_S^{\star}}(W)|
	= |\nu_{\chi(S)}(W)| + \alpha(S, W) \text{ and }
	|N_{G_S^{\star}}(W)|
	= |N_{\chi(S)}(W)| - \beta(S, W).\label{eq:obs}
\end{align}

To incorporate the constraints $\sum_{v \in W} y_v \ge |\nu_{G_S^{\star}}(W)|$ and $\sum_{v \in W} x_v \le |N_{G_S^{\star}}(W)|$, it remains to deal with $|\nu_{\chi(S)}(W)|$ and $|N_{\chi(S)}(W)|$.
Note that we cannot compute $|\nu_{\chi(S)}(W)|$ or $|N_{\chi(S)}(W)|$ without  $M^{\star}$ given (in particular, we do not even know $\chi(S)$).
Moreover, it would be costly to guess these values (which can be $\Omega(n)$) for every $S \in \mathcal{V}$ and $W \subseteq V$ or the values of $\chi(S)$ (which can be $\Omega(|C|)$) for every $S\in \mathcal V$.
Instead, we relate these two values to $\max_{c \in C} |\nu_c(W)|$ and $\min_{c \in C} |N_c(W)|$ by guessing their respective differences. 
Although these differences may be $\Omega(n)$, we discover that we can cap them at $k$:
When they are larger, then the arising constraints are already covered by Constraint (\ref{eq:maxMinILP}) from the original ILP. 
Formally, we guess $X(S, W) = \min(|N_{\chi(S)}(W)| - \min_{c \in C} |N_{c}(W)|, k)$ and $Y(S, W) = \min(\max_{c \in C} |\nu_{c}(W)| - |\nu_{\chi(S)}(W)|, k)$ for each $S \in \mathcal{V}$ and $W \subseteq V$ by iterating over all possible $X, Y \colon \mathcal{V} \times 2^V \to \{ 0, \dots, k \}$ (for each of them we have at most $(k + 1)^{k \cdot 2^k} \in O(2^{2^{O(k)}})$ choices).

To account for the constraints $\sum_{v \in W} y_v \ge |\nu_{G_S^{\star}}(W)|$ and $\sum_{v \in W} x_v \le |N_{G_S^{\star}}(W)|$, we now add the following constraint to the ILP for each $S\in \mathcal{V}$ and $W\subseteq V$:
 \begin{equation*}
 	\begin{aligned}
 	&\sum_{v \in W} y_v \ge \max_{c \in C} |\nu_c(W)| - Y(S, W) + \alpha(S, W) \text{ and } \label{eq:ne-constraint} \\
 	&\sum_{v \in W} x_v \le \min_{c \in C} |N_c(W)| + X(S, W) - \beta(S, W). 
 	\end{aligned}
 \end{equation*}
We later argue in the proof why these constraints need to be fulfilled.

Iterating over all described guesses and for each solving the constructed ILP, we get: 

\begin{restatable}{theorem}{maxminnonempty}\label{th:max-min-k-ne}
	{\normalfont \textsc{Max-Min Fair Matching}}  with the non-emptiness constraint can be solved in $O^*(2^{2^{O(k)}})$ time.
\end{restatable}

\begin{proof}
    We use the notation and intuition developed above and provide a formal description of the algorithm and its proof of correctness here. 
    For the sake of completeness, some of the arguments that were part of the high-level discussion presented above get repeated.
    
    For $\ell = 0$, the ILP previously developed for the \textsc{Max-Min Fair Matching} problem without non-emptiness constraint amended with the constraint $y_v > 0$ for all $v \in V$ solves the problem, as $y_v>0$ ensures that every vertex $v \in V$ is matched to at least $x_v = y_v > 0$ vertices of each color and at least one vertex from each color needs to be matched to each $v\in V$ in every matching satisfying the non-emptiness constraint in case $\ell=0$.

    For $\ell > 0$, we execute the algorithm described in the following once for all combinations of the following (and return yes as soon as one execution returns yes, and return no otherwise): 
    \begin{itemize}
     \item A function $\mu \colon V \to 2^V$ (there are $2^{O(k^2)}$  possibilities). 
     \item A partition $\mathcal{V}$ of $V$  (there are $k^{O(k)}$ possibilities).
     \item Two functions $X, Y \colon \mathcal{V} \times 2^V \to \{ 0, \dots, k \}$ (for each of them there are  $(k + 1)^{k \cdot 2^k} \in O(2^{2^{O(k)}})$ possibilities).
    \end{itemize}
    The intended meaning of these objects is as described above.s
    We first check whether our guesses are compatible with each other: 
    For $c \in C$ and $S \in \mathcal{V}$, we say that $c$ is \emph{compatible} with $S$~if
    \begin{itemize}
	\item 
	there is a matching $M_c'$ in $G_c$ with $|M_c'(v)|=1$ and $N(M_c'(v)) = \mu(v)$ for all $v \in S$, and
	\item
	for every $W \subseteq V$,  $X(S, W) = \min(|N_c(W)| - \min_{c' \in C} |N_{c'}(W)|, k)$ and $Y(S, W) = \min(\max_{c'} |\nu_{c'}(W)| - |\nu_c(W)|, k)$.
    \end{itemize}
    Note that for each subset $S \in \mathcal{V}$, the color $\chi(S)$ of vertices matched to $S$ in $M^{\star}$ should be compatible with $S$.
	Essentially, we substitute $\chi(S)$ with compatible colors because it would be computationally too expensive to guess $\chi(S)$ for all $S \in \mathcal{V}$.
    We check whether there exist different compatible colors for each $S\in \mathcal{V}$ via an auxiliary bipartite graph $H$ constructed as follows.
    In the auxiliary, each vertex from the left side corresponds to a color $c \in C$ and each vertex from the right side corresponds to a subset $S \in \mathcal{V}$.
    We have an edge between two vertices if and only if the corresponding color and subset are compatible. 
    We return no if $H$ does not admit a (one-to-one) matching that covers the right side; otherwise we proceed. 
   
    We formulate an ILP as follows. 
    For every $v \in V$, we again use two variables $x_v \le y_v \in \mathbb{N}$, which encodes the bounds on the number of vertices of each color matched to $v$.
    For the $\ell$-fairness, we add: $y_v - x_v \le \ell$ for all $v \in V$.
    Moreover, by \Cref{lemma:k:forward}, we have the following constraints:
    \begin{align}\label{eq:firstconstraints}
        \sum_{v \in W} y_v \ge \max_{c \in C} |\nu_c(W)| \text{ and } \sum_{v \in W} x_v \le \min_{c \in C} |N_c(W)| \quad \forall \, W \subseteq V.
    \end{align}
    Moreover, for each $S \in \mathcal{V}$ we add the following constraints to enforce the non-emptiness constraint:
\begin{equation} 
	\begin{aligned}
	&\sum_{v \in W} y_v \ge \max_{c \in C} |\nu_c(W)| - Y(S, W) + \alpha(S, W) \text{ and } \\
	&\sum_{v \in W} x_v \le \min_{c \in C} |N_c(W)| + X(S, W) - \beta(S, W) \quad \forall W \subseteq V. \label{eq:secondconstraints}
	\end{aligned}
\end{equation}
    (Recall that $\alpha(S, W) := |\{ v \in S \cap W \mid \mu(v) \not \subseteq W \}|$ and $\beta(S, W) := |\{ v \in S \setminus W \mid \mu(v) \cap W \ne \emptyset \}|$.)
    The resulting ILP has $O(k)$ variables and $O(2^k)$ constraints.
    We solve the ILP in $O^{\star}(k^{O(k)})$ time using Lenstra's algorithm \cite{Kan87,Len83} and return yes if it admits a feasible solution and no otherwise. 
    Doing this for all possible combinations specified in the beginning results in an overall running time of $O^{\star}(2^{2^{O(k)}})$.
    It remains to prove the completeness and soundness of the ILP. 
    
    \subparagraph{Completeness.}
    Assume that there is a left-perfect matching $M$ such that $M(v) \ne \emptyset$ for each $v \in V$.
   For each $v\in V$, let $u_v$ be an arbitrary element of $M(v)$.
   Moreover, we define $M^{\star}$ as $M^{\star} := \{ \{ u_v, v \} \mid v \in V \} \subseteq M$. 
   Note that it holds that $|M^{\star}(v)| = 1$ for each $v \in V$.
    For each $v\in V$, let $\chi(v)$ be the color of $M^{\star}(v)$. 
    
    We consider the execution of the algorithm for the following choices of our guesses: 
    \begin{itemize}
     \item For each $v\in V$, let $\mu(v):=N_G(M^{\star}(v))$. 
     \item Let $\mathcal{V}$ be a partition of $V$ according to $\chi(v)$, i.e., two vertices $v,v'\in V$ are part of the same set in $\mathcal{V}$ if and only if $\chi(v)=\chi(v')$. 
     \item For each $S\in \mathcal{V}$, let $\chi(S)$ be the color of $\chi(v)$ for all $v\in S$. For each $S\in \mathcal{V}$ and $W\subseteq V$, let $X(S, W) := \min(|N_{\chi(S)}(W)| - \min_{c \in C} |N_{c}(W)|, k)$ and $Y(S, W) := \min(\max_{c \in C} |\nu_{c}(W)| - |\nu_{\chi(S)}(W)|, k)$.
    \end{itemize}
    
    For the first part of the algorithm, let $M':=\{\{\chi(S),S\}\mid S\in \mathcal{V}\}$. 
    We claim that $M'$ is a one-to-one matching in $H$ covering the right side (and thus that the algorithm did not return no in the first step). 
    Note that by the construction of $\mu$, $X$ and $Y$, all edges in $M'$ are also part of $H$. 
    Moreover, by the construction of $\mathcal{V}$, it follows that $\chi(S)\neq \chi(S')$ for $S\neq S'\in \mathcal{V}$. 
    Thus, $M'$ is a one-to-one matching in $H$ covering the right side. 
    
    For each $v\in V$, let $x_v:=\min_{c\in C} |M(v)_c|$ and $y_v:=\max_{c\in C} |M(v)_c|$. 
    We claim that these variable assignments satisfy all constraints of the above constructed ILP. 
    As $M$ is $\ell$-fair, they clearly satisfy the fairness constraint. 
    Moreover as $M$ is a left-perfect matching by \Cref{lemma:k:forward} they also satisfy Constraints \ref{eq:firstconstraints}. 
    To prove that they also satisfy Constraints~\ref{eq:secondconstraints}, we use the notion introduced in the initial high-level discussion.
    Note that for each $S\in \mathcal{V}$, $M$ is also a left-perfect matching in $G^{\star}_S$.
    Thus, by \Cref{lemma:k:forward} for every $S\subseteq W$ it holds that 
    \begin{equation}\label{eq:intermediate}
     \sum_{v \in W} y_v \ge |\nu_{G_S^{\star}}(W)|\quad \text{ and } \quad \sum_{v \in W} x_v \le |N_{G_S^{\star}}(W)|.
    \end{equation}
    To justify that our described variable assignments satisfy Constraints (\ref{eq:secondconstraints}), let us fix some $S\in \mathcal{V}$ and $W\subseteq V$.  
    We make a case distinctions.
    If $|N_{\chi(S)}(W)| - \min_{c \in C} |N_{c}(W)| \le k$ (resp., $\max_{c \in C} |\nu_{c}(W)| - |\nu_{\chi(S)}(W)| \le k$), then by the definition of $X(S,W)$ (resp. $Y(S,W)$) and \Cref{eq:obs}, we have  $|N_{G_S^{\star}}(W)| = \min_{c \in C} |N_c(W)| + X(S, W) - \beta(S, W)$ (resp.,  $|\nu_{G_S^{\star}}(W)| = \max_{c \in C} |\nu_c(W)| - Y(S, W) + \alpha(S, W)$).
    Thus, Constraints (\ref{eq:secondconstraints}) need to be fulfilled in this case because \Cref{eq:intermediate} holds.
    Otherwise, it holds that $X(S, W)=k$ (resp. $Y(S, W)=k$) and  Constraints (\ref{eq:secondconstraints}) are fulfilled because  Constraints (\ref{eq:firstconstraints}) are fulfilled: 
    First note that by definition $\alpha(S,W)\leq k$ and  $\beta(S,W)\leq k$.
    Since $X(S, W) = k \ge \beta(S, W)$  (resp., $Y(S, W) = k \ge \alpha(S, W)$), we have $\sum_{v \in W} x_v \le \min_{c \in C} |N_c(W)| \le \min_{c \in C} |N_c(W)| + X(S, W) - \beta(S, W)$ (resp., $\sum_{v \in W} y_v \ge \max_{c \in C} |\nu_c(W)| \ge \max_{c \in C} |\nu_c(W)| - Y(S, W) + \alpha(S, W)$).
 
    \subparagraph{Soundness.}
	Assume that there exists a combination of $\mu \colon V \to 2^V$, $\mathcal{V} \subseteq 2^V$, and $X, Y \colon \mathcal{V} \times 2^V \to \{ 0, \dots, k \}$ such that the compatibility graph $H$ has a one-to-one matching  $M'$ covering the right side.
	Moreover, assume that $\{ x_v, y_v \mid v \in V \}$ is a solution for the constructed ILP.
	Then, using that the variables satisfy Constraints (\ref{eq:firstconstraints}), we can prove as in the proof of \Cref{th:max-min-k} that for every color $c \in C \setminus \{ M'(S) \mid S \in \mathcal{V} \}$ there is a matching $M_c$ in the graph $G_c = G[U_c \cup V]$ such that $x_v \le |M_c(v)| \le y_v$ for each $v \in V$.
	Considering the remaining colors, let $c$ be a color such that $c = M'(S)$ for some $S \in \mathcal{V}$.	
	Our goal is to show that there is a matching $M_c$ in $G_c$ such that $x_v \le |M_c(v)| \le y_v$ for each $v \in V$ and $M_c(v) \ne \emptyset$ for each $v \in S$.
	By the definition of compatibility, we have a matching $M_c'$ such that $|M_c'(v)| = 1$ and $N(M_{c}'(v)) = \mu(v)$ for each $v \in S$.
	Now consider the graph $G_c'$ obtained from $G_c$ by deleting all edges incident to a vertex $u \in U$ that is incident to some edge in $M_c'$ and then adding $M_c'$.
	Note that this is completely analogous to the construction of $G^{\star}_S$.
	We then have for each $W \subseteq V$,
	\begin{align*}
		&\sum_{v \in W} y_v \ge \max_{c' \in C} |\nu_{c'}(W)| - Y(S, W) + \alpha(S, W) \ge |\nu_c(W)| + \alpha(S, W) = |\nu_{G_c'}(W)| \text{ and } \\
		&\sum_{v \in W} x_v \le \min_{c' \in C} |N_{c'}(W)| + X(S, W) - \beta(S, W) \le |N_c(c)| - \beta(S, W) = |N_{G_c'}(W)|.
	\end{align*}
	Here, we have the first inequalities since $\{ x_v, y_v \mid v \in V \}$ satisfy Constraints (\ref{eq:secondconstraints}).
	The second inequalities follows because $c$ and $S$ are compatible and thus $X(S, W) = \min(|N_c(W)| - \min_{c' \in C} |N_{c'}(W)|, k)$ and $Y(S, W) = \min(\max_{c'\in C} |\nu_{c'}(W)| - |\nu_c(W)|, k)$.
	The last inequalities follow from the construction of $G_c'$ analogous to \Cref{eq:obs}.

	As we showed in the proof of \Cref{th:max-min-k}, this implies that there is a matching $M_c$ in $G_{c}'$ such that $x_v \le |M_c(v)| \le y_v$.
	Note that $M_c$ constitutes a matching in $G_c$ as well.
	Note also that $M_c(v) \ne \emptyset$ for each $v \in S$:
	The neighborhood of $M_c'(v)$ is $\{ v \}$ in $G_{c}'$ and hence $M_c'(v)$ must be matched to $v$ in $M_c$.
	To conclude the proof, observe that $M := \bigcup_{c \in C} M_c$ is an $\ell$-fair matching in $G$ with $M(v) \ne \emptyset$ for each $v \in V$.
\end{proof}

\subsection{MoV Fair Matching}
\label{sec:k:mov}
We now develop an FPT algorithm for \textsc{MoV Fair Matching} for the parameter $k$, which also works with the non-emptiness constraint.
Our algorithm has two parts.
In the first part, we give an FPT algorithm (using an ILP) for an auxiliary problem called \textsc{Targeted MoV Fair Matching}.
This is a variant of \textsc{MoV Fair Matching}, where for each $v\in V$, the two most frequent colors appearing in $M(v)$ are given as part of the input.
We establish the soundness of the ILP for \textsc{Targeted MoV Fair Matching} using again \Cref{lemma:k:separation,lemma:k:backward}.
In the second part, we present a (randomized) parameterized reduction from \textsc{MoV Fair Matching} to \textsc{Targeted Mov Fair Matching} using the color coding technique~\cite{DBLP:journals/jacm/AlonYZ95} (notably, the ``colors'' from the color coding technique are different from our colors).
To apply this technique, we show that the colors that appear (second) most frequently in $M(v)$ for some $v \in V$ ``stand out'' in a fair matching that fulfills certain conditions.
Then, the color coding technique essentially allows us to determine these colors.

\subparagraph*{Part I.}

First, we define an auxiliary problem, which we call \textsc{Targeted MoV Fair Matching}.
The input for \textsc{MoV Fair Matching} is also part of the input for \textsc{Targeted MoV Fair Matching}.
Moreover, \textsc{Targeted MoV Fair Matching} takes as input two functions
$\mu^1, \mu^2 \colon V \to C$. In
\textsc{Targeted MoV Fair Matching}, we ask for an $\ell$-fair matching $M$ such that for every vertex $v \in V$, $\mu^1(v)$ (resp., $\mu^2(v)$) is the most (resp., second most) frequent color among the vertices $M(v)$ matched to $v$ in $M$.

We now develop an FPT algorithm for \textsc{Targeted MoV Fair Matching} by means of an ILP.
Let $C_{1, 2} = \{ \mu^1(v), \mu^2(v) \mid v \in V \}$ be the set of colors that appear (second) most frequent among the vertices matched to some vertex in $V$ and let $C' = C \setminus C_{1, 2}$ be the set of other colors.
Notably, the size of $C_{1,2}$ is linearly bounded in our parameter $k$.
For every $v \in V$, we introduce a variable $y_v$ which represents the number of vertices of color $\mu^2(v)$ matched to $v$.
The values of $y_v$ need to be chosen in a way such that, for every color $c' \in C'$, there is a matching $M_{c'}$ in $G_{c'}$ such that $|M_{c'}(v)| \le y_v$ for all $v \in V$.
By \Cref{lemma:k:forward}, we obtain the following constraint which must be fulfilled and add it to the ILP:
\begin{align*}
	\sum_{v \in W} y_v \ge \max_{c' \in C'} |\nu_{c'}(W)| \quad \forall W \subseteq V.
\end{align*}
For the vertices of colors in $C_{1, 2}$, we impose constraints in the	same way as in \Cref{sec:kc}.
For every $c \in C_{1, 2}$ and $v \in V$, we introduce a variable $z_v^c$ which represents the number of vertices of color $c$ matched to $v$.
Then, we add the following constraints according to \Cref{lemma:k:forward}.
\begin{align}
	|\nu_c(W)| \le \sum_{v \in W} z_v^c \le |N_c(W)| \quad \forall W \subseteq V, c \in C_{1, 2}. \label{eq:MoVc}
\end{align}
Moreover, to ensure that $\mu^{\cdot}(v)$ is respected, we set for all $v\in V$: 
\begin{align*}
 y_v\geq z_v^c, \quad \forall c\in C_{1,2}\setminus \{\mu^1(v)\}
\end{align*}

In case we impose the non-emptiness constraint, we additionally require that $z_v^{\mu^1(v)} \ge 1$ for all $v \in V$.
Finally, we encode the $\ell$-fairness: $y_v = z_v^{\mu^2(v)} \text{ and } z_v^{\mu^2(v)} \le z_v^{\mu^1(v)} \le z_v^{\mu^2(v)} + \ell$ for all $v\in V$.

In order to show the correctness of the ILP, with the help of \Cref{lemma:k:separation}, we prove the following adaptation of \Cref{lemma:k:backward}, in which we show that there exists a matching of the vertices of colors from $C'$ to vertices from $V$ respecting $y_v$ for all $v\in V$.

\begin{restatable}{lemma}{backwardadapt}
	\label{lemma:k:backwardadapt}
	Let $G = (U \cupdot V, E)$ be a bipartite graph and let $\{ z_v \in \mathbb{N} \mid v \in V\}$ be a set of integers.
	Suppose that $\sum_{v \in W} z_v \ge |\nu_G(W)|$ for every $W \subseteq V$.
	Then, there is a left-perfect many-to-one matching $M$ such that $M(v) \le z_v$ for every $v \in V$.
\end{restatable}
\begin{proof}
	We use \Cref{lemma:k:separation,lemma:k:backward}.
	Let $f, f', g \colon 2^V \to \mathbb{N}$ be functions such that $f(W) = 0$, $f'(W) = \sum_{v \in W} z_v$, and $g(W) = |\nu(W)|$ for every $W \subseteq V$.
	Then, $f$ and $f'$ are modular by \Cref{lemma:mod} and $g$ is supermodular by \Cref{lemma:nu}.
	It follows from \Cref{lemma:k:separation} that there is a modular function $h \colon 2^V \to \mathbb{N}$ such that $g(W) \le h(W) \le f'(W)$ for every $W \subseteq V$ and $h(V) = \max_{W \subseteq V} f(W) + g(V \setminus W)$.
	Since $f(W) = 0$ for every $W \subseteq V$, we have $h(V) = \max_{W \subseteq V} g(W) = \max_{W \subseteq V} |\nu(W)| = |\nu(V)| = |U|$.
	By \Cref{lemma:k:backward}, we obtain a left-perfect many-to-one matching $M$  with $|M(v)| = h(v) \le f'(v)= z_v$ for every~$v \in V$.
\end{proof}

Using \Cref{lemma:k:backward,lemma:k:backwardadapt}, we can now show that the above constructed ILP is feasible if and only if the given \textsc{Targeted MoV Fair Matching} is a yes-instance.
\begin{restatable}{proposition}{ktargeted}\label{prop:k:targeted}
	\textsc{Targeted MoV Fair Matching} can be solved in $O^\star(k^{O(k^2)})$ time even with the non-emptiness constraint.
\end{restatable}
\begin{proof}
	It is easy to see that our ILP uses $O(k^2)$ variables and $O(2^k)$ constraints and that the construction takes $O^\star(2^k)$ time.
	Using Lenstra's algorithm \cite{Len83,Kan87}, this ILP can be solved in $O^\star(k^{O(k^2)})$ time.
	It remains to prove the correctness of our ILP.

	The completeness of our ILP follows from \Cref{lemma:k:forward}.
	For the soundness, suppose that $\{ y_v \mid v \in V \} \cup \{ z_v^c \mid v \in V, c \in C_{1, 2} \}$ is a feasible solution for the constructed ILP.
	Then, for every $c \in C_{1, 2}$, we have that $\sum_{v \in W} z_v^c \ge |\nu_{c}(W)|$ for every $W \subseteq V$ and $\sum_{v \in V} z_v^c = |U_c|$ (by Constraint \ref{eq:MoVc} for $W=V$).
	By \Cref{lemma:k:backward}, it follows that $G_c$ has a left-perfect many-to-one matching $M_c$ with $|M_c(v)| = z_v^c$ for every $v \in V$ and $c \in C_{1, 2}$.
	Moreover, for every $W \subseteq V$ and $c \in C \setminus C_{1, 2}$, we have that $|\nu_c(W)| \le \sum_{v \in W} y_v$.
	By \Cref{lemma:k:backwardadapt}, it follows that $G_c$ has a left-perfect many-to-one matching $M_c$ with $|M_c(v)| \le y_v = |M_{\mu^2(v)}(v)|$ for every $v \in V$ and $c \in C_{1, 2}$.
	Combining $M_c$ for all colors $c \in C$ yields an $\ell$-fair matching for $G$.
\end{proof}

\subparagraph*{Part II.}
We will employ the color coding technique to reduce \textsc{MoV Fair Matching} to \textsc{Targeted Mov Fair Matching}.
To do so, we introduce another auxiliary problem called \textsc{$\mathcal{Q}$-MoV Fair Matching}.
To define the problem, we first introduce additional notation.
Suppose that $M$ is a matching in the input graph $G = (U \cupdot V, E)$.
Let $\mu_M^1, \mu_M^2 \colon V \to C$ be mappings such that for every vertex $v \in V$, $\mu_M^1(v)$ (resp., $\mu_M^2(v)$) is the most (resp., second most) frequent color among the vertices $M(v)$ matched to $v$ in $M$.
When the maximum or second maximum is achieved by more than one color, we break ties according to a fixed linear order $\le_C$ on $C$.
Let $\mathcal{V} = \{ v^1, v^2 \mid v \in V \}$ be a set containing $2|V|$ elements and let $\mathcal{P}(M)$ be a partition of $\mathcal{V}$ into subsets $S \subseteq \mathcal{V}$ where for every $S\in \mathcal{P}(M)$, $v^i, v'^j \in S$ if and only if $\mu_M^i(v) = \mu_M^j(v')$ for $v,v'\in V$ and $i,j\in [2]$.

Using this, we define \textsc{$\mathcal{Q}$-MoV Fair Matching}. 
Here, $\mathcal{Q}$ is a partition of $\mathcal{V}$ and we assume that $\mathcal{Q}$ is fixed.
The input of \textsc{$\mathcal{Q}$-MoV Fair Matching} is identical to the input of \textsc{MoV Fair Matching}.
The difference is that \textsc{$\mathcal{Q}$-MoV Fair Matching} asks for a left-perfect $\ell$-fair many-to-one matching $M$ in $G$ \emph{consistent} with $\mathcal{Q}$, that is, $\mathcal{P}(M) = \mathcal{Q}$.
Clearly, an instance of \textsc{MoV Fair Matching} is a yes-instance if and only if there exists a partition $\mathcal{Q}$ of $\mathcal{V}$ such that the corresponding \textsc{$\mathcal{Q}$-MoV Fair Matching} instance is a yes-instance.
Since there are at most $k^{O(k)}$ ways to partition $\mathcal{V}$ (which is a set of $2k$ elements), we can afford to ``guess'' $\mathcal{Q}$ in our FPT algorithm.
We thus focus on solving \textsc{$\mathcal{Q}$-MoV Fair Matching}.

In particular, we assume without loss of generality that $\mathcal{Q}$ is \emph{top-color maximal}, which is defined as follows. 
For a matching $M$, let $\sigma(M) := \sum_{v \in V, i \in \{ 1, 2 \}} |M(v)_{\mu^i_M(v)}|$ be the sum of the occurrences of the two most frequent colors in $M(v)$ over all vertices $v \in V$.
Let $\mathcal{M}_{\mathcal{Q}}$ be the set of all $\ell$-fair left-perfect matchings consistent with $\mathcal{Q}$.
Then, let $\sigma_{\mathcal{Q}} = 0$ if $\mathcal{M}_{\mathcal{Q}} = \emptyset$ and $\sigma_{\mathcal{Q}} = \max_{M_{\mathcal{Q}} \in \mathcal{M}_{\mathcal{Q}}} \sigma(M_\mathcal{Q})$ otherwise.
Now, we say that $\mathcal{Q}$ is \emph{top-color maximal} if $\sigma_\mathcal{Q} \geq \sigma_{\mathcal{Q}'}$ for all partitions $\mathcal{Q}'$ of $\mathcal{V}$.

We solve \textsc{$\mathcal{Q}$-MoV Fair Matching} using color coding. 
We make a case distinction based on whether $\ell>0$ or $\ell=0$. 
To apply the color coding method, we show a rather technical lemma stating that if a color $c$ is among the two most frequent colors in $M(v)$ for some $v \in V$ in a certain fair matching $M$, then $c$ needs to ``stand out''---the number of its occurrence in the neighborhood of $v$ in $G$ is greater than for any other color $c'$, unless $c'$ also ``stands out''.

We first prove this result for $\ell>0$ and afterwards an analogous, slightly weaker result for $\ell=0$.

\begin{restatable}{lemma}{kcmax}
	\label{lemma:k:cmax}
	Let $\mathcal{I}$ be a yes-instance of {\normalfont\textsc{$\mathcal{Q}$-MoV Fair Matching}} with $\ell > 0$ and $\mathcal{Q}$ be top-color maximal.
	Then, an $\ell$-fair left-perfect many-to-one matching $M$ consistent with $\mathcal{Q}$ with $\sigma(M_{\mathcal{Q}}) = \sigma_{\mathcal{Q}}$ fulfills $|N_{\mu_M^i(v)}(v)| \ge |N_{c'}(v)|$ for for every $v\in V$, $i\in [2]$, and $c' \in  \{ c \in C \mid \forall v \in V, i \in \{ 1, 2 \} \colon c \ne \mu_M^i(v) \}$.
\end{restatable}
\begin{proof}
	Let $C_M := \{ \mu_{M}^i(v) \in C \mid v \in V, i \in \{ 1, 2 \} \}$ be the set of colors which are one of the two most frequent colors in $M(v)$ for some $v \in V$ and let $C_M' := C \setminus C_M$ be the set of other colors.

	Assume for the sake of contradiction that there is some $v\in V$ and $i\in [2]$ with $\mu_{M}^i(v) = c \in C_M$ and color $c' \in C_M'$ such that $|N_{c'}(v)| > |N_c(v)|$. 
	Recall that by the definition of $C'_M$ it holds that $|M(v)_c|\geq |M(v)_{c'}|$.
	Consider a left-perfect matching $M'$ obtained from $M$ as follows.
	Initially, let $M' := M$.
	We then repeat the following procedure $|M(v)_{c}| - |M(v)_{c'}| + 1$ times:
	We delete from $M'$ an arbitrary edge $\{ u, v' \} \in M'$ such that $u \in N(v) \subseteq U$ is a vertex of color $c'$ and $v'$ is some vertex in $V - v$, and then add an edge $\{ u, v \}$.
	Note that this is always possible, as $|M(v)_{c'}|\leq |M(v)_{c}| \le |N_c(v)| < |N_{c'}(v)|$.
	We claim that $M'$ is a left-perfect $\ell$-fair matching.
	It is easy to verify the left-perfectness of $M'$.
	For the $\ell$-fairness, since $c' \in C_M'$, for every $v'\in V-v$, the number of occurrences of the two most frequent colors in $M'(v')$ has not changed compared to $M(v')$, i.e., $\max^1_{c'' \in C} |M'(v')_{c''}|=\max^1_{c'' \in C} |M(v')_{c''}|$ and $\max^2_{c'' \in C} |M'(v')_{c''}| = \max^2_{c'' \in C} |M(v')_{c''}|$. 
	Thus, $M'(v')$ is $\ell$-fair for each $v' \in V - v$.
	It thus remains to show that $\max^1_{c'' \in C} |M'(v)_{c''}| - \max^2_{c'' \in C} |M'(v)_{c''}| \le \ell$.
	We consider three cases:
	\begin{itemize}
		\item
		If $c=\mu_{M}^1(v)$, then as we added $|M(v)_{c}| - |M(v)_{c'}| + 1$ vertices of color $c'$ to $M(v)$, we have $\max^1_{c'' \in C} |M'(v)_{c''}|=|M'(v)_{c'}| =|M(v)_{c}|+1$ and $\max^2_{c'' \in C} |M'(v)_{c''}|= |M(v)_c|$, and thus $\max^1_{c'' \in C} |M'(v)_{c''}| - \max^2_{c'' \in C} |M'(v)_{c''}| = 1 \le \ell$.
		\item
		Suppose that $c=\mu_{M}^2(v)$ and $\max^1_{c'' \in C} |M(v)_{c''}| > \max^2_{c'' \in C} |M(v)_{c''}|$.
		Then, we have $\max^1_{c'' \in C} |M'(v)_{c''}| = \allowbreak \max^1_{c'' \in C} |M(v)_{c''}|$ and $\max^2_{c'' \in C} |M(v)_{c''}| = \max^2_{c'' \in C} |M(v)_{c''}| + 1$, and thus $\max^1_{c'' \in C} |M'(v)_{c''}| - \max^2_{c'' \in C} |M'(v)_{c''}| \le \ell - 1 \le \ell$.
		\item
		Suppose that $c=\mu_{M}^2(v)$ and $\max^1_{c'' \in C} |M(v)_{c''}| = \max^2_{c'' \in C} |M(v)_{c''}|$.
		Then, we have $\max^1_{c'' \in C} |M'(v)_{c''}| =  \max^1_{c'' \in C} |M(v)_{c''}| + 1$ and $\max^2_{c'' \in C} |M(v)_{c''}| = \max^2_{c'' \in C} |M(v)_{c''}|$, and thus $\max^1_{c'' \in C} |M'(v)_{c''}| - \max^2_{c'' \in C} |M'(v)_{c''}| \le 1 \le \ell$.
	\end{itemize}
	This proves the $\ell$-fairness of $M$.

	However, the existence of $M'$ contradicts the assumption that $\mathcal{Q}$ is top-color maximal: 
	Observe that $c'$ is among the two most frequent colors in $M'(v)$ and $|M'(v)_c|>|M(v)_{\mu_M^2(v)}|$.
	Moreover, as $M(v)\subseteq M'(v)$ and the number of occurrences of each of the two most frequent colors in $M(v')$ for all $v\in V-v$ is the same in $M$ and $M'$, we have $\sigma(M')>\sigma(M)$.
	However, this contradicts the assumption that $\mathcal{Q}$ is top-color maximal, as $\sigma(M_{\mathcal{P}(M')})\geq \sigma(M')>\sigma(M)=\sigma(M_{\mathcal{Q}}) = \sigma_{\mathcal{Q}}$.
\end{proof}

We now prove an analogous, slightly weaker lemma for the case $\ell = 0$:
\begin{lemma}
	\label{lemma:k:cmaxl0}
	Let $\mathcal{I}$ be a yes-instance of {\normalfont\textsc{$\mathcal{Q}$-MoV Fair Matching}} with $\ell = 0$ and $\mathcal{Q}$ be top-color maximal.
	Then, a $0$-fair left-perfect many-to-one matching $M$ consistent with $\mathcal{Q}$ with $\sigma(M_{\mathcal{Q}}) = \sigma_{\mathcal{Q}}$ fulfills $|N_{\mu_M^i(v)}(v)| \ge |N_{c'}(v)|$ for every $v\in V$, $i\in [2]$, and $c' \in C' -c_v'$, where $C' = \{ c' \in C \mid \forall v \in V, i \in \{ 1, 2 \} \colon c' \ne \mu_M^i(v) \}$ and $c_v' \in C'$ is some color for each $v\in V$. (In other words, the constraint is satisfied for all but one color $c_v$.)
\end{lemma}
\begin{proof}
	Let $C_M := \{ \mu_{M}^i(v) \in C \mid v \in V, i \in \{ 1, 2 \} \}$ be the set of colors which are one of the two most frequent colors in $M(v)$ for some $v \in V$ and let $C_M' := C \setminus C_M$ be the set of other colors.

	Assume for the sake of contradiction that there is some $v\in V$ and $i\in [2]$ with $\mu_{M}^i(v) = c \in C_M$ and two colors $c_v', c_v''\in C_M'$ such that $|N_{c_v'}(v)| > |N_c(v)|$ and $|N_{c_v''}(v)| > |N_c(v)|$.
	Recall that by the definition of $C'_M$ it holds that $|M(v)_c|\geq |M(v)_{c_v'}|$ and $|M(v)_c|\geq |M(v)_{c_v''}|$.
	Consider a left-perfect matching $M'$ obtained from $M$ as follows.
	Initially, let $M' := M$.
	For both colors $c'_v$, respectively, $c''_v$, we then repeat the following procedure $|M(v)_{c}| - |M(v)_{c_v'}| + 1$ (resp., $|M(v)_c| - |M(v)_{c_v''}| + 1$) times:
	We delete from $M'$ an arbitrary edge $\{ u, v' \} \in M'$ such that $u \in N(v) \subseteq U$ is a vertex of color $c_v'$ (resp., $c_v''$) and $v'$ is a vertex in $V - v$, and then add edge $\{ u, v \}$.
	Note that this is always possible, as $|M(v)_{c_v'}| \le|M(v)_{c}| \le |N(v)_c| < |N_{c_v'}(v)|$ and $|M(v)_{c_v''}| \le|M(v)_{c}|  \le |N(v)_c| < |N(v)_{c_v''}(v)|$.
	One can show that $M'$ is a left-perfect $0$-fair many-to-one matching: 
	For $0$-fairness note that the two most frequent colors in $M(v')$ for all $v'\in V-v$ remain unchanged, while in $M(v)$ we have that $\max^1_{c''' \in C} |M'(v)_{c'''}| - \max^2_{c''' \in C} |M'(v)_{c'''}| = |M'(v)_{c'}| - |M'(v)_{c''}| = 0$).
	However, the existence of $M'$ contradicts assumption that $\mathcal{Q}$ is top-color maximal (using an analogous argument as in the proof of \Cref{lemma:k:cmax}).
\end{proof}

With \Cref{lemma:k:cmax,lemma:k:cmaxl0} at hand, we are ready to give a randomized reduction from \textsc{$\mathcal{Q}$-MoV Fair Matching} to \textsc{Targeted MoV Fair Matching}.

\begin{restatable}{lemma}{kcc}
	\label{lemma:k:cc}
	Let $\mathcal{I}$ be an instance of \textsc{$\mathcal{Q}$-MoV Fair Matching}.
	We can compute in polynomial time an instance $\mathcal{J}$ of \textsc{Targeted MoV Fair Matching} such that (i) $\mathcal{J}$ is a yes-instance with probability at least $\delta$ (where $1 / \delta \in k^{O(k)}$) if $\mathcal{I}$ is a yes-instance~and $\mathcal{Q}$ is top-color maximal, and
		(ii) $\mathcal{J}$ is a no-instance if $\mathcal{I}$ is a no-instance.
\end{restatable}
\begin{proof}
	We first present a proof for the case $\ell>0$ and afterwards a proof for the case $\ell=0$. 
    \proofsubparagraph*{Case: $\mathbf{\ell>0}$}
	We describe a polynomial-time procedure to construct an instance $\mathcal{J}$ of \textsc{Targeted MoV Fair Matching} from a given instance $\mathcal{I} = (G=(U \cupdot V,E),C,\col,\ell)$ of \textsc{$\mathcal{Q}$-MoV Fair Matching}.
	We randomly assign each color $c\in C$ to one of the subsets in $\mathcal{Q}$ with the intended meaning that if we assign color $c\in C$ to $S\in \mathcal{Q}$, then in a matching $M$ in $\mathcal{I}$ one of the following holds: 
	(i) for $v^1\in S$, $c$ appears as the most frequent color in $M(v)$ and for $v^2\in S$, $c$ appears as the second most frequent color in $M(v)$ and $c$ appears nowhere else as the most or second most frequent color or
    (ii) $c$ does not appear as the most or second most frequent color in $M(v)$ for any $v\in V$.
	Formally let $\lambda \colon C \to \mathcal{Q}$ be a function assigning each color to a subset in $\mathcal{Q}$, where each assignment is chosen uniformly and independently at random.\footnote{In the language of color coding, the function $\lambda$ is often described as assignments of ``colors'' (the ``colors'' in color coding are different from the colors used here).} 
	For every $v\in V$ and $i\in [2]$ with $v^i\in S$ for some $S \in \mathcal{Q}$, we find $c_{v}^i = \arg \max |N_c(v)|$, where $\arg \max$ is taken over all colors $c$ with $\lambda(c) = S$. 
	Then, we construct an instance $\mathcal{J}$ of \textsc{Targeted MoV Fair Matching}, where $\mu^i(v) = c_{v}^i$ for every $v^i \in \mathcal{V}$.
	If a hypothetical matching $M$ such that $\mu_M^i(u) = \mu^i(u)=c_{v}^i$ for $v\in V$ and $i\in [2]$ is not consistent with $\mathcal{Q}$, we let $\mathcal{J}$ be a trivial no-instance.
	Clearly, the construction of $\mathcal{J}$ takes polynomial time.

	Suppose that $\mathcal{I}$ is a yes-instance and $\mathcal{Q}$ is top-color maximal.
	We show that in this case $\mathcal{J}$ is a yes-instance with probability at least $\delta$ with $\delta^{-1} \in k^{O(k)}$.
	Let $M$ be an $\ell$-fair left-perfect many-to-one matching $M$ consistent with $\mathcal{Q}$ with $\sigma(M_{\mathcal{Q}}) = \sigma_{\mathcal{Q}}$.
	Assume that $\lambda(\mu_{M}^i(v)) = S_v^i$ holds for every $v \in V$ and $i \in [2]$, where $S_v^i \in \mathcal{Q}$ denotes the subset in $\mathcal{Q}$ to which $v^i$ belongs.
	We claim that under this assumption on $\lambda$, the instance $\mathcal{J}$ constructed by our procedure is a yes-instance.
	In fact, we show that $M$ is a solution of $\mathcal{J}$ (which also directly implies that there is a solution consistent with $\mathcal{Q}$ and thus that no trivial no-instance is returned).
	Since $M$ is an left-perfect $\ell$-fair matching in $G$, we only have to show that $\mu^i_M(v) = \mu^i(v) = c_v^i$ for every $v \in V$ and $i \in [2]$.

	Let $C_M := \{ \mu^i_M(v) \in C \mid v \in V, i \in [2] \}$ be the set of colors which are among the two most frequent colors in $M(v)$ for some $v \in V$ and let $C_M' := C \setminus C_M$ be the set of other colors.
	By \Cref{lemma:k:cmax}, for every $v \in V$, $c' \in C_M'$, and $i \in [2]$, we have $|N_{\mu_{M}^i(v)}(v)| \ge |N_{c'}(v)|$.
	Since we always break ties according to a fixed linear order (including when we find $\mu_M^i(v)$), we have $\mu_{M}^i(v) = \arg \max_{c \in C_M' \cup \{ \mu_M^i(v) \}} |N_c(v)|$.
	Recall that when constructing $\mathcal{J}$ we have defined $c_v^i = \arg \max |N_c(v)|$, where the maximum is taken over the set $C_v^i := \{ c \in C \mid \lambda(c) = S_v^i \}$.
	By our assumption that $\lambda(\mu_{M}^i(v)) = S_v^i$ for every $v \in V$ and $i \in [2]$, we have $\mu_M^i(v) \in C_v^i$.
	We also have $c \notin C_v^i$ for every $c \in C_M \setminus \{ \mu_M^i(v) \}$, which implies that $C_v^i \subseteq C_M' \cup \{ \mu_M^i(v) \}$.
	Consequently, we obtain $c_v^i = \arg \max_{c \in C_v^i} |N_c(v)|= \mu_M^i(v)$ for every $v \in V$ and $i \in [2]$.
	It follows that $M$ is a solution of $\mathcal{J}$.

	Finally, observe that the probability that $\lambda(\mu_{M}^i(v)) = S_v^i$ for every $v\in V$ and $i\in [2]$ is at least $|\mathcal{Q}|^{-|\mathcal{Q}|} \ge (2k)^{-2k}$, since $\lambda$ is chosen uniformly and independently at random. 
	Thus, if $\mathcal{I}$ is a yes-instance, $\mathcal{J}$ is a yes-instance with probability at least $(2k)^{-2k}$.
	
	If $\mathcal{J}$ is a yes-instance, then $\mathcal{I}$ is also a yes-instance, as, by construction of $\mathcal{J}$, if $\mathcal{J}$ is not a trivial no instance, a solution $M$ for $\mathcal{J}$ needs to satisfy $\mathcal{P}(M)=\mathcal{Q}$ and is thus also a solution for $\mathcal{I}$.
	
	\proofsubparagraph*{Case: $\mathbf{\ell=0}$}
	For $\ell = 0$, we choose $\lambda \colon C \to \mathcal{Q} \cup \{ \emptyset \}$ uniformly and independently at random.
	Informally speaking, we wish to have every color in $\{c'_v\mid v\in V\}$ (where $c_v'$ is the color from \Cref{lemma:k:cmaxl0}) map to $\emptyset$ in $\lambda$. 
	This means that these colors cannot be ``falsely'' picked as the most or second most frequent for some $v$. 
	Thus, \Cref{lemma:k:cmax} serves as a viable analogue of \Cref{lemma:k:cmaxl0} for the remaining part of the proof. 
	The event that $\lambda(c_v') = \emptyset$ for all $c \in \{c'_v\mid v\in V\}$ occurs with probability at least $(|\mathcal{Q}| + 1)^{-k}$.
	Moreover, the probability of the event that $\lambda(\mu_{M}^i(v)) = S$ for all $v \in V$ and $i \in [2]$ is at least $(|\mathcal{Q}| + 1)^{-|\mathcal{Q}|} \ge (2k + 1)^{-2k}$.
	Since these two events are independent, using a similar argumentation as for $\ell>0$ one can show that the lemma also holds for $\ell = 0$.
\end{proof}

We repeat the algorithm of \Cref{lemma:k:cc} independently $\delta^{-1} \in k^{O(k)}$ times on the given instance $\mathcal{I}$ of \textsc{$\mathcal{Q}$-Fair Matching}.
If $\mathcal{I}$ is a yes-instance and $\mathcal{Q}$ is top-color maximal, then at least one instance of \textsc{Targeted Mov Fair Matching} returned by the algorithm is a yes-instance with probability at least $1 - (1 - \delta)^{\delta^{-1}} \ge 1 - 1 / e$.
By \Cref{prop:k:targeted}, a yes-instance of \textsc{Targeted Mov Fair Matching} can be recognized in $O^\star(k^{O(k^2)})$ time.
We thus have a randomized algorithm to solve \textsc{$\mathcal{Q}$-MoV Fair Matching} in $O^\star(k^{O(k^2)})$ time if $\mathcal{Q}$ is top-color maximal.
Recall that $\mathcal{Q}$ is a partition of a $2k$-element set.
So there are at most $k^{O(k)}$ choices for $\mathcal{Q}$ (one of which needs to be top-color maximal), which gives us the following theorem (we remark that our algorithm can be derandomized using a standard method \cite{DBLP:books/sp/CyganFKLMPPS15}):

\begin{theorem}
	There is a randomized $O^\star(k^{O(k^2)})$-time algorithm to solve \textsc{MoV Fair Matching} even with the non-emptiness constraint.
\end{theorem}

\section{Complexity Dichotomies with respect to $|C|$ and Maximum Degree}
\label{sec:dic}

In this section, we study the computational complexity of \textsc{Max-Min/MoV Fair Matching} for fixed values of $\Delta_U$, $\Delta_V$, and $|C|$.
Recall that $\Delta_U$ (resp., $\Delta_V$) is the maximum degree of all vertices in $U$ (resp., $V$).
We identify several computational dichotomies regarding these parameters (see \Cref{fig:dichotomies}).

\subsection{Dichotomy with respect to $|C|$}

We first show a dichotomy on $|C|$.
In particular, \textsc{Max-Min/MoV Fair Matching} is polynomial-time solvable for $|C| = 2$ (even if there are arbitrary size lower bounds), while it is NP-hard for $|C|=3$:

\begin{restatable}{theorem}{dicc}
	\label{thm:dic:c}
	\textsc{Max-Min/MoV Fair Matching} is polynomial-time solvable for $|C| = 2$ and NP-hard for $|C| \ge 3$. The polynomial-time result even holds in the presence of arbitrary size lower bounds.
\end{restatable}

We first prove that \textsc{Max-Min/MoV Fair Matching} is NP-hard for $|C| \ge 3$ and then that it is polynomial-time solvable for $|C| = 2$.

\subparagraph*{Hardness.}

Stoica et al.~\cite{DBLP:conf/atal/StoicaCDG20} showed that \textsc{MoV Fair Matching} is NP-hard for $|C| = 3$.
We show that \textsc{Max-Min Fair Matching} is also NP-hard even if there are three colors.

\begin{proposition} \label{pr:hardnessMaxMin}
	\textsc{Max-Min Fair Matching} is NP-hard for $|C| = 3$, even if $\Delta_U \le 3$ and $\Delta_V \le 3$.
\end{proposition}
\begin{proof}
	We reduce from the following NP-hard variant of \textsc{3-Dimensional Matching}:
	Given three disjoint sets $X$, $Y$, and $Z$ and a subset $T \subseteq X\times Y \times Z$ of triplets where each element from $X\cup Y \cup Z$ occurs in at most three triples from $T$, the question is whether there is a perfect 3D-matching, i.e., a subset of triples $T'\subseteq T$ such that each element from $X\cup Y \cup Z$ occurs exactly in one triple from $T'$ \cite[Problem SP1]{DBLP:books/fm/GareyJ79}.

	Given an instance $(X\cupdot Y \cupdot Z, T)$ of \textsc{3-Dimensional Matching}, we construct an instance of \textsc{Max-Min Fair Matching} as follows.
	We set $\ell=0$ and introduce three colors $C = \{ \alpha, \beta, \gamma \}$.
	For each triple $t\in T$, we add a vertex $v_t$ to $V$.
	For each element $s\in X$ (resp., $s\in Y$ and $s\in Z$), we introduce a vertex $u_s$ of color $\alpha$ (resp., $\beta$ and $\gamma$) to $U$ and connect it to the vertices from $V$ corresponding to triples in which $s$ occurs.
	Note that $\Delta_U, \Delta_V \le 3$.

	Suppose that $M$ is a fair matching for the constructed instance.
	Note that if $M(v_t)$ is non-empty for some $t\in T$, then all vertices from $U$ corresponding to the three elements from $t$ must be assigned to $M(v_t)$, as $M(v_t)$ is $0$-fair.
	Thus, each $0$-fair left-perfect many-to-one matching in the constructed \textsc{Max-Min Fair Matching} instance corresponds to a perfect 3D-matching in the given \textsc{3-Dimensional Matching} instance and vice versa.
\end{proof}

\subparagraph*{Algorithm.}

We now show that \textsc{Max-Min/MoV Fair Matching} is polynomial-time solvable for $|C| = 2$.
Our algorithm works for arbitrary size lower bounds.

\begin{proposition}
	\label{prop:dic:cpoly}
	\textsc{Max-Min/MoV Fair Matching} can be solved in polynomial time for $|C| = 2$ even with arbitrary size lower bounds.
\end{proposition}
\begin{proof}
Let $\mathcal{I}=(G=(U \cupdot V,E), C, \col, \ell)$ be an instance of \textsc{Max-Min/MoV Fair Matching}  with $C = \{ \alpha, \beta \}$ and size lower bound $\smin$.
We will describe a polynomial-time algorithm that determines whether there is an $\ell$-fair matching $M$ such that $|M(v)|\geq \smin$ for every $v \in V$ by reducing to \textsc{Maximum Weighted Matching}.

For this, we will construct a graph $G'=(V',E')$ and a subset $S \subseteq V'$ of vertices such that there is a one-to-one matching $M'$ in $G'$ where all vertices from $S$ are matched if and only if $\mathcal{I}$ is a yes-instance.
Note that given $G'$ and $S$, we can determine in polynomial time whether such a matching $M'$ exists, as this problem can be easily reduced to an instance of \textsc{Maximum Weighted Matching}.

\proofsubparagraph*{Construction.}
We give the general description of how we construct $G'$ and $S$.
We define a graph $G_v'$ for every $v \in V$, whose connected components all contain at most four vertices.
In every connected component of $G_v'$, we mark at most one vertex by $1^{\star}$ and at most on vertex by $2^{\star}$. 
Moreover, we mark some vertices by $s^{\star}$.
To construct $G'$, we add the vertices $U$ and the graph $G_v'$ for every $v \in V$.
Moreover, for every $u \in U_{\alpha}$ (resp., $u\in U_{\beta}$) and $v \in N_G(u)$, we add an edge between $u$ and every vertex marked by $1^{\star}$ (resp., $2^{\star}$) in $G_v$.
We include $U$ and every vertex marked by $s^{\star}$ into $S$.

\begin{figure}[t]
	\centering
	\begin{tikzpicture}[xscale=.95]
		\begin{scope}
			\node[vertex, label={[label distance=0.1cm]90:$1^{\star}$}] (v11) at (0, 0) {};
			\node[vertex, label={[label distance=0.1cm]270:$s^{\star}$}] (v12) at (1, 0) {};
			\node[vertex, label={[label distance=0.1cm]270:$s^{\star}$}] (v13) at (2, 0) {};
			\node[vertex, label={[label distance=0.1cm]90:$2^{\star}$}] (v14) at (3, 0) {};
			\node at (3.5, -.7) {$H_1$};
			\draw (v11) -- (v12) -- (v13) -- (v14);
		\end{scope}
		\begin{scope}[shift={(4.3, 0)}]
			\node[vertex, label={[label distance=0.1cm]90:$1^{\star}$}, label={[label distance=0.1cm]270:$s^{\star}$}] (v11) at (0, 0) {};
			\node[vertex] (v12) at (1, 0) {};
			\node[vertex, label={[label distance=0.1cm]90:$2^{\star}$}, label={[label distance=0.1cm]270:$s^{\star}$}] (v14) at (2, 0) {};
			\node at (2.5, -.7) {$H_2$};
			\draw (v11) -- (v12) -- (v14);
			\draw (v11) to [out=45,in=135] (v14);
		\end{scope}
		\begin{scope}[shift={(7.6, 0)}]
			\node[vertex, label={[label distance=0.1cm]90:$1^{\star}$}, label={[label distance=0.1cm]270:$s^{\star}$}] (v11) at (0, 0) {};
			\node[vertex] (v12) at (1, 0) {};
			\node[vertex, label={[label distance=0.1cm]90:$2^{\star}$},label={[label distance=0.1cm]270:$s^{\star}$}] (v14) at (2, 0) {};
			\node at (2.5, -.7) {$H_3$};
			\draw (v11) -- (v12) -- (v14);
		\end{scope}
		\begin{scope}[shift={(10.9, 0)}]
			\node[vertex, label={[label distance=0.1cm]90:$1^{\star}$}] (v11) at (0, 0) {};
			\node[vertex, label={[label distance=0.1cm]270:$s^{\star}$}] (v12) at (1, 0) {};
			\node[vertex, label={[label distance=0.1cm]90:$2^{\star}$}] (v14) at (2, 0) {};
			\node at (2.5, -.7) {$H_4$};
			\draw (v11) -- (v12) -- (v14);
		\end{scope}
	\end{tikzpicture}

	\begin{tikzpicture}
		\begin{scope}
			\node[
				vertex,
				label={[label distance=0.1cm]90:$1^{\star}$},
				label={[label distance=0.1cm]270:$s^{\star}$}
			] (v11) at (0, 0) {};
			\node[
				vertex,
				label={[label distance=0.1cm]90:$2^{\star}$},
				label={[label distance=0.1cm]270:$s^{\star}$}
			] (v11) at (1, 0) {};
			\node at (1.5, -.7) {$H_5$};
		\end{scope}
		\begin{scope}[shift={(3, 0)}]
			\node[
				vertex,
				label={[label distance=0.1cm]100:$1^{\star}$},
				label={[label distance=0.1cm]80:$2^{\star}$},
				label={[label distance=0.1cm]270:$s^{\star}$}
			] (v11) at (0, 0) {};
			\node at (.5, -.7) {$H_6$};
		\end{scope}
	\end{tikzpicture}
	\caption{Possible connected components $H_1,\dots, H_6$ of $G'_v$.}
	\label{fig:cc}
\end{figure}
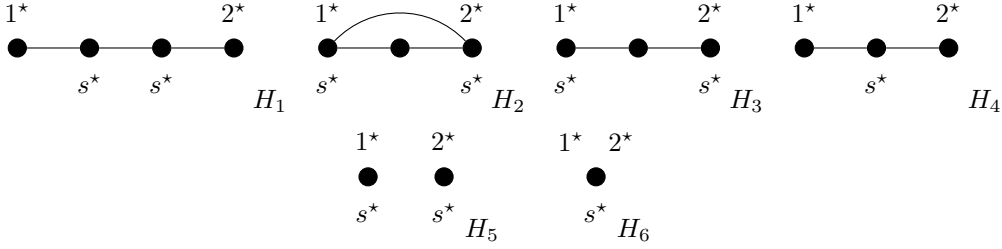

Now we describe how to construct the $G_v'$'s.
In fact, $G'_v$ will look the same for all $v\in V$ and only depends on $\ell$ and $p$ (and $n$).
Every connected component of $G_v'$ is one of the six graphs $H_1, \cdots, H_6$ (with markers) depicted in \Cref{fig:cc}, where we include $s_i$ copies of $H_i$ into $G_v'$. 
The values of $s_i$ that are specified in \Cref{table:ni}, where we have four possible cases depending on whether $\ell>0$, $p < \ell$ and the parity of $p - \ell$.

\begin{table}
	\begin{center}
		\begin{tabularx}{.78\textwidth}{lccccccc}
			\toprule
			& $s_1$ & $s_2$ & $s_3$ & $s_4$ & $s_5$ & $s_6$ \\
			\midrule
			$\ell = 0$ & $n$ & 0 &  0 & 0 & $\lceil \frac{p}{2} \rceil$ & 0 \\
			$0 \leq p < \ell$ & $n$ & 1 & 0 & $\ell - p - 1$ &  0 & $p$ \\
			$p \ge \ell > 0$ and $p - \ell \equiv 0$ & $n$ & 0 & 1 & 0 & $\frac{1}{2}(p - \ell) $ & $\ell - 1$ \\
			$p \ge \ell > 0$ and $p - \ell \equiv 1$ & $n$ & 0 &  0 & 1 & $\frac{1}{2}(p - \ell + 1) $ & $\ell - 1$ \\
			\bottomrule
		\end{tabularx}
	\end{center}
	\caption{The choice of $s_i$ for $i \in [6]$.}
	\label{table:ni}
\end{table}

\proofsubparagraph*{Correctness.}
For every graph $H_i$, let $\mathcal{A}_i$ be a set of row vectors $x \in \{ 0, 1 \}^2$ such that the graph obtained by deleting from $H_i$ a vertex marked by $1^{\star} $ if $x_{1} = 1$ and a vertex marked by $2^{\star} $ if $x_{2} = 1$ has a matching in which all vertices marked by $s^{\star}$ are matched.
Note that
\begin{equation}
	\begin{aligned}
	&\mathcal{A}_1 = \left\{ [0, 0], [1, 1] \right\}, \quad 
	\mathcal{A}_2 = \left\{ [0, 0], [0, 1], [1, 0], [1, 1] \right\},\quad
	\mathcal{A}_3 = \left\{ [0, 1], [1, 0], [1, 1] \right\}, \\
	&\mathcal{A}_4 = \left\{ [0, 0], [0, 1], [1, 0] \right\}, \quad 
	\mathcal{A}_5 = \left\{ [1, 1] \right\}, \quad
	\mathcal{A}_6 = \left\{ [1, 0], [0, 1] \right\}.
	\end{aligned} \label{eq:ai}
\end{equation}

Let $\mathcal{P}$ be the collection of pairs $(i, x)$, where $i \in [6]$ and $x \in \mathcal{A}_i$.
For $m_{1}, m_{2},p,\ell\in \mathbb{N}$, we say that a vector $a \in \mathbb{N}^{\mathcal{P}}$ \emph{represents} $[m_{1}, m_{2}]$ with respect to $p$ and $\ell$ if
\begin{equation}
	\label{eq:ni}
	[m_{1}, m_{2}] = \sum_{i \in [6]} \sum_{x \in \mathcal{A}_i} a_{(i,x)} x \text{ and }
	\sum_{x \in \mathcal{A}_i} a_{(i,x)} = s_i \text{ for each } i \in [6],
\end{equation}
where $s \in \mathbb{N}^{6}$ is a vector such that $s_i$ is defined according to \Cref{table:ni}.
We will show the following: 
\begin{equation} \tag{$\star$}
\parbox{0.9\textwidth}{
	 For $m_1,m_2,\smin, \ell\in \mathbb{N}$, we have $m_1 + m_2 \ge p$ and $|m_{1} - m_{2}| \le \ell$ if and only if there exists a vector $a \in \mathbb{N}^{\mathcal{P}}$ representing $[m_{1}, m_{2}]$ with respect to $p$ and $\ell$.
}
\end{equation}
From ($\star$) we can easily establish the correctness of our construction: 
Assume that an $\ell$-fair matching $M$ in $G$ is given with $|M(v)|\geq \smin$ for all $v\in V$. 
Then for some $v\in V$, let $m_1=|M(v)_{\alpha}|$ and $m_2=|M(v)_{\beta}|$. 
We then have $m_1 + m_2 \ge p$ and $|m_{1} - m_{2}| \le \ell$.
By ($\star$), there is a vector $a \in \mathbb{N}^{\mathcal{P}}$ representing $[m_{1}, m_{2}]$. 
This implies that we can find a matching $M'_v$ in $G'$ where all vertices from $G'_v$ that are in $S$ are matched and the vertices from $U$ that are matched to some vertex in $G'_v$ are exactly $M(v)$.
More precisely, for each $(i, x) \in \mathcal{P}$, the vertex marked by $1^{\star}$ (resp., $2^{\star}$) should be matched to some vertex from $U_\alpha$ (resp., $U_\beta$) in $M'_v$ if $x_1 = 1$ (resp., $x_2 = 1$) in $a_{(i, x)}$ copies of $H_i$ in $G_v'$.
Notably, the constructed matching $M'_v$ is disjoint for all $v\in V$. 
Thus, $M'=\bigcup_{v\in V} M'_v$ is a solution to the constructed instance. 
For the reverse direction, assume that there is a matching $M'$ in $G'$ matching all vertices from $S$. 
Then, for $v\in V$, let $m_1$ (resp., $m_2$) be the number of vertices from $U_{\alpha}$ (resp., $U_{\beta}$) matched to a vertex from $G'_v$. 
As $M'$ matches all vertices from $S$, there exists $a \in \mathbb{N}^{\mathcal{P}}$ representing $[m_{1}, m_{2}]$. 
By ($\star$), we obtain $m_1 + m_2 \ge p$ and $|m_{1} - m_{2}| \le \ell$. 
Thus, the matching $M$ where for every $v\in V$ we match to $v$ all vertices from $U$ that are matched to a vertex from $G'_v$ in $M'$ is a solution to the given instance.   

We now turn to the proof of ($\star$). We first consider the \emph{if} part.
Assume that for some $m_1,m_2,\smin, \ell\in \mathbb{N}$ we have a vector $a \in \mathbb{N}^\mathcal{P}$ representing $[m_{1}, m_{2}]$ with respect to $p$ and $\ell$.
We show that $m_1 + m_2 \ge p$ and $|m_{1} - m_{2}| \le \ell$.
Since $m_{1} + m_{2} = \sum_{i \in [6]} \sum_{x \in \mathcal{A}_i} a_{(i,x)} (x_1 + x_2)$ and  $m_{1} - m_{2} = \sum_{i \in [6]} \sum_{x \in \mathcal{A}_i} a_{(i,x)} (x_1 - x_2)$, we have by \Cref{eq:ai},
\begin{equation}
	\begin{aligned}
		\label{eq:goal}
		&m_{1} + m_{2} \ge \sum_{i \in [6]} \sum_{x \in \mathcal{A}_i} a_{(i,x)} \cdot \min_{x \in \mathcal{A}_i} (x_1 + x_2) = \sum_{i \in [6]} s_i \cdot \min_{x \in \mathcal{A}_i} (x_1 + x_2) = s_3 + 2s_5 + s_6 \text{ and } \\
		&|m_{1} - m_{2}| \le \sum_{i \in [6]} \sum_{x \in \mathcal{A}_i} a_{(i,x)} \cdot \max_{x \in \mathcal{A}_i} |x_1 - x_2| \le \sum_{i \in [6]} s_i \cdot \max_{x \in \mathcal{A}_i} |x_1 - x_2| = s_2 + s_3 + s_4 + s_6. 
	\end{aligned}
\end{equation}
It is straightforward to verify that $s_3 + 2s_5 + s_6 \geq p$ and $s_2 + s_3 + s_4 + s_6 \le \ell$ hold for every case given in \Cref{table:ni}.
Thus, we have that $m_1+m_2\geq p$ and $|m_{1} - m_{2}|\leq \ell$, which concludes the first direction.

We now turn to the \emph{only-if} part of ($\star$).
Suppose for some $m_1,m_2,\smin, \ell\in \mathbb{N}$ that $m_{1} + m_{2} \ge p$ and $|m_{1} - m_{2}| \le \ell$.
Without loss of generality, assume that $m_1 \le m_2$.
We show that there exists a vector $a \in \mathbb{N}^{\mathcal{P}}$ representing $[m_{1}, m_{2}]$ with respect to $p$ and $\ell$.
For each $(i, x) \in \mathcal{P}$, let $e_i^x \in \mathbb{N}^{\mathcal{P}}$ be a vector such that $(e_i^x)_{(i,x)} = 1$ and all other entries are 0.
Further, for each $i \in [6]$, let $e_i$ be a vector such that $e_i = 1$ and all other entries are $0$. 

For $\ell = 0$, we have $m_1 = m_2 \geq \lceil \frac{1}{2} p \rceil$. 
Let $t=m_1-\lceil \frac{1}{2} p \rceil$. 
Then, $\lceil \frac{1}{2} p \rceil e_5^{(1, 1)}+t e_1^{(1,1)}+(n-t) e_1^{(0,0)}$ represents $[m_1, m_2]$.

For $\ell > 0$, we make a case distinction based on the relation of $p$ and $\ell$. 
In the following, when we speak of some $p,\ell,m_1,m_2$, then we always assume that $m_{1} + m_{2} \ge p$ and $|m_{1} - m_{2}| \le \ell$.
First we consider the case $p<\ell$.
We prove that there is an $a\in \mathbb{N}^{\mathcal{P}}$ representing $[m_1,m_2]$ with respect to $p$ and $\ell$ by induction on $(\ell - p, \ell)$ assuming that $\ell>0$ and $p<\ell$. 

 \proofsubparagraph*{Base case.} Our base case is $(\ell - p, \ell)=(1,1)$ (as $\ell>0$ and $\ell>p$), implying that $\ell=1$ and $p=0$. 
In this case, we have $s = [n, 1, 0, 0, 0, 0]$.
Note that $m_2 \in \{ m_1, m_1 + 1 \}$.
If $m_2 = m_1$, then $(n - m_1) e_1^{(0, 0)} + m_1 e_1^{(1, 1)}+e_2^{(0,0)}$ represents $[m_1, m_2]$ with respect to $p$ and $\ell$.
If $m_2 = m_1 + 1$, then $(n - m_1) e_1^{(0, 0)} + m_1 e_1^{(1, 1)} + e_2^{(0, 1)}$ represents $[m_1, m_2]$ with respect to $p$ and $\ell$.

\proofsubparagraph*{Inductive step.}
As our induction hypothesis, we assume that for $(x,y) = (\ell', \ell' - p') \in \mathbb{N}^2$ with $y\geq x\geq 1$, the following holds:
If $p' \le m_{1}' + m_{2}'$ and $|m_{1}' - m_{2}'| \le \ell'$ for $m_1', m_2' \in \mathbb{N}$, then there exists a vector $a' \in \mathbb{N}^{\mathcal{P}}$ representing $[m_{1}', m_{2}']$ with respect to $p'$ and $\ell'$.
Our inductive step consists of two parts: 

\textbf{Case 1. }$(x,y) = (\ell', \ell' - p') \to(x+1,y) = (\ell, \ell - p)$:  This means that we have $\ell=\ell'+1$ and $p=p'+1$.  
Let $s \in \mathbb{N}^6$ (resp., $s' \in \mathbb{N}^6$) be the vectors defined by \Cref{table:ni} for $p$ and $\ell$ (resp., $p'$ and $\ell'$).  
Note that $s= s'+ e_6$.
Let $b$ be a vector representing $[m_1, m_2-1]$ with respect to $p'=p-1$ and $\ell'=\ell - 1$ (which exists by the induction hypothesis, as we have that $\ell'=y\geq 1$, $|m_1-m_2|\leq \ell$, and $m_1+m_2\geq p$).
Then, $b+e_6^{(0,1)}$ represents $[m_1,m_2]$ with respect to $p$ and $\ell$.

\textbf{Case 2. }$(x,y) = (\ell', \ell' - p') \to(x+1,y+1) = (\ell, \ell - p)$:
This means that we have $\ell=\ell'+1$ and $p=p'$.  
Let $s \in \mathbb{N}^6$ (resp., $s' \in \mathbb{N}^6$) be the vectors defined by \Cref{table:ni} for $p$ and $\ell$ (resp., $p'$ and $\ell'$).  
Note that $s = s' + e_4$.
If $m_2 - m_1 \le \ell - 1$, then let $b$ be a vector representing $[m_1, m_2]$ with respect to $p'=p$ and $\ell'=\ell - 1$ (which exists by the induction hypothesis).
Then, $b + e_4^{(0, 0)}$ represents $[m_1, m_2]$.
Otherwise, we have $m_2 - m_1 = \ell$.
Then, $(n - m_1) e_1^{(0, 0)} + m_1 e_1^{(1, 1)} + \sum_{i \in [2, 4, 6]}s_i  e_i^{(0, 1)}$ represents $[m_1, m_2]$ with respect to $p$ and $\ell$, as $\sum_{i \in [2, 4, 6]}s_i = \ell$.

\medskip

We now move on to the second case $p\geq \ell$. 
We prove that there is an $a\in \mathbb{N}^{\mathcal{P}}$ representing $[m_1,m_2]$ with respect to $p$ and $\ell$ by induction on $(p-\ell, p)$ assuming that $p\geq \ell>0$. 
Notably we thus have $p-\ell<p$.

\proofsubparagraph*{Base case.}Our base cases are 
$(p-\ell,p)\in \{(0,1),(1,2)\}$. 
For $(p-\ell,p)=(0,1)$, we have $p=1$ and $\ell=1$. 
In this case, we have $s = [n, 0, 1, 0, 0, 0]$. 
Note that $m_2 \in \{ m_1, m_1 + 1 \}$ and $m_2>0$.
If $m_2 = m_1$, then $(n - m_1) e_1^{(0, 0)} + (m_1-1) e_1^{(1, 1)}+  e_3^{(1, 1)}$ represents $[m_1, m_2]$ with respect to $p$ and $\ell$.
If $m_2 = m_1 + 1$, then $(n - m_1) e_1^{(0, 0)} + m_1 e_1^{(1, 1)} + e_3^{(0, 1)}$ represents $[m_1, m_2]$ with respect to $p$ and $\ell$.

For $(p-\ell,p)=(1,2)$, we have $p=2$ and $\ell=1$. 
In this case, we have $s = [n, 0, 0, 1, 1, 0]$.  
Note that $m_2 \in \{ m_1, m_1 + 1 \}$ and hence $m_1 \ge \lceil \frac{1}{2}(p - 1) \rceil = 1$.
If $m_2 = m_1$, then $(n - m_1) e_1^{(0, 0)} + (m_1-1) e_1^{(1, 1)}+ e_4^{(0,0)}+ e_5^{(1, 1)}$ represents $[m_1, m_2]$ with respect to $p$ and $\ell$.
If $m_2 = m_1+1$, then $(n - m_1) e_1^{(0, 0)} + (m_1-1) e_1^{(1, 1)}+ e_4^{(0,1)}+ e_5^{(1, 1)}$ represents $[m_1, m_2]$ with respect to $p$ and $\ell$.

\proofsubparagraph*{Inductive step.}
As our induction hypothesis, we assume that for $(x,y) = (p' - \ell', p') \in \mathbb{N}^2$ with $y\geq x\geq 1$, the following holds:
If $p' \le m_{1}' + m_{2}'$ and $|m_{1}' - m_{2}'| \le \ell'$ for $m_1', m_2' \in \mathbb{N}$, then there exists a vector $a' \in \mathbb{N}^{\mathcal{P}}$ representing $[m_{1}', m_{2}']$ with respect to $p'$ and $\ell'$.
Our inductive step consists of two parts: 

\textbf{Case 1. }$(x,y) = (p' - \ell', p') \to(x,y+1) = (p - \ell, p)$:
This means that we have $\ell=\ell'+1$ and $p=p'+1$.  
Let $s \in \mathbb{N}^6$ (resp., $s' \in \mathbb{N}^6$) be the vectors defined by \Cref{table:ni} for $p$ and $\ell$ (resp., $p'$ and $\ell'$).
Note in this case we have $s=s'+s_6$. 
Let $b$ be a vector representing $[m_1, m_2-1]$ with respect to $p'=p-1$ and $\ell'=\ell - 1$ (which exists by the induction hypothesis, as we have that $\ell'=y-x\geq 1$, $|m_1-m_2|\leq \ell$, and $m_1+m_2\geq p$).
Then, $b+e_6^{(0,1)}$ represents $[m_1,m_2]$ with respect to $p$ and~$\ell$.

\textbf{Case 2. }$(x,y) = (p' - \ell', p') \to(x+2,y+2) = (p - \ell, p)$:
This means that we have $\ell=\ell'$ and $p=p'+2$.  
Let $s \in \mathbb{N}^6$ (resp., $s' \in \mathbb{N}^6$) be the vectors defined by \Cref{table:ni} for $p$ and $\ell$ (resp., $p'$ and $\ell'$).
Note in this case we have $s=s'+s_5$.
Let $b$ be a vector representing $[m_1-1,m_2-1]$ with respect to $p'=p-2$ and $\ell'$ (which exists by the induction hypothesis, as $m_1-1+m_2-1\geq p'$). 
Then, $b + e_5^{(1, 1)}$ also represents $[m_1, m_2]$ with respect to $p$ and $\ell$.

This concludes the proof of ($\star$).
\end{proof}

\subsection{Dichotomy with respect to Maximum Degree}

We now present complexity dichotomies with respect to $\Delta_U$ and $\Delta_V$. All NP-hardness results here hold for three colors, while all polynomial-time results hold for an arbitrary number of colors.
We start by considering \textsc{MoV} (\Cref{sec:dich_mov}) and afterwards turn to \textsc{Max-Min} (\Cref{sec:dich_mm}).

\subsubsection{MoV} \label{sec:dich_mov}

\begin{restatable}{theorem}{dicmov}
	\label{thm:dic:mov}
	\textsc{MoV Fair Matching} is polynomial-time solvable if $\Delta_U \le 1$ or $\Delta_V \le 4$ (even with the non-emptiness constraint) and NP-hard otherwise.
\end{restatable}

Note that polynomial-time solvability for $\Delta_U=1$ is obvious, as in this case we know exactly where to assign each vertex from the left side.
To prove \Cref{thm:dic:mov}, we first show that \textsc{MoV Fair Matching} is polynomial-time solvable for $\Delta_V \le 4$ (even with the non-emptiness constraint) (\Cref{prop:dic:mov:algo}) and that it is NP-hard for $\Delta_U = 2$ and $\Delta_V = 5$ (\Cref{prop:mov:hardness}).

\subparagraph*{Algorithm.}
We start with giving an algorithm that solves \textsc{MoV} in polynomial-time for $\Delta_V \le 4$.
Our algorithm is a polynomial-time reduction to a polynomial-time solvable special case of \textsc{General Factor}, which we will define in the proof.

\begin{restatable}[\appmark]{proposition}{movdel}
	\label{prop:dic:mov:algo}
	\textsc{MoV Fair Matching} is polynomial-time solvable for $\Delta_V \le 4$ even with the non-emptiness constraint.
\end{restatable}

\begin{proof}[Proof Sketch.]
	We give a polynomial-time reduction to \textsc{General Factor}:
	In an instance of \textsc{General Factor} the input is an undirected graph $H = (W, F)$ and a degree list function $L \colon W \to 2^{\mathbb{N}}$ such that $L(w) \subseteq \{ 0, \dots, \deg_G(w) \}$ for every vertex $w \in W$.
	The problem asks for a spanning subgraph $H' = (W, F')$ for $F' \subseteq F$ such that $\deg_{H'}(w) \in L(w)$ for every $w \in W$.
	By a result of Cornu{\'{e}}jols \cite{DBLP:journals/jct/Cornuejols88}, \textsc{General Factor} is polynomial-time solvable if for every $w \in W$, the degree list $L(w)$ has gaps of size at most one, i.e., $\{ \min L(w), \min L(w) + 1, \dots, \max L(w) \} \setminus L(w)$ does not contain any two consecutive integers.
	Given an instance $\mathcal{I}=(G=(U\cupdot V, E), C, \col, \ell)$
	of \textsc{MoV Fair Matching} with $\Delta_V \le 4$, we will construct an equivalent instance $\mathcal{J} = (H = (W, F), L)$ of this polynomial-time solvable special case of \textsc{General Factor}.
	
	In the following, we give an overview of our construction.
	For every $v \in V$, we add a subgraph $H_v = (W_v, F_v)$ to $H$ which contains all vertices from $U$ adjacent to $v$ in $G$ ($N_{G}(v) \subseteq W_v$) but no other vertices from $U$.
	For the construction of $H_v$, we make extensive case distinctions depending on $N_G(v)$ and $\ell$.
	See \Cref{fig:hv-ex} for two examples.
	The construction of $H_v$ for all other cases are deferred to the appendix.
	The graph $H$ is then the union of $H_v$ for all $v\in V$; notably, a vertex $u$ from $U$ may appear in multiple subgraphs $H_v$, in which case we identify all occurrences of $u$ and merge them into one vertex.
	Moreover, for each $u\in U$, we set $L(u) = \{ 1 \}$, which ensures that $u$ is ``matched'' in every solution of~$\mathcal{J}$.
	
	To show that $\mathcal{I}$ and $\mathcal{J}$ are equivalent, it suffices to show the following for every $v \in V$:
	\begin{equation} \tag{$\star$}
	\parbox{0.9\textwidth}{
		A vertex set $S \subseteq N_G(v)$ is $\ell$-fair if and only if there is a spanning subgraph $H_v' = (W_v, F_v')$ of $H_v$ such that $\deg_{H_v'}(u) = 1$ for every $u \in S$, $\deg_{H_v'}(u) = 0$ for every $u \in N(v) \setminus S$, and $\deg_{H_v'}(v') \in L(v')$ for every $v' \in W_v \setminus U$.
	}
	\end{equation}
	To see why $(\star)$ is sufficient, assume that $(\star)$ holds true.
	For the forward direction, suppose that $\mathcal{I}$ is a yes-instance, i.e., there is a $\ell$-fair left-perfect many-to-one matching $M$ in $G$.
	Then, $M(v) \subseteq N_G(v)$ is $\ell$-fair for every $v \in V$.
	Hence, as we assume that $(\star)$ holds, we have a subgraph $H_v'$ for every $v \in V$ that satisfies the degree constraints of $(\star)$.
	Consider a spanning subgraph $H'$ whose edge set is the union of the edge set of $H_v'$ over all $v$.
	It is easy to verify that $H'$ constitutes a solution for $\mathcal{J}$.
	For the converse direction, suppose that $\mathcal{J}$ is a yes-instance, i.e., there is a spanning subgraph $H' = (W, F')$ of $H$ with $\deg(w) \in L(w)$ for every $w \in W$.
	Then, the subgraph of $H'$ induced by $W_v$ satisfies all the degree constraints of $(\star)$.
	For $v\in V$, let $S_v \subseteq N_G(v)$ be the set of vertices from $U$ that have a neighbor in $H'[W_v]$.
	By $(\star)$, $S_v$ is $\ell$-fair.
	Moreover, since $L(u) = \{ 1 \}$ for every $u \in U$, every vertex appears in $S_v$ for exactly one vertex $v \in V$.
	It follows that a matching $M$ with $M(v) = S_v$ for every $v \in V$ is a $\ell$-fair left-perfect many-to-one matching in $\mathcal{I}$. 
	We remark that we can adapt our algorithm to handle the non-emptiness constraint.
\end{proof}
\begin{figure}
	\centering
	\begin{subfigure}{.5\textwidth}
		\centering
		\resizebox{.6\textwidth}{!}{\begin{tikzpicture}[yscale=0.8]
			\node[square, label={[label distance=0.15cm]180:$u_1$}] (u1) at (0, 3) {};
			\node[square, label={[label distance=0.15cm]180:$u_2$}] (u2) at (0, 2) {};
			\node[triangle, label={[label distance=0.15cm]180:$u_3$}] (u3) at (0, 1) {};
			\node[itriangle, label={[label distance=0.15cm]0:$u_4$}] (u4) at (0, 0) {};

			\node[vertex, label={[label distance=0.00cm]95:$v_1 \colon \{ 1 \}$}] (v1) at (2, 2.5) {};
			\node[vertex, label={[label distance=0.00cm]95:$v_2 \colon \{ 1 \}$}] (v2) at (2, 1.5) {};
			\node[vertex, label={[label distance=0.00cm]95:$v_3 \colon \{ 1 \}$}] (v3) at (2, .5) {};
			\node[vertex, label={[label distance=0.25cm]270:$v_4 \colon \{ 0, 1, 3 \}$}] (w) at (3.5, 1.5) {};

			\draw (u1) -- (v1) -- (u2);
			\draw (u3) -- (v2) -- (w);
			\draw (u4) -- (v3) -- (w);
			\draw (w)-- (v1);
		\end{tikzpicture}}
		\caption{$H_v$ for $\ell = 0$ and $N(v)$ has two vertices of color $\alpha$, one vertex of color $\beta$, and one vertex of color $\gamma$.}
	\end{subfigure}
	\hfill
		\begin{subfigure}{.45\textwidth}
		\centering
	\resizebox{.6\textwidth}{!}{\begin{tikzpicture}[yscale=0.6]
		\node[triangle, label={[label distance=0.15cm]180:$u_1$}] (u1) at (0, 3) {};
		\node[triangle, label={[label distance=0.15cm]180:$u_2$}] (u2) at (0, 2) {};
		\node[itriangle, label={[label distance=0.15cm]0:$u_3$}] (u3) at (0, 1) {};
		\node[itriangle, label={[label distance=0.15cm]0:$u_4$}] (u4) at (0, 0) {};

		\node[vertex, label={[label distance=0.1cm]0:$v_1 \colon \{ 1, 2 \}$}] (v1) at (2, 2.3) {};
		\node[vertex, label={[label distance=0.1cm]0:$v_2 \colon \{ 1, 2 \}$}] (v2) at (2, .7) {};

		\draw (u1) -- (v1) -- (u2);
		\draw (u3) -- (v2) -- (u4);
		\draw (v1) -- (v2);
	\end{tikzpicture}}
	\caption{$H_v$ for $\ell = 1$ and $N(v)$ has two vertices of color $\alpha$ and two vertices of color $\beta$.}
	\end{subfigure}
	\caption{Exemplary constructions of $H_v$.}
	\label{fig:hv-ex}
\end{figure}
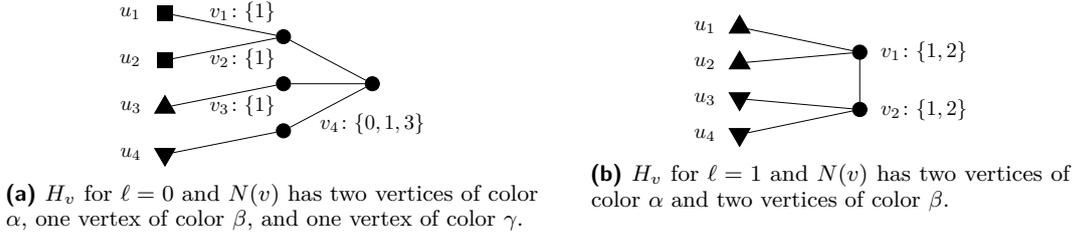

\subparagraph*{Hardness.}
It was already proven by Stoica et al.~\cite{DBLP:conf/atal/StoicaCDG20} that \textsc{MoV Fair Matching} is NP-hard.
However, the maximum degree is unbounded in their reduction.
We strengthen their result by showing that \textsc{MoV Fair Matching} is also NP-hard for three colors if the degree of vertices in $U$ is at most two and the degree of vertices in $V$ is at most five.

\begin{proposition}
	\label{prop:mov:hardness}
	\textsc{MoV Fair Matching} is NP-hard even for $\Delta_U = 2$ and $\Delta_V = 5$.
\end{proposition}
\begin{proof}
	We reduce from an NP-hard variant of \textsc{Satisfiability} where each clause consists of exactly three literals and each variable occurs in exactly two clauses positively and in exactly two clauses negatively \cite{DBLP:journals/eccc/ECCC-TR03-049}.
	An instance of \textsc{Satisfiability} consists of a set $X$ of variables and a set $Y$ of clauses.

	\begin{figure}
		\centering
		\begin{tikzpicture}[xscale=0.8, yscale=1.2]
			\node[triangle, label={[label distance=0.15cm]180:$u_x^{\ngt, \alpha}$}] (ux1) at (0, 0) {};
			\node[triangle, label={[label distance=0.15cm]180:$u_x$}] (ux2) at (0, 1) {};
			\node[triangle, label={[label distance=0.15cm]180:$u_x^{\ps, \alpha}$}] (ux3) at (0, 2) {};

			\node[vertex, label={[label distance=0.15cm]90:$v_x^\ps$}] (vx1) at (2, 1.4) {};
			\node[vertex, label={[label distance=0.15cm]270:$v_x^\ngt$}] (vx2) at (2, .6) {};

			\node[itriangle, label={[label distance=0.15cm]180:$u_x^{\ps, \beta}$}] (uy) at (4, 1.4) {};
			\node[itriangle, label={[label distance=0.15cm]180:$u_x^{\ngt, \beta}$}] (uy2) at (4, .6) {};

			\node[square, label={[label distance=0.15cm]0:$u_x^{\ps, 1}$}] (uxp1) at (4, 2.6) {};
			\node[square, label={[label distance=0.15cm]0:$u_x^{\ps, 2}$}] (uxp2) at (4, 2) {};
			\draw[dashed] (5.8, 3.1) -- (uxp1);
			\draw[dashed] (5.8, 2.5) -- (uxp2);

			\node[square, label={[label distance=0.15cm]0:$u_x^{\ngt, 1}$}] (uxn1) at (4, 0) {};
			\node[square, label={[label distance=0.15cm]0:$u_x^{\ngt, 2}$}] (uxn2) at (4, -.6) {};
			\draw[dashed] (5.8, .5) -- (uxn1);
			\draw[dashed] (5.8, -.1) -- (uxn2);

			\draw (ux3) -- (vx1) -- (ux2);
			\draw (ux2) -- (vx2) -- (ux1);
			\draw (uxp1) -- (vx1) -- (uxp2);
			\draw (uxn1) -- (vx2) -- (uxn2);

			\draw (uy) -- (vx1);
			\draw (uy2) -- (vx2);

		\end{tikzpicture}
		\hspace{1em}
		\begin{tikzpicture}[xscale=0.8, yscale=1.2]
			\node[vertex, label={[label distance=0.15cm]180:$v_y$}] (vy) at (0, 2) {};
			\node[vertex, label={[label distance=0.15cm]180:$v_y'$}] (vyp) at (0, 0) {};

			\node[triangle, label={[label distance=0.15cm]0:$u_y$}] (uy) at (2, 0) {};
			\node[triangle, label={[label distance=0.15cm]40:$u_y^2$}] (uy2) at (2, 1) {};
			\node[triangle, label={[label distance=0.15cm]40:$u_y^1$}] (uy1) at (2, 2) {};

			\node[vertex, label={[label distance=0.15cm]90:$v_y^1$}] (vy1) at (4, 1.7) {};
			\node[vertex, label={[label distance=0.15cm]270:$v_y^2$}] (vy2) at (4, .7) {};

			\node[itriangle, label={[label distance=0.15cm]180:$u_y^{1, \beta}$}] (uy11) at (6, 2) {};
			\node[square, label={[label distance=0.15cm]0:$u_y^{1, \gamma}$}] (uy12) at (6, 1.4) {};
			\node[itriangle, label={[label distance=0.15cm]180:$u_y^{2, \beta}$}] (uy21) at (6, 1) {};
			\node[square, label={[label distance=0.15cm]0:$u_y^{2, \gamma}$}] (uy22) at (6, 0.4) {};

			\draw (vy) -- (uy) -- (vyp);
			\draw (vy) -- (uy1) -- (vy1);
			\draw (vy) -- (uy2) -- (vy2);
			\draw (uy11) -- (vy1) -- (uy12);
			\draw (uy21) -- (vy2) -- (uy22);

			\draw[dashed] (vy) -- (-1, 2.6);
			\draw[dashed] (vy) -- (-1, 3);
			\draw[dashed] (vyp) -- (-1, -.6);
		\end{tikzpicture}
		\caption{A variable gadget (left) and clause gadget (right) from the proof of \Cref{prop:mov:hardness}. Vertices of color $\alpha$ are depicted as triangles, vertices of color $\beta$ as down-pointing triangles, and vertices of color $\gamma$ as squares.
		A dotted line denotes an edge between a variable gadget and a clause gadget.}
		\label{fig:mov:hardness}
	\end{figure}
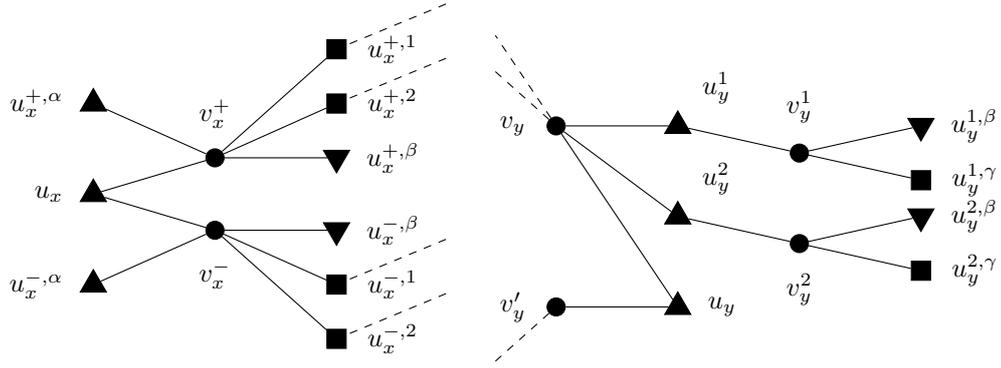

	\proofsubparagraph*{Construction.}
	We set $\ell=0$ and use three colors $\alpha$, $\beta$, and $\gamma$.
	We will construct a graph $G = (U \cupdot V, E)$ consisting of a variable gadget for each variable and a clause gadget for each clause.
	We start with describing the variable gadgets (see \Cref{fig:mov:hardness} for an illustration).
	For each variable $x$, we do the following:
	\begin{itemize}
		\item
		We introduce three vertices $u_x, u_x^{\ps, \alpha}, u_x^{\ngt, \alpha}\in U$ of color $\alpha$, two vertices $u_x^{\ps, \beta}, u_x^{\ngt, \beta}\in U$ of color $\beta$, four vertices $u_x^{\ps, 1}, u_x^{\ps, 2}, u_x^{\ngt, 1}, u_x^{\ngt, 2}\in U$ of color $\gamma$, and two vertices $v_x^\ps, v_x^\ngt \in V$.
		\item
		We add edges such that $v_x^{\ps}$ is adjacent to $u_x, u_x^{\ps, \alpha}, u_x^{\ps, \beta}, u_x^{\ps, 1}, u_x^{\ps, 2}$ and $v_x^{\ngt}$ is adjacent to $u_x, u_x^{\ngt, \alpha}, u_x^{\ngt, \beta}, u_x^{\ngt, 1}, u_x^{\ngt, 2}$.
	\end{itemize}
	Essentially, matching $u_x$ to $v_x^{\ps}$ (resp., $v_x^{\ngt}$) equates to assigning false (resp., true) to the variable $x$.
	Whenever $u_x$ is matched to $v_x^{\ps}$ (resp., $v_x^{\ngt}$), both $u_x^{\ps, 1}$ and $u_x^{\ps, 2}$ (resp., $u_x^{\ngt, 1}$ and $u_x^{\ngt, 2}$) must also be matched to $v_x^{\ps}$ (resp., $v_x^{\ngt}$) in order to respect the 0-fairness constraint.

	We next describe the clause gadgets (see \Cref{fig:mov:hardness} for an illustration).
	For each clause $y \in Y$, we do the following:
	\begin{itemize}
		\item We introduce three vertices $u_y, u_y^1, u_y^2\in U$ of color $\alpha$, two vertices $u_y^{1, \beta}, u_y^{2, \beta}\in U$ of color $\beta$, two vertices $u_y^{1, \gamma}$, $u_y^{2, \gamma}\in U$ of color $\gamma$, and four vertices $v_y, v_y', v_y^1, v_y^2 \in V$.
		\item We add edges such that $v_y$ is adjacent to $u_y, u_y^1, u_y^2$, such that $v_y'$ is adjacent to $u_y$, and such that $v_y^i$ is adjacent to $u_y^{i}, u_y^{i, \beta}, u_y^{i, \gamma}$ for each $i \in [2]$.
	\end{itemize}
	Note that $u_y$ must be matched to either $v_y$ or $v_y'$.
	Moreover, all the neighbors of $v_y$ and $v_y'$ in the clause gadget have color $\alpha$.
	To achieve $0$-fairness, it follows that $v_y$ or $v_y'$ must be matched to at least one vertex of color $\beta$ or $\gamma$ (in fact this will be always $\gamma$) from a variable gadget.

	Lastly, we connect variable gadgets and clause gadgets.
	For each $y \in Y$, let $l_y^1, l_y^2, l_y^3$ be the literals in $y$.
	For every $y \in Y$ and $i \in [3]$, we add an edge as follows:
	\begin{itemize}
		\item
		If $l_y^i = x$ and the literal $l_y^i$ is the $j$th occurrence ($j \in [2]$) of $x$, then add an edge between $u_x^{+, j}$ and $v_y$ (for $i=3$, in this case we add an edge between $u_x^{+, j}$ and $v_y'$).
		\item
		If $l_y^i = \overline{x}$ and the literal $l_y^i$ is the $j$th occurrence ($j \in [2]$) of $\overline{x}$, then add an edge between $u_x^{-, j}$ and $v_y$ (for $i=3$, in this case we add an edge between $u_x^{-, j}$ and $v_y'$).
	\end{itemize}
	This concludes the construction.
	It is easy to verify that $\Delta_U \le 2$ and $\Delta_V \le 5$.

	\proofsubparagraph*{Correctness.}
	($\Rightarrow$)
	Let $\varphi$ be a satisfying assignment of the given \textsc{Satisfiability} instance.
	We construct a matching $M$ in $G$ as follows.
	First, we add edges to $M$ such that
	\begin{align*}
		M(v_x^\ps) = \begin{cases}
			\{ u_x^{\ps, \alpha}, u_x^{\ps, \beta} \} & \text{ if $x$ is true in $\varphi$,} \\
			N(v_x^\ps) & \text{ if $x$ is false in $\varphi$,} \\
		\end{cases}
		\quad
		M(v_x^\ngt) = \begin{cases}
			N(v_x^\ngt) & \text{ if $x$ is true in $\varphi$,} \\
			\{ u_x^{\ngt, \alpha}, u_x^{\ngt,\beta} \} & \text{ if $x$ is false in $\varphi$} \\
		\end{cases}
	\end{align*}
	for every $x \in X$.
	Moreover, for every $x \in X$ and $i \in [2]$, we add to $M$ an edge between $u_{x}^{\ps, i}$ (resp., $u_x^{\ngt, i}$) and its neighbor in a clause gadget if $x$ is true (resp., false) in $\varphi$.
	Now every vertex from the left side occurring in a variable gadgets is matched.
	We now discuss how to match the vertices in clause gadgets.
	We add to $M$ edges $\{ v_y^i, u_y^{i, \beta} \}$ and $\{ v_y^i, u_y^{i, \gamma} \}$ for every $y \in Y$ and $i \in [2]$.
	Moreover, we do as follows for every $y \in Y$:
	\begin{itemize}
		\item
		If all literals are true in $\varphi$, then add $\{ v_y, u_y^1 \}$, $\{ v_y, u_y^2 \}$, and $\{ v_y', u_y \}$.
		\item
		If only $l_y^1$ and $l_y^2$ are true, then add $\{ v_y, u_y^1 \}$, $\{ v_y^2, u_y^2 \}$, and $\{ v_y, u_y \}$.
		\item
		If only $l_y^1$ (or $l_y^2$) and $l_y^3$ are true, then add $\{ v_y, u_y^1 \}$, $\{ v_y^2, u_y^2 \}$, and $\{ v_y', u_y \}$.
		\item
		If only $l_y^1$ (or $l_y^2$) is true, then add $\{ v_y^1, u_y^1 \}$, $\{ v_y^2, u_y^2 \}$, and $\{ v_y, u_y \}$.
		\item
		If only $l_y^3$ is true, then add $\{ v_y^1, u_y^1 \}$, $\{ v_y^2, u_y^2 \}$, and $\{ v_y', u_y \}$.
	\end{itemize}
	It is straightforward to verify that $M$ is a left-perfect $\ell$-fair matching.

	$(\Leftarrow)$
	Assume that there is a left-perfect $\ell$-fair matching in the constructed \textsc{MoV Fair Matching} instance.
	We claim that the assignment $\varphi$, where $x$ is set to true if $u_x \in M(v_x^{-})$ and to false if $u_x \in M(v_x^+)$, is a satisfying assignment for the given formula.
	Assume for a contradiction that a clause $y \in Y$ is not satisfied by $\varphi$.
	Note that $u_y \in M(v_y)$ or $u_y \in M(v_y')$.
	Since $M(v_y)$ and $M(v_y')$ are both $0$-fair, at least one of $M(v_y)$ and $M(v_y')$ contains a vertex of color $\gamma$ or $\beta$.
	Without loss of generality, assume that $u_x^{\ps, 1} \in M(v_y) \cup M(v_y')$.
	Then, we have $u_x \in M(v_x^-)$ since otherwise $M(v_x^+)$ is not 0-fair.
	This, however, implies that $x$ is true in $\varphi$, a contradiction.
\end{proof}

Notably this result also clearly shows hardness for arbitrary $\Delta_U\geq 2$ and $\Delta_V\geq 5$, as we can simply add an isolated trivial yes-instance with the respective vertex degrees to the constructed instance. 

\subsubsection{Max-Min} \label{sec:dich_mm}

For \textsc{Max-Min Fair Matching}, we prove that fewer cases are polynomial-time solvable than for \textsc{MoV Fair Matching}:
\begin{restatable}{theorem}{dicmm}
	\label{thm:dic:mm}
	\textsc{Max-Min Fair Matching} is polynomial-time solvable if $\Delta_U \le 1$, $\Delta_V \le 2$, or $(\Delta_U, \Delta_V) = (2, 3)$ (even with the non-emptiness constraint) and NP-hard~otherwise.
\end{restatable}
Again note that the case $\Delta_u\leq 1$ is trivial.

\subparagraph*{Algorithm.}

We first show the polynomial-time solvability for $\Delta_V \le 2$.
We then give a polynomial-time algorithm for $\Delta_U = 2$ and $\Delta_V = 3$.

\begin{proposition}
	\textsc{Max-Min Fair Matching} is polynomial-time solvable for $\Delta_V \leq 2$ even with the non-emptiness constraint.
\end{proposition}
\begin{proof}
	Let $\mathcal{I}=(G=(U \cupdot V,E),C,\col,\ell)$ be an instance of \textsc{Max-Min Fair Matching}.
	We make a case distinction based on the value of $\ell$.
	For each case, we first consider the case without the non-emptiness constraint.

	If $\ell\geq 2$, then $\mathcal{I}$ is a yes-instance if and only if every vertex from $U$ has at least one neighbor (as each vertex $v \in V$ has at most two neighbors, we can assign all vertices from $U$ arbitrarily).
	If the non-emptiness constraint is imposed, then $\mathcal{I}$ is a yes-instance if and only if there exists a left-perfect matching covering all vertices from $V$ in $G$. 

	Suppose that $\ell=0$.
	If $|C| \ge 3$, then $M(v)$ becomes $0$-fair only if $M(v) = \emptyset$ since $|M(v)| \le 2$.
	Thus, we may assume that $|C| = 2$:
	So we can solve \textsc{Max-Min Fair Matching} in polynomial time by \Cref{prop:dic:cpoly} for $\ell=0$ even with the non-emptiness constraint.

	For $\ell=1$, we follow a different approach: 
	Let $V'\subseteq V$ be the subset of vertices from $V$ that are adjacent to two vertices of the same color from $U$. 
	Then, we return yes if there is a left-perfect matching $M$ in $G$ with $|M(v)| \in \{ 0, 1 \}$ for all $v\in V'$ and $|M(v)|\in \{0, 1, 2\}$ for $v\in V\setminus V'$.
	Note that this reduces to a polynomial-time solvable case of \textsc{General Factor}. 
	The algorithm is correct because no vertex from $V'$ can be matched to two vertices because their two neighbors are of the same color which would result in a violation of $1$-fairness. 
	Vertices from $V\setminus V'$, on the other hand, can be matched to no, one, or both of their two neighbors because they are of different colors. 
	If the non-emptiness constraint is present, we exclude zero from the degree list of each vertex from $V$.
\end{proof}

\begin{proposition}
	\textsc{Max-Min Fair Matching} is polynomial-time solvable for $\Delta_U = 2$ and $\Delta_V = 3$ even with the non-emptiness constraint.
\end{proposition}
\begin{proof}
    We first consider the case without the non-emptiness constraint.
	Consider an instance $\mathcal{I} = (G=(U \cupdot V,E),C,\col,\ell)$ of \textsc{Max-Min Fair Matching}.
	Since \textsc{Max-Min Fair Matching} is polynomial-time solvable for $|C| = 2$ by \Cref{prop:dic:cpoly}, we assume that $|C| \ge 3$.
	We also assume that $\ell < 3$, since any matching in $G$ is $\ell$-fair for $\ell \ge 3$.

	Suppose that $\ell = 0$. 
	We may assume that $|C| = 3$, since $|C| \ge 4$ would imply that there is no $0$-fair matching.
	Thus, for each vertex $v\in V$ we have to either assign all vertices from $N(v)$ or no vertex from $N(v)$ to $v$.
	We further may assume that there is no vertex $u \in U$ with $\deg(u) \le 1$, since such vertices can be easily preprocessed---for the neighbor $v \in V$ of $u$, simply delete $v$ and $N(v)$ if $N(v)$ is 0-fair, otherwise we can conclude that the given instance is a no-instance.
	We may also assume that $N(v)$ is 0-fair for every $v \in V$, as $|C|\geq 3$ and $\Delta_V=3$ and otherwise we cannot assign any vertex to $v$ (so we delete $v$).
	Consider the graph $G_0$ whose vertex set is $V$ and whose edge set contains an edge $\{ v, v' \}$ if and only if there is a vertex $u \in U$ with $N(u) = \{ v, v' \}$ (note that all vertices from $U$ have degree two in $G$).
	We claim that $\mathcal{I}$ is a yes-instance if and only if $G_0$ is bipartite.

	If there is a $0$-fair matching $M$ in $G$, then the vertex set $S = \{ v \in V \mid M(v) \ne \emptyset \}$ is an independent set in $G_0$, as $|M(v)|=3$ for all $v\in V$.
	Moreover, at least one endpoint of every edge in $G_0$ is in $S$, since $M$ is left-perfect.
	It follows that $G_0$ is bipartite.
	Conversely, if $G_0$ is bipartite, then take one bipartition $(V_1', V_2')$ of $G_0$.
	Since every edge in $G_0$ is incident with one vertex in $V_1'$, we can construct a $0$-fair matching $M$ by setting $M(v_1) = N_{G}(v_1)$ for every $v_1 \in V_1'$ and $M(v) = \emptyset$ for every $v_2 \in V_2'$.

	Suppose that $\ell = 1$.
	We again construct an auxiliary graph $G_1$.
	First add $U$ to $G_1$.
	We do the following for every $v \in V$:
	\begin{itemize}
		\item
		If $N(v)$ contains at most one vertex of each color, then add a degree-one neighbor to every $u \in N(v)$.
		\item
		If every vertex in $N(v)$ has the same color, then add a vertex $v'$ to $G_1$ and add edges $\{ u, v' \}$ for every $u \in N(v)$.
		\item
		Otherwise, $N(v)$ consists of two vertices $u_1, u_2$ of color, say $\alpha$, and one vertex $u_3$ of color, say $\beta$.
		Then, add two vertices:
		one vertex adjacent to $u_1$ and $u_2$ and another adjacent to $u_3$.
	\end{itemize}
	It is straightforward to verify that $G_1$ has a one-to-one matching that covers every vertex in $U$ if and only if $\mathcal{I}$ has a $\ell$-fair left-perfect many-to-one matching (for the correctness of the last case note that we assume that $|C|\geq 3$).

	Finally, suppose that $\ell = 2$.
	We construct an auxiliary graph $G_2$.
	First add $U$ to $G_1$.
	We do the following for every $v \in V$:
	\begin{itemize}
		\item if $N(v)$ consists of three vertices of the same color, then add two vertices $v_1$ and $v_2$ and add edges such that both $v_1$ and $v_2$ are adjacent to every vertex of $N(v)$, and
		\item otherwise, add a degree-one neighbor to every vertex in $N(v)$.
	\end{itemize}
	It is straightforward to verify that $G_2$ has a one-to-one matching that covers every vertex in $U$ if and only if $\mathcal{I}$ has a $\ell$-fair left-perfect many-to-one matching.
	
	\medskip
	If we impose the non-emptiness constraint, then the algorithm needs to be adapted. 
	For $\ell=0$, we have a yes-instance if and only if each each vertex from $U$ has only one neighbor and $N(v)$ is $0$-fair for each $v\in V$. 
	For $\ell=1$, we reduce the problem to an instance of the polynomial-time solvable variant of \textsc{General Factor} introduced in \Cref{prop:dic:mov:algo}: 
    We modify the given graph $G$ and set the allowed degrees for every $v \in V$ as follows.
	\begin{itemize}
		\item
		If $N(v)$ contains at most one vertex of each color, then we set the allowed degrees of $v$ to $\{1,2,3\}$.
		\item
		If every vertex in $N(v)$ has the same color, then we set the allowed degrees of $v$ to $\{1\}$
		\item
		Otherwise, $N(v)$ consists of two vertices $u_1, u_2$ of color, say $\alpha$, and one vertex $u_3$ of color, say $\beta$.
		Then, we replace $v$ by three vertices $v_1$, $v_2$, and $v_3$. 
		$v_1$ is adjacent to $u_1$, $u_2$, $v_2$, and $v_3$, and $v_1$ has allowed degrees $\{2\}$. 
		$v_2$ is adjacent to $u_3$, $v_1$, and $v_3$, and $v_2$ has allowed degrees $\{2\}$. 
		$v_3$ is adjacent to $v_1$ and $v_2$ and has allowed degrees $\{0,1\}$. 
	\end{itemize}
	For $\ell=2$, we again reduce the problem to an instance of the polynomial-time solvable variant of \textsc{General Factor} introduced in \Cref{prop:dic:mov:algo} on the given graph $G$: 
    We set the allowed degrees of $v\in V$ to $\{1,2\}$ if  $N(v)$ consists of three vertices of the same color and to $\{1,2,3\}$ otherwise.
\end{proof}

\subparagraph*{Hardness.}
Note that we have already proven in \Cref{pr:hardnessMaxMin} that \textsc{Max-Min Fair Matching} is NP-hard for $\Delta_U \ge 3$ and $\Delta_V \ge 3$. Thus, to establish \Cref{thm:dic:mm}, it remains to prove the following:
\begin{proposition}
	\textsc{Max-Min Fair Matching} is NP-hard for $\Delta_U = 2$ and $\Delta_V = 4$.
\end{proposition}
\begin{proof}
	Our construction for this proposition is similar but more involved than the one for \Cref{pr:hardnessMaxMin}:

	\proofsubparagraph*{Construction.} We reduce from the variant of \textsc{3-Dimensional Matching} where each element occurs in at most three triples (as defined in the proof of  \Cref{pr:hardnessMaxMin}). 
	Given an instance $(X\cupdot Y \cupdot Z, T)$ of \textsc{3-Dimensional Matching}, we construct an instance of \textsc{Max-Min Fair Matching} as follows.
	Again, we set $\ell=0$ and introduce three colors $\alpha$, $\beta$, and $\gamma$.
	For each triple $t\in T$, we add a vertex $v_t$ to $V$ and for each element $s\in X\cupdot Y \cupdot Y$, we add three vertices $v_s^1$, $v_s^2$, and $v_s^3$ to $V$.
	Moreover, for each element $s\in X$, we do the following:
	\begin{itemize}
		\item We introduce three vertices $u^1_s$, $u^2_s$, and $u^3_s$ of color $\alpha$ to $U$ and connect each of them to a different vertex from $V$ corresponding to one of the three triples in which $s$ occurs.
		\item We connect $u^i_s$ to $v^i_s$ for each $i\in [3]$.
		\item We add two vertices $a^1_s$ and $a^2_s$ of color $\beta$ and two vertices $b^1_s$ and $b^2_s$ of color $\gamma$ to $U$.
		We connect $a_s^1$ to $v_s^1$ and $v_s^3$, $b_s^1$ to $v_s^2$ and $v_s^3$, and $a_s^2$ and $b_s^2$ to $v_s^1$ and $v_s^2$.
		Note that each of these four vertices has exactly two neighbors (see \Cref{fig:abv} for an illustration).
	\end{itemize}
	For each element $s\in Y$ (resp., $s\in Z$), we do exactly the same except that vertices $u^i_s$ have color $\beta$ (resp., $\gamma$), vertices $a^i_s$ have color $\gamma$ (resp., $\alpha$), and vertices $b^i_s$ have color $\alpha$ (resp., $\beta$).

	It is easy to see that $\deg(u)=2$ for all $u\in U$ and $\deg(v)\leq 4$ for all $v\in V$.
	The general idea of the reduction is that for each $s\in X\cupdot Y \cupdot Z$, exactly one of $u^1_s$, $u^2_s$, and $u^3_s$ cannot be matched to one of $v_s^1$, $v_s^2$, and $v_s^3$.
	This implies that exactly one of $u^1_s$, $u^2_s$, and $u^3_s$ must be matched to a vertex $v_t$ for some $t\in T$.
	As we show, this can be achieved in a 0-fair matching if and only if there is a perfect 3D-matching.

	\smallskip
	{\bfseries ($\Rightarrow$)}
	Assume that there is a perfect 3D-matching $T'\subseteq T$ in the given \textsc{3-Dimensional Matching} instance.
	We construct a $0$-fair left-perfect many-to-one matching $M$ as follows.
	For each $t\in T'$ and $s\in t$, we match the vertex $u_s^i$ (for some $i \in [3]$) to $v_t$ where $u_s^i$ is a vertex adjacent to $v_t$.
	Then, $M(v_t)$ is comprised of three vertices of distinct colors, and hence $0$-fair.
	Furthermore, we define $M(v_s^1)$, $M(v_s^2)$, and $M(v_s^3)$ for each $s\in X\cupdot Y \cupdot Z$ as follows.
	Note that exactly one of $u_s^1$, $u_s^2$, and $u_s^3$ is already matched to $v_t$ for some $t \in T'$.
	\begin{itemize}
		\item If $u_s^1$ is already matched, then let $M(v_s^1) = \emptyset$, $M(v_s^2) = \{ u_s^2, a_s^2, b_s^2 \}$, and $M(v_s^3) = \{ u_s^3, a_s^1, b_s^1\}$.
		\item If $u_s^2$ is already matched, then let $M(v_s^1) = \{ u_s^1, a_s^2, b_s^2\}$, $M(v_s^2) = \emptyset$, and $M(v_s^3) = \{ u_s^3, a_s^1, b_s^1\}$.
		\item If $u_s^3$ is already matched, then let $M(v_s^1) = \{ u_s^1, a_s^1, b_s^2\}$, $M(v_s^2) = \{ u_s^2, a_s^2, b_s^1 \}$, and $M(v_s^3) = \emptyset$.
	\end{itemize}
	It is easy to see that in every case, $M(v_s^1)$, $M(v_s^2)$, and $M(v_s^3)$ are $0$-fair.
	Thus, the constructed matching is $0$-fair and clearly a left-perfect many-to-one matching.

	\begin{figure}[t]
		\begin{center}
		\begin{tikzpicture}[yscale=1.3]
			\node[triangle, label={[label distance=0.15cm]270:$a_s^1$}] (a1) at (0, 0) {};
			\node[triangle, label={[label distance=0.15cm]270:$a_s^2$}] (a2) at (1, 0) {};
			\node[itriangle, label={[label distance=0.15cm]90:$b_s^1$}] (b1) at (2, 0) {};
			\node[itriangle, label={[label distance=0.15cm]90:$b_s^2$}] (b2) at (3, 0) {};

			\node[vertex, label={[label distance=0.15cm]90:$v_s^1$}] (v1) at (0, 1) {};
			\node[vertex, label={[label distance=0.15cm]90:$v_s^2$}] (v2) at (1.5, 1) {};
			\node[vertex, label={[label distance=0.15cm]90:$v_s^3$}] (v3) at (3, 1) {};

			\draw (v3) -- (a1) -- (v1);
			\draw (v2) -- (a2) -- (v1);
			\draw (v3) -- (b1) -- (v2);
			\draw (v2) -- (b2) -- (v1);
		\end{tikzpicture}
		\end{center}
		\caption{The connection between $v_s^1, v_s^2, v_s^3$ and $a_s^1, a_s^2, b_s^1, b_s^2$.}
		\label{fig:abv}
	\end{figure}
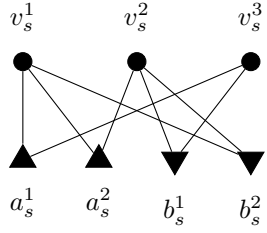

	{\bfseries ($\Leftarrow$)}
	Assume that there is a $0$-fair left-perfect many-to-one matching $M$ in the constructed \textsc{Max-Min Fair Matching} instance.
	We show that $\{t\in T\mid M(v_t)\neq \emptyset\}$ constitutes a perfect 3D-matching.
	As $M$ is 0-fair, for every triplet $t \in T$ with $M(v_t) \ne \emptyset$, every element $s \in t$ satisfies $u_s^i \in M(v_t)$ for some $i \in [3]$.
	It thus suffices to show that for every $s \in X \cupdot Y \cupdot Z$, exactly one of $u_s^1$, $u_s^2$, and $u_s^3$ is matched to $v_t$ for some $t \in T$ in $M$ with $s\in t$.
	Consider an element $s \in X\cupdot Y \cupdot Z$ and let $c$ (resp., $c'$) be the color of $u_s^1$, $u_s^2$, and $u_s^3$ (resp., $a_s^1$ and $a_s^2$).
	As $M$ is left-perfect, we have $a_s^1, a_s^2 \in M(v_s^1) \cup M(v_s^2) \cup M(v_s^3)$ (recall that $N(a_s^1), N(a_s^2) \subseteq \{ v_s^1, v_s^2, v_s^3 \}$).
	By construction, $N(\{ v_s^1, v_s^2, v_s^3 \})$ has no vertex of color $c'$ other than $a_s^1$ and $a_s^2$, and thus as $M$ is left-perfect and $a_s^1, a_s^2 \in M(v_s^1) \cup M(v_s^2) \cup M(v_s^3)$,  the set $M(v_s^1) \cup M(v_s^2) \cup M(v_s^3)$ contains exactly two vertices of color $c'$.
	Since $M$ is 0-fair, $M(v_s^1) \cup M(v_s^2) \cup M(v_s^3)$ thus contains exactly two vertices of color $c$.
	This implies that exactly one of $u_s^1, u_s^2, u_s^3$ is excluded from $M(v_s^1) \cup M(v_s^2) \cup M(v_s^3)$.
	Consequently, as $M$ is left-perfect, for every $s\in X\cupdot Y \cupdot Z$, it holds that there is exactly one $i^\star \in [3]$ such that $M(u_s^{i^\star})=v_{t}$ for some $t \in T$ with $s\in t$.
	Thus, we can conclude that $\{t\in T\mid M(v_t)\neq \emptyset\}$ is a perfect 3D-matching.
\end{proof}

This result also clearly shows hardness for arbitrary $\Delta_U \ge 2$ and $\Delta_V \ge 4$, as we can simply add an isolated trivial yes-instance with the respective vertex degrees to the constructed instance. 

\toappendix{

\section{Remainder of the Proof of \Cref{prop:dic:mov:algo}.}

\movdel*
\begin{proof}
	
Recall that we need to construct a subgraph $H_v$ for every $v \in V$ that satisfies the following:
\begin{equation*} \tag{$\star$}
\parbox{0.9\textwidth}{
A vertex set $S \subseteq N_G(v)$ is $\ell$-fair if and only if there is a spanning subgraph $H_v' = (W_v, F_v')$ of $H_v$ such that $\deg_{H_v'}(u) = 1$ for every $u \in S$, $\deg_{H_v'}(u) = 0$ for every $u \in N(v) \setminus S$, and $\deg_{H_v'}(v') \in L(v')$ for every $v' \in W_v' \setminus U$.
}
\end{equation*}

We now describe the construction of $H_v$ for $v\in V$ for different situations depending on $N_G(v)$ and $\ell$.
We omit the proofs of $(\star)$ for each $H_v$, since they are rather straightforward when $H_v$ is given and start by considering the case without the non-emptiness constraint.
For $U' \subseteq U$, let $\oc(U')$ denote the $|C|$-dimensional vector whose $i$th entry is the number of occurrences of the $i$th most frequently occurring color in $U'$, i.e., $\max^i_{c \in C} |U_c'|$.
To simplify notation, we omit trailing zero entries in $\oc(U')$.
We consider more than a dozen of cases depending on $\ell$ and $\oc(N(v))$ and start by discussing some easy cases:
If $N_G(v)$ consists of vertices of one color, then independent of $\ell$, $H_v$ consists of a vertex $v_1$ with $L(v_1) = \{ 0, \dots, \ell \}$ and all vertices from $N_G(v)$, where $v_1$ is adjacent to all vertices from $N_G(v)$.
So we may assume that $N_G(v)$ has at least two colors.
Moreover, if $\ell \ge \oc(N_G(v))_1$, then any subset  $S \subseteq N_G(v)$ is $\ell$-fair.
In this case, $H_v$ consists of a vertex $v_1$ with $L(v_1) = \{ 0, \dots, \deg_G(v) \}$ and all vertices from $N_G(v)$, where $v_1$ is adjacent to all vertices from $N_G(v)$.
So we may also assume that $\ell < \oc(N(v))_1$.
From  $\ell < \oc(N(v))_1$, $\Delta_V\leq 4$, and that $N_G(v)$ has at least two colors for each vertex $v\in V$, it follows that  $\ell \le 2$ needs to hold.
For all remaining cases the construction of $H_v$ is given in \Cref{fig:dic:ell2} (for $\ell = 2$), \Cref{fig:dic:ell1} (for $\ell = 1$), and \Cref{fig:dic:ell0} (for $\ell = 0$). 

If we impose the non-emptiness constraint, i.e., in the matching $M$ to be found $M(v)\neq \emptyset$ for each $v\in V$ holds, then we need to slightly modify the above presented constructions for $H_v$. 
For $\ell=0$ and $\ell=2$ these modifications are straightforward. 
For $\ell=0$, for instance, for $\oc(N(v)) = (2, 1, 1)$ (\Cref{fig:dic:ell0:211}), we need to change the set of allowed degrees for $v_4$ to $\{0,1\}$ (thereby forcing that two of $v_1$, $v_2$ and $v_3$ are matched to a vertex from $U$), and for $\oc(N(v)) = (2, 2)$ (\Cref{fig:dic:ell0:22}), we need to delete the edge between $v_1$ and $v_2$ (thereby forcing that both $v_1$ and $v_2$ are matched to vertices from $U$). 
For $\ell=2$, we need to adapt the construction presented in \Cref{fig:dic:ell2}  by changing the set of allowed degrees of $v_1$ to $\{0,2,3\}$.
For $\ell=1$ the modifications are more intricate.
Thus, we explicitly present the constructions for this case in \Cref{fig:dic:ell1:nonempty}.
\end{proof}

\begin{figure}
	\captionsetup[subfigure]{justification=centering}
	\centering
	\begin{subfigure}{.34\textwidth}
	\begin{tikzpicture}
		\node[triangle, label={[label distance=0.15cm]180:$u_1$}] (u1) at (0, 3) {};
		\node[triangle, label={[label distance=0.15cm]180:$u_2$}] (u2) at (0, 2) {};
		\node[triangle, label={[label distance=0.15cm]180:$u_3$}] (u3) at (0, 1) {};
		\node[itriangle, label={[label distance=0.15cm]0:$u_4$}] (u4) at (0, 0) {};

		\node[vertex, label={[label distance=0.15cm]0:$v_1 \colon \{ 0, 1, 2, 3 \}$}] (v1) at (2, 2.3) {};
		\node[vertex, label={[label distance=0.15cm]0:$v_2 \colon \{ 1 \}$}] (v2) at (2, .7) {};

		\draw (u1) -- (v1);
		\draw (u2) -- (v1);
		\draw (u3) -- (v1);
		\draw (u4) -- (v2);
		\draw (v1) -- (v2);
	\end{tikzpicture}
	\caption{$\oc(N(v)) = (3, 1)$.}
	\end{subfigure}
	\caption{The construction of $H_v$ for the case that $\ell = 2$.}
	\label{fig:dic:ell2}
\end{figure}
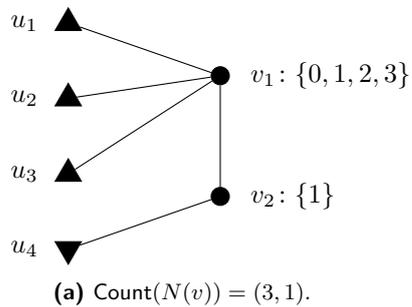

\begin{figure}
	\captionsetup[subfigure]{justification=centering}
	\centering
	\begin{subfigure}{.3\textwidth}
		\centering
		\begin{tikzpicture}
			\node[triangle, label={[label distance=0.15cm]180:$u_1$}] (u1) at (0, 1) {};
			\node[itriangle, label={[label distance=0.15cm]0:$u_2$}] (u2) at (0, 0) {};
			\node[vertex, label={[label distance=0.15cm]270:$v_1 \colon \{ 0, 2 \}$}] (v1) at (2, .5) {};
			\draw (u1) -- (v1) -- (u2);
		\end{tikzpicture}
		\caption{$\oc(N(v)) = (1, 1)$.}
	\end{subfigure}
	\begin{subfigure}{.3\textwidth}
		\centering
		\begin{tikzpicture}
			\node[triangle, label={[label distance=0.15cm]180:$u_1$}] (u1) at (0, 2) {};
			\node[triangle, label={[label distance=0.15cm]180:$u_2$}] (u2) at (0, 1) {};
			\node[itriangle, label={[label distance=0.15cm]0:$u_3$}] (u3) at (0, 0) {};

			\node[vertex, label={[label distance=0.15cm]90:$v_1 \colon \{ 1 \}$}] (v1) at (2, 1.6) {};
			\node[vertex, label={[label distance=0.15cm]270:$v_2 \colon \{ 1 \}$}] (v2) at (2, .4) {};

			\draw (u1) -- (v1) -- (u2);
			\draw (v1) -- (v2) -- (u3);
		\end{tikzpicture}
		\caption{$\oc(N(v)) = (2, 1)$.}
	\end{subfigure}
	\begin{subfigure}{.35\textwidth}
		\centering
		\begin{tikzpicture}
			\node[triangle, label={[label distance=0.15cm]180:$u_1$}] (u1) at (0, 2) {};
			\node[itriangle, label={[label distance=0.15cm]0:$u_2$}] (u2) at (0, 1) {};
			\node[square, label={[label distance=0.15cm]180:$u_3$}] (u3) at (0, 0) {};

			\node[vertex, label={[label distance=0.15cm]270:$v_1 \colon \{ 0, 2, 3 \}$}] (v) at (2, 1) {};

			\draw (v) -- (u1);
			\draw (v) -- (u2);
			\draw (v) -- (u3);
		\end{tikzpicture}
		\caption{$\oc((v)) = (1, 1, 1)$}
	\end{subfigure}

	\vspace{5ex}

	\begin{subfigure}{.45\textwidth}
		\centering
		\begin{tikzpicture}
			\node[triangle, label={[label distance=0.15cm]180:$u_1$}] (u1) at (0, 3) {};
			\node[triangle, label={[label distance=0.15cm]180:$u_2$}] (u2) at (0, 2) {};
			\node[triangle, label={[label distance=0.15cm]180:$u_3$}] (u3) at (0, 1) {};
			\node[itriangle, label={[label distance=0.15cm]0:$u_4$}] (u4) at (0, 0) {};

			\node[vertex, label={[label distance=0.15cm]0:$v_2 \colon \{ 1 \}$}] (v1) at (2, .7) {};
			\node[vertex, label={[label distance=0.15cm]0:$v_1 \colon \{ 1 \}$}] (v2) at (2, 2.3) {};

			\draw (u1) -- (v2);
			\draw (u2) -- (v2);
			\draw (u3) -- (v2);
			\draw (u4) -- (v1);
			\draw (v1) -- (v2);
		\end{tikzpicture}
		\caption{$\oc(N(v)) = (3, 1)$}
	\end{subfigure}
	\begin{subfigure}{.45\textwidth}
		\centering
		\begin{tikzpicture}
			\node[triangle, label={[label distance=0.15cm]180:$u_1$}] (u1) at (0, 3) {};
			\node[triangle, label={[label distance=0.15cm]180:$u_2$}] (u2) at (0, 2) {};
			\node[itriangle, label={[label distance=0.15cm]0:$u_3$}] (u3) at (0, 1) {};
			\node[itriangle, label={[label distance=0.15cm]0:$u_4$}] (u4) at (0, 0) {};

			\node[vertex, label={[label distance=0.1cm]270:$v_2 \colon \{ 2 \}$}] (v1) at (2, .7) {};
			\node[vertex, label={[label distance=0.1cm]90:$v_1 \colon \{ 2 \}$}] (v2) at (2, 2.3) {};
			\node[vertex, label={[label distance=0.1cm]270:$v_2' \colon \{ 1 \}$}] (w1) at (3.5, .7) {};
			\node[vertex, label={[label distance=0.1cm]90:$v_1' \colon \{ 1 \}$}] (w2) at (3.5, 2.3) {};

			\draw (u1) -- (v2) -- (u2);
			\draw (u3) -- (v1) -- (u4);
			\draw (v1) -- (v2);
			\draw (v1) -- (w1) -- (w2) -- (v2);
		\end{tikzpicture}
		\caption{$\oc(N(v)) =(2, 2)$}  \label{fig:dic:ell0:22}
	\end{subfigure} 

	\vspace{5ex}

	\begin{subfigure}{.45\textwidth}
		\begin{tikzpicture}
			\node[square, label={[label distance=0.15cm]180:$u_1$}] (u1) at (0, 3) {};
			\node[square, label={[label distance=0.15cm]180:$u_2$}] (u2) at (0, 2) {};
			\node[triangle, label={[label distance=0.15cm]180:$u_3$}] (u3) at (0, 1) {};
			\node[itriangle, label={[label distance=0.15cm]0:$u_4$}] (u4) at (0, 0) {};

			\node[vertex, label={[label distance=0.25cm]0:$v_1 \colon \{ 1 \}$}] (v1) at (2, 2.5) {};
			\node[vertex, label={[label distance=0.02cm]95:$v_2 \colon \{ 1 \}$}] (v2) at (2, 1.5) {};
			\node[vertex, label={[label distance=0.02cm]95:$v_3 \colon \{ 1 \}$}] (v3) at (2, .5) {};
			\node[vertex, label={[label distance=0.25cm]270:$v_4 \colon \{ 0, 1, 3 \}$}] (w) at (3.5, 1.5) {};

			\draw (u1) -- (v1) -- (u2);
			\draw (u3) -- (v2) -- (w);
			\draw (u4) -- (v3) -- (w);
			\draw (w)-- (v1);
		\end{tikzpicture}
		\caption{$\oc(N(v)) = (2, 1, 1)$} \label{fig:dic:ell0:211}
	\end{subfigure}
	\begin{subfigure}{.45\textwidth}
		\begin{tikzpicture}
			\node[triangle, label={[label distance=0.15cm]180:$u_1$}] (u1) at (0, 3) {};
			\node[itriangle, label={[label distance=0.15cm]0:$u_2$}] (u2) at (0, 2) {};
			\node[square, label={[label distance=0.15cm]180:$u_3$}] (u3) at (0, 1) {};
			\node[star4, label={[label distance=0.15cm]180:$u_4$}] (u4) at (0, 0) {};

			\node[vertex, label={[label distance=0.15cm]290:$v_1 \colon \{ 0, 2, 3, 4 \}$}] (v) at (2, 1.5) {};

			\draw (u1) -- (v) -- (u2);
			\draw (u3) -- (v) -- (u4);
		\end{tikzpicture}
		\caption{$\oc(N(v)) = (1, 1, 1, 1)$}
	\end{subfigure}
	\caption{The construction of $H_v$ for the case that $\ell = 0$.}
	\label{fig:dic:ell0}
\end{figure}
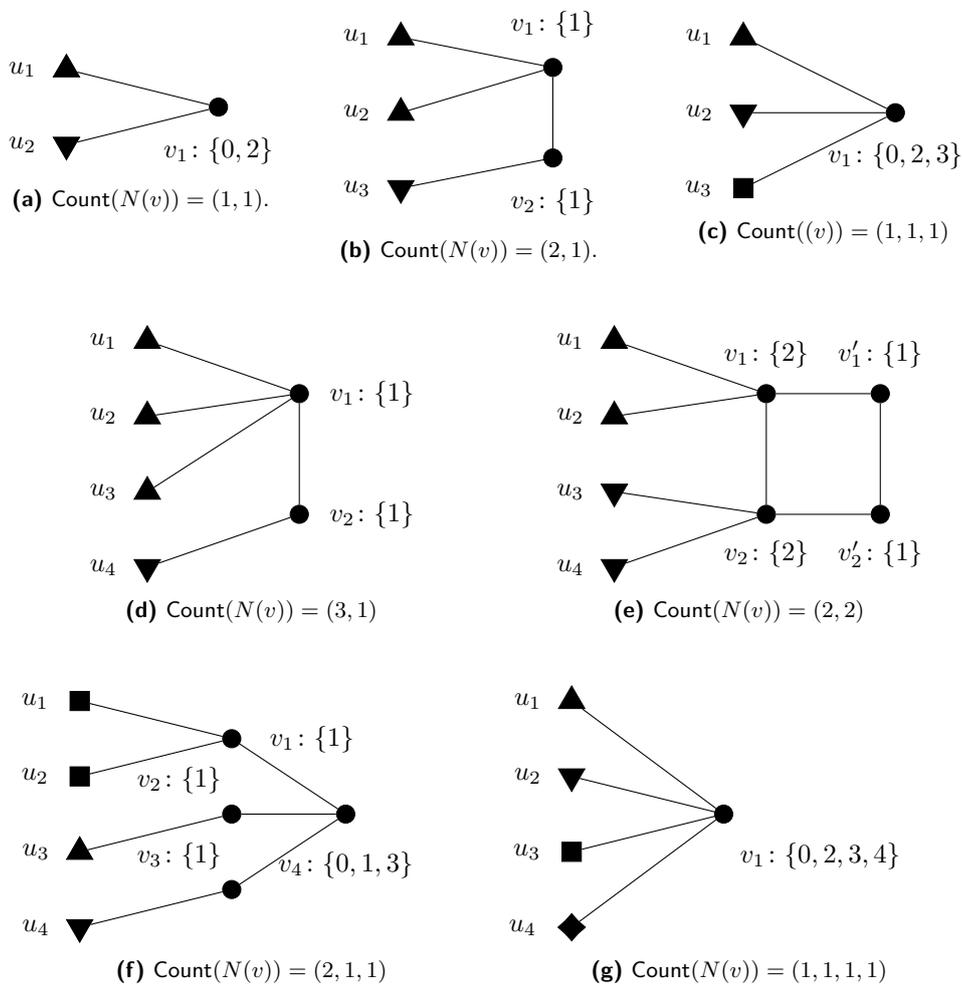

\begin{figure}
	\captionsetup[subfigure]{justification=centering}
	\centering
	\begin{subfigure}{.4\textwidth}
		\centering
		\begin{tikzpicture}
			\node[triangle, label={[label distance=0.15cm]180:$u_1$}] (u1) at (0, 2) {};
			\node[triangle, label={[label distance=0.15cm]180:$u_2$}] (u2) at (0, 1) {};
			\node[itriangle, label={[label distance=0.15cm]0:$u_3$}] (u3) at (0, 0) {};

			\node[vertex, label={[label distance=0.15cm]90:$v_1 \colon \{ 0, 1, 2 \}$}] (v1) at (2, 1.6) {};
			\node[vertex, label={[label distance=0.15cm]270:$v_2 \colon \{ 1 \}$}] (v2) at (2, .4) {};

			\draw (u1) -- (v1) -- (u2);
			\draw (v1) -- (v2) -- (u3);
		\end{tikzpicture}
		\caption{$\oc(N(v)) = (2, 1)$.}
	\end{subfigure}
	\begin{subfigure}{.4\textwidth}
	\begin{tikzpicture}
		\node[triangle, label={[label distance=0.15cm]180:$u_1$}] (u1) at (0, 3) {};
		\node[triangle, label={[label distance=0.15cm]180:$u_2$}] (u2) at (0, 2) {};
		\node[triangle, label={[label distance=0.15cm]180:$u_3$}] (u3) at (0, 1) {};
		\node[itriangle, label={[label distance=0.15cm]0:$u_4$}] (u4) at (0, 0) {};

		\node[vertex, label={[label distance=0.15cm]0:$v_1 \colon \{ 0, 1, 2 \}$}] (v1) at (2, 2.3) {};
		\node[vertex, label={[label distance=0.15cm]0:$v_2 \colon \{ 1 \}$}] (v2) at (2, .7) {};

		\draw (u1) -- (v1);
		\draw (u2) -- (v1);
		\draw (u3) -- (v1);
		\draw (u4) -- (v2);
		\draw (v1) -- (v2);
	\end{tikzpicture}
	\caption{$\oc(N(v)) = (3, 1)$.}
	\end{subfigure}

	\vspace{5ex}

	\begin{subfigure}{.4\textwidth}
	\begin{tikzpicture}
		\node[triangle, label={[label distance=0.15cm]180:$u_1$}] (u1) at (0, 3) {};
		\node[triangle, label={[label distance=0.15cm]180:$u_2$}] (u2) at (0, 2) {};
		\node[itriangle, label={[label distance=0.15cm]0:$u_3$}] (u3) at (0, 1) {};
		\node[itriangle, label={[label distance=0.15cm]0:$u_4$}] (u4) at (0, 0) {};

		\node[vertex, label={[label distance=0.1cm]0:$v_1 \colon \{ 1, 2 \}$}] (v1) at (2, 2.3) {};
		\node[vertex, label={[label distance=0.1cm]0:$v_2 \colon \{ 1, 2 \}$}] (v2) at (2, .7) {};

		\draw (u1) -- (v1) -- (u2);
		\draw (u3) -- (v2) -- (u4);
		\draw (v1) -- (v2);
	\end{tikzpicture}
	\caption{$\oc(N(v)) =(2, 2).$}
	\end{subfigure}
	\begin{subfigure}{.4\textwidth}
	\begin{tikzpicture}
		\node[square, label={[label distance=0.15cm]180:$u_1$}] (u1) at (0, 3) {};
		\node[square, label={[label distance=0.15cm]180:$u_2$}] (u2) at (0, 2) {};
		\node[triangle, label={[label distance=0.15cm]180:$u_3$}] (u3) at (0, 1) {};
		\node[itriangle, label={[label distance=0.15cm]0:$u_4$}] (u4) at (0, 0) {};

		\node[vertex, label={[label distance=0.1cm]0:$v_1 \colon \{ 0, 1, 2 \}$}] (v1) at (2, 2.3) {};
		\node[vertex, label={[label distance=0.1cm]0:$v_2 \colon \{ 1, 2 \}$}] (v2) at (2, .7) {};

		\draw (u1) -- (v1) -- (u2);
		\draw (u3) -- (v2) -- (u4);
		\draw (v1) -- (v2);
	\end{tikzpicture}
	\caption{$\oc(N(v)) =(2, 1, 1).$}
	\end{subfigure}

	\caption{The construction of $H_v$ for the case that $\ell = 1$.}
	\label{fig:dic:ell1}
\end{figure}
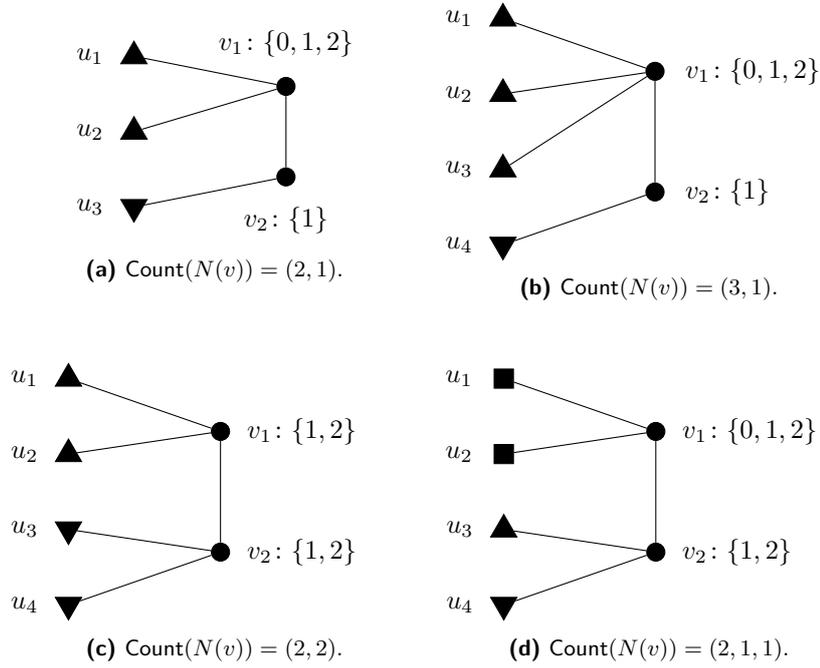

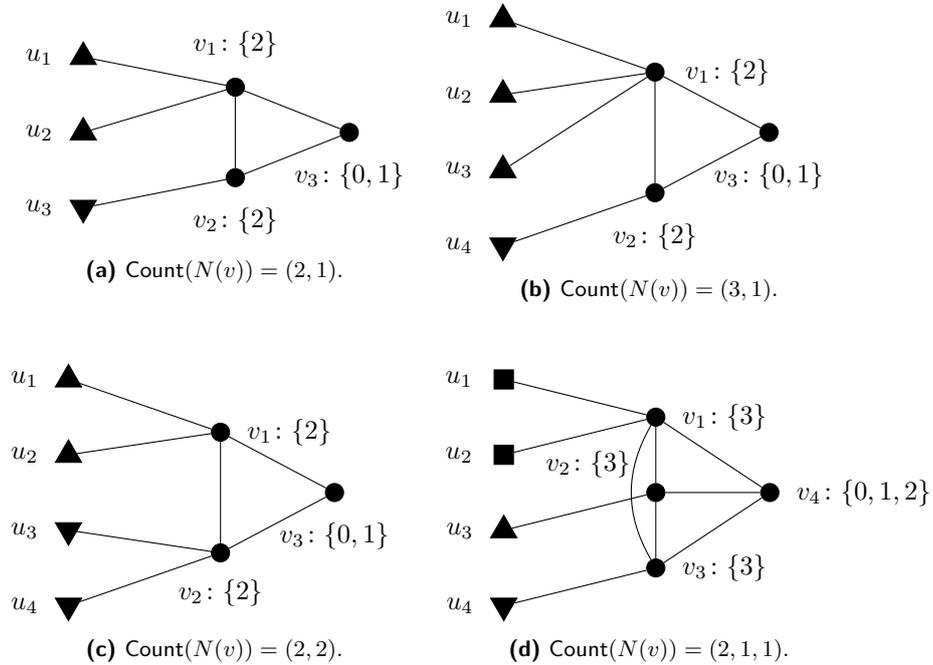
\begin{figure}
	\captionsetup[subfigure]{justification=centering}
	\centering
	\begin{subfigure}{.4\textwidth}
		\centering
		\begin{tikzpicture}
			\node[triangle, label={[label distance=0.15cm]180:$u_1$}] (u1) at (0, 2) {};
			\node[triangle, label={[label distance=0.15cm]180:$u_2$}] (u2) at (0, 1) {};
			\node[itriangle, label={[label distance=0.15cm]0:$u_3$}] (u3) at (0, 0) {};

			\node[vertex, label={[label distance=0.15cm]90:$v_1 \colon \{  2 \}$}] (v1) at (2, 1.6) {};
			\node[vertex, label={[label distance=0.15cm]270:$v_2 \colon \{ 2 \}$}] (v2) at (2, .4) {};
			\node[vertex, label={[label distance=0.15cm]270:$v_3 \colon \{ 0,1 \}$}] (v3) at (3.5, 1) {};

			\draw (u1) -- (v1) -- (u2);
			\draw (v1) -- (v2) -- (u3);
			\draw (v1) -- (v3);
			\draw (v2) -- (v3);
		\end{tikzpicture}
		\caption{$\oc(N(v)) = (2, 1)$.}
	\end{subfigure}
	\begin{subfigure}{.4\textwidth}
	\begin{tikzpicture}
		\node[triangle, label={[label distance=0.15cm]180:$u_1$}] (u1) at (0, 3) {};
		\node[triangle, label={[label distance=0.15cm]180:$u_2$}] (u2) at (0, 2) {};
		\node[triangle, label={[label distance=0.15cm]180:$u_3$}] (u3) at (0, 1) {};
		\node[itriangle, label={[label distance=0.15cm]0:$u_4$}] (u4) at (0, 0) {};

		\node[vertex, label={[label distance=0.15cm]0:$v_1 \colon \{ 2 \}$}] (v1) at (2, 2.3) {};
		\node[vertex, label={[label distance=0.15cm]270:$v_2 \colon \{ 2 \}$}] (v2) at (2, .7) {};
		\node[vertex, label={[label distance=0.15cm]270:$v_3 \colon \{ 0,1 \}$}] (v3) at (3.5, 1.5) {};

		\draw (u1) -- (v1);
		\draw (u2) -- (v1);
		\draw (u3) -- (v1);
		\draw (u4) -- (v2);
		\draw (v1) -- (v2);
		\draw (v1) -- (v3);
        \draw (v2) -- (v3);
	\end{tikzpicture}
	\caption{$\oc(N(v)) = (3, 1)$.}
	\end{subfigure}

	\vspace{5ex}

	\begin{subfigure}{.4\textwidth}
	\begin{tikzpicture}
		\node[triangle, label={[label distance=0.15cm]180:$u_1$}] (u1) at (0, 3) {};
		\node[triangle, label={[label distance=0.15cm]180:$u_2$}] (u2) at (0, 2) {};
		\node[itriangle, label={[label distance=0.15cm]0:$u_3$}] (u3) at (0, 1) {};
		\node[itriangle, label={[label distance=0.15cm]0:$u_4$}] (u4) at (0, 0) {};

		\node[vertex, label={[label distance=0.1cm]0:$v_1 \colon \{ 2 \}$}] (v1) at (2, 2.3) {};
		\node[vertex, label={[label distance=0.1cm]-90:$v_2 \colon \{ 2 \}$}] (v2) at (2, .7) {};
		\node[vertex, label={[label distance=0.15cm]270:$v_3 \colon \{ 0,1 \}$}] (v3) at (3.5, 1.5) {};

		\draw (u1) -- (v1) -- (u2);
		\draw (u3) -- (v2) -- (u4);
		\draw (v1) -- (v2);
		\draw (v1) -- (v3);
        \draw (v2) -- (v3);
	\end{tikzpicture}
	\caption{$\oc(N(v)) =(2, 2).$}
	\end{subfigure}
	\begin{subfigure}{.4\textwidth}
	\begin{tikzpicture}
		\node[square, label={[label distance=0.15cm]180:$u_1$}] (u1) at (0, 3) {};
		\node[square, label={[label distance=0.15cm]180:$u_2$}] (u2) at (0, 2) {};
		\node[triangle, label={[label distance=0.15cm]180:$u_3$}] (u3) at (0, 1) {};
		\node[itriangle, label={[label distance=0.15cm]0:$u_4$}] (u4) at (0, 0) {};

		\node[vertex, label={[label distance=0.1cm]0:$v_1 \colon \{ 3 \}$}] (v1) at (2, 2.5) {};
		\node[vertex, label={[label distance=0.1cm]160:$v_2 \colon \{ 3 \}$}] (v2) at (2, 1.5) {};
		\node[vertex, label={[label distance=0.1cm]0:$v_3 \colon \{ 3 \}$}] (v3) at (2, 0.5) {};
		\node[vertex, label={[label distance=0.1cm]0:$v_4 \colon \{ 0, 1, 2 \}$}] (v4) at (3.5, 1.5) {};

		\draw (u1) -- (v1) -- (u2);
		\draw (u3) -- (v2); 
		\draw (u4) -- (v3); 
		\draw (v1) -- (v2);
		\draw (v2) -- (v3);
		\path (v1) edge[bend right=30] node [right] {} (v3);
		\draw (v1) -- (v4);
		\draw (v2) -- (v4);
		\draw (v3) -- (v4);
	\end{tikzpicture}
	\caption{$\oc(N(v)) =(2, 1, 1).$}
	\end{subfigure}

	\caption{The construction of $H_v$ for the case that $\ell = 1$ with non-emptiness constraint.}
	\label{fig:dic:ell1:nonempty}
\end{figure}

}

\section{Fair Matching on Complete Bipartite Graphs} \label{se:cliques}
A natural special case of \textsc{Fair Matching} is when the underlying graph is complete, i.e., each vertex in $U$ can be assigned to any vertex in $V$.
This special case is also among the three problems introduced by Stoica
et al.\ \cite{DBLP:conf/atal/StoicaCDG20} (they called it \textsc{Fair Regrouping\_X}).
Stoica~et~al.~\cite{DBLP:conf/atal/StoicaCDG20} presented a straightforward
XP algorithm for \textsc{MoV Fair Matching} on a complete graph with size constraints parameterized by $|V|$ but left open the classical complexity.
We partially settle this open question by proving that \textsc{MoV Fair Matching} on complete bipartite graphs is polynomial-time solvable even with the non-emptiness constraint. 
In fact, we find a precise characterization of yes-instances, which turns out to be surprisingly simple.
However, it requires an intricate analysis to prove this, especially when the non-emptiness constraint is present. Afterwards, we prove an analogous result for \textsc{Max-Min Fair matching}.

To simplify notation, we assume that $C = \{ c_1, \dots, c_{|C|} \}$ and that $|U_{c_i}|\geq |U_{c_{i+1}}|$ for each $i\in [1,|C|-1]$ and set $|U_{c_i}| := 0$ for $i > |C|$.

\begin{restatable}{theorem}{clique}\label{th:clique}
	A \textsc{MoV Fair Matching} instance $\mathcal{I}=(G=(U \cupdot V,E),C,\col,\ell)$ with $G$ being a complete bipartite graph is a yes-instance if and only if $|U_{c_1}|\leq \ell k + \sum_{i\in [k]} |U_{c_{i+1}}|$.
	With the non-emptiness constraint, $\mathcal{I}$ is a yes-instance if and only if it additionally satisfies: $\ell > 0$ and $n\ge k$, or $\ell=0$ and $n \ge 2k$.
\end{restatable}
\begin{proof}
	We first prove the theorem without the non-emptiness constraint and later extend the proof. 
	We assume that $V = \{ v_1, \dots, v_k \}$.
	
	\proofsubparagraph*{\textbf{No size constraints}.}
	We consider both directions separately and start by proving that if there is a left-perfect $\ell$-fair many-to-one matching $M$ in $\mathcal{I}$, then $|U_{c_1}|\leq \ell k + \sum_{i\in [k]} |U_{c_{i+1}}|$ needs to be satisfied:
	For each $i \in [k]$, let $c_i'$ denote the most frequent color in $M(v_i)$ among $C \setminus \{ c_1 \}$.
	Since $M(v_i)$ is $\ell$-fair for each $i \in [k]$, we then have $|U_{c_1}| = \sum_{i \in [k]} |M(v_i)_{c_1}| \le \sum_{i \in [k]} (|M(v_i)_{c_i'}| + \ell) \le \ell k + \sum_{i\in [k]} |U_{c_{i+1}}|$.

	We now prove the reverse direction: 
	If $\mathcal{I}$ fulfills $|U_{c_1}|\leq \ell k + \sum_{i\in [k]} |U_{c_{i+1}}|$, then there is a left-perfect $\ell$-fair many-to-one matching $M$ in $\mathcal{I}$. To do so, we assume that $\mathcal{I}$ satisfies this constraint and present an algorithm that constructs $M$. 

	We refer to the following algorithm as Algorithm 1.
	Let $U' := U$ and let $i := 1$. 
	We do the following as long as there are at least $|U'_{c_{i + 1}}| + \ell$ vertices of color $c_1$ in $U'$:
	\begin{itemize}
	\item
	Delete $|U'_{c_{i + 1}}|+\ell$ vertices of color $c_1$ from $U'$ and match them to $v_i$.
	\item
	Delete $|U'_{c_{i + 1}}|$ vertices of color $c_{i + 1}$ from $U'$ and match them to $v_i$.
	\item
	Increment $i$ by 1.
	\end{itemize}
	If after this, there are vertices of color~$c_1$ left in~$U'$, then we match $\min(|U'_{c_{i+1}}|,|U'_{c_1}|)$ vertices of color~$c_{i+1}$ and $|U'_{c_1}|$~vertices of color~$c_1$ from $U'$ to $v_i$. 
	Note that as the constraint $|U_{c_1}|\leq \ell k + \sum_{i\in [k]} |U_{c_{i+1}}|$ is fulfilled, Algorithm 1 terminates for $i\leq k$ and thus every vertex deleted from $U'$ is matched to some vertex from $V$.
	Finally, we match the remaining vertices in $U'$ to $v_1$. Note that $M(v_1)$ remains $\ell$-fair:
	$U'$ only contains vertices of color $c_{j}$ for $j\ge i + 1$ and $|U_{c_{j}}| \le |U_{c_2}|$ by our assumption ($U_{c_2}\subseteq M(v_1)$ by construction).
	Thus, $M$ is $\ell$-fair and we thus have constructed a left-perfect $\ell$-fair many-to-one matching.
	
	\proofsubparagraph*{\textbf{Non-emptiness constraint}.}
	Assume that there is a left-perfect $\ell$-fair many-to-one matching $M$ for $\mathcal{I}$ with $M(v)\neq \emptyset$ for each $v\in V$. 
	
	We start by proving the forward direction. 
	Since we have $M(v) \ne \emptyset$ for each $v\in V$, $n\geq k$ holds.
	For $\ell = 0$, we actually have $|M(v)| \ge 2$ and hence we have $n \ge 2k$.
	Applying again the same argument as for the first part without the non-emptiness constraint, 
	$|U_{c_1}|\leq \ell k + \sum_{i\in [k]} |U_{c_{i+1}}|$ needs to be satisfied:	
	For each $i \in [k]$, let $c_i'$ denote the most frequent color in $M(v_i)$ among $C \setminus \{ c_1 \}$.
	Since $M(v_i)$ is $\ell$-fair for each $i \in [k]$, we then have $|U_{c_1}| = \sum_{i \in [k]} |M(v_i)_{c_1}| \le \sum_{i \in [k]} (|M(v_i)_{c_i'}| + \ell) \le \ell k + \sum_{i\in [k]} |U_{c_{i+1}}|$.  
	
	We now turn to the reverse direction and prove that if $\mathcal{I}$ satisfies $|U_{c_1}|\leq \ell k + \sum_{i\in [k]} |U_{c_{i+1}}|$ and satisfies
	$\ell > 0$ and $n\geq k$ or satisfies $n \ge 2k$, then we can construct a left-perfect $\ell$-fair many-to-one matching. 
	For this, let $M$ be the matching constructed by Algorithm 1 presented in the first part for the case without non-emptiness constraint (Algorithm 1 is still applicable here as $|U_{c_1}|\leq \ell k + \sum_{i\in [k]} |U_{c_{i+1}}|$ is satisfied). 
	Suppose that $M$ matches vertices from $U$ to $k'$ vertices $\{ v_1,\dots, v_{k'} \}$ from $V$, i.e., $M(v_i) \ne \emptyset$ for $i \in [k']$.
	By the assumption that $|U_{c_1}|\leq \ell\cdot k + \sum_{i\in [k]} |U_{c_i}|$, it holds that $k\geq k'$. 
	Moreover, as argued above, $M(v_i)$ is $\ell$-fair for every $i \in [k']$.
	If $k' = k$, then we are done.
	So assume that $k' < k$.
	
	We now reassign some vertices and match them to $v_i$ for $i>k'$ to satisfy the non-emptiness constraint. We distinguish two cases based on whether $\ell>0$ or $\ell=0$:	
	
	If $\ell >0$, for each $j\in [k'+1,k]$, we do the following.
	Choose an arbitrary $v_i$ for $i\in [j-1]$ with $|M(v_i)|\geq 2$ and some vertex $u \in M(v_i)$ of the most frequent color in $M(v_i)$ (such a vertex always exists as $n\geq k$).
	We delete the edge $\{u,v_i\}$ from $M$ and match $u$ to $v_j$ instead.
	Clearly, $M(v_j)=\{u\}$ is $1$-fair.
	Moreover, $M(v_i)$ remains $\ell$-fair:
	Unless $\MOV(M(v_i)) = 0$, deleting $u$ does not increase $\MOV(M(v_i))$ (and in case $\MOV(M(v_i)) = 0$ this only increases the margin of victory by one).
	Since we always assume that $n \ge k$, this procedure eventually produces a left-perfect $\ell$-fair many-to-one matching. 
	
	It remains to consider the case $\ell=0$. We again describe an algorithm that takes as input the left-perfect $\ell$-fair many-to-one matching $M$ with $M(v_i)\neq \emptyset$ for $i\in [k']$	returned by Algorithm 1 and construct from this a matching respecting the non-emptiness constraint. 
	We start by executing the following steps exhaustively one after another and immediately stop if $M(v)\neq \emptyset$ for all $v\in V$:
	\begin{itemize}
		\item
		If $|M(v_i)| > 2$ for some $i \in [2, k']$, then let $u\in M(v_i)_{c_1}$ and $u'\in M(v_i)_{c_{i + 1}}$. We delete $\{u,v_i\}$ and $\{u',v_i\}$ from $M$ and match both of them to some vertex $v\in V$ with $M(v)=\emptyset$.
		Observe that $M(v_i)$ remains $0$-fair, as it only contains vertices of color $c_1$ and $c_{i+1}$ and that $M(v) = \{ u, u' \}$ also remains $0$-fair.
		\item
		If $M(v_1)$ contains vertices of at least four colors (note that $M(v_1)$ contains at least one vertex of color $c_1$ and $c_2$), then choose two vertices $u, u'$ of distinct colors in $C \setminus \{ c_1, c_2 \}$.
		Then, delete $\{u,v_1\}$ and $\{u',v_1\}$ from $M$ and set $M(v)=\{ u, u' \}$ for some $v\in V$ with $M(v)=\emptyset$. As the number of occurrences of the two most frequent colors (i.e., $c_1$ and $c_2$) in $M(v_1)$ does not change, $M(v_1)$ remains $0$-fair. 
		\item
		If $M(v_1)$ contains vertices of $c_1$ and $c_2$, and exactly one more color, say, $c_3'$, and $|M(v_1)_{c_1}| = |M(v_1)_{c_2}| > |M(v_1)_{c_3'}|$, then let $u\in M(v_1)_{c_1}$ and $u'\in M(v_1)_{c_2}$. We delete $\{u,v_1\}$ and $\{u',v_1\}$ from $M$ and set $M(v)=\{ u, u' \}$ for some $v\in V$ with $M(v)=\emptyset$. As the number of vertices of color $c_1$ and $c_2$ both decrease by exactly one, $M(v_1)$ remains $0$-fair.
	\end{itemize}
	Suppose that after executing this procedure there is still a $v\in V$ with $M(v)=\emptyset$. 
	Then, there are two possibilities: 
	If $M(v_1)$ contains only vertices of $c_1$ and $c_2$, then as long as $|M(v_1)|>2$ we can execute the procedure specified in the first case on $M(v_1)$ until $M(v)\neq \emptyset$ for all $v\in V$ (this is always possible as $n\geq 2k$ and in this case we have $|M(v_i)|\in {0,2}$ for all $i\in [2,k]$). 
	
	Otherwise, $M(v_1)$ contains vertices of three colors $c_1$, $c_2$, and $c'_3$. 
	As $M(v_1)$ is $0$-fair throughout the procedure, we then have $\alpha := |M(v_1)_{c_1}|=|M(v_1)_{c_2}|=|M(v_1)_{c_3'}|$ (note that by construction of Algorithm 1, we have that initially $c_1$ and $c_2$ appear at least as often as any other color in $M(v_1)$) and $|M(v_1)_c| = 0$ for all $c \in C \setminus \{ c_1, c_2, c_3' \}$.
	Since $|M(v)|=2$ or $|M(v)|=0$ for all $v\in V\setminus \{v_1\}$, there are $\frac{1}{2}(n - 3 \alpha) + 1$ vertices $v\in V$ with $M(v)\ne \emptyset$.
	Thus, it remains to show that $M(v_1)$ can be partitioned into $\kappa := k - \frac{1}{2}(n - 3 \alpha)$ sets that are all $0$-fair.
	By the assumption that $n \ge 2k$, we have $2 \kappa \le 3 \alpha$.
	
	As argued above, it suffices to show that the constructed set $M(v_1)$ can be partitioned into $\kappa$ $0$-fair sets for any integer $\kappa > 0$ with $2 \kappa \le 3 \alpha$, i.e., $\kappa \le \lfloor 3 \alpha / 2 \rfloor$.
	For each $i \in [3]$, let $u_i^1, \dots, u_i^\alpha$ be the vertices of color $c_i$ ($c_i'$ for $i = 3$) in $M(v_1)$.
	Observe that $M(v_1)$ can be partitioned into $\kappa$ $\ell$-fair sets whenever $\kappa \le \alpha$:
	Let the $j$th set be $\{ u_1^j, u_2^j, u_3^j \}$ for each $j \in [\kappa - 1]$ and let the $\kappa$th set be the set of all other vertices.
	Now we address the remaining case by proving via induction on $\alpha$ that there is a partition of $M(v_1)$ into $\kappa$ $0$-fair sets for any integer $\kappa$ with $\alpha < \kappa \le \lfloor 3 \alpha / 2 \rfloor$.
	We consider four cases depending on the value of $\alpha$.
	\begin{itemize}
		\item
		It clearly holds for $\alpha = 1$, since there is no integer $\kappa$ satisfying $\alpha < \kappa \le \lfloor 3 \alpha / 2 \rfloor$.
		\item
		For $\alpha = 2$, we can infer that $\kappa = 3$.
		Observe that $M(v_1)$ can be partitioned into $\kappa = 3$ $0$-fair sets: $\{ u_1^1, u_2^2 \}, \{ u_2^1, u_3^2 \}, \{ u_3^1, u_1^2 \}$.
		\item
		Suppose that $\alpha \ge 3$ and that $\alpha$ is odd, i.e., $\alpha = 2 \beta + 1$ for some $\beta \in \mathbb{N}$.
		Then, $\kappa \le \lfloor 3 \alpha / 2 \rfloor = 3 \beta + 1$.
		By the induction hypothesis, $M(v_1) \setminus \{ u_1^\alpha, u_2^\alpha, u_3^\alpha \}$ can be partitioned into $\kappa'$ $0$-fair sets for each $\kappa' \le \lfloor 3 (\alpha - 1) / 2 \rfloor = 3 \beta$.
		Since $\{u_1^\alpha, u_2^\alpha, u_3^\alpha\}$ is a $0$-fair set, we obtain a desired partition of $M(v_1)$.
		\item
		Suppose that $\alpha \ge 3$ and that $\alpha$ is even, i.e., $\alpha = 2 \beta$ for $\beta \in \mathbb{N}$.
		Then, $\kappa \le \lfloor 3 \alpha / 2 \rfloor = 3 \beta$.
		By the induction hypothesis, $M(v_1) \setminus \{ u_1^{\alpha - 1}, u_1^\alpha, u_2^{\alpha - 1}, u_2^\alpha, u_3^{\alpha - 1}, u_3^\alpha \}$ can be partitioned into $\kappa'$ $0$-fair sets for each $\kappa' \le \lfloor 3 (\alpha - 2) / 2 \rfloor = 3 \beta - 3$.
		Since  $\{ u_1^{\alpha - 1}, u_1^\alpha, u_2^{\alpha - 1}, u_2^\alpha, u_3^{\alpha - 1}, u_3^\alpha \}$ can be partitioned into three $0$-fair sets as in the case $\alpha = 2$, we obtained a desired partition of $M(v_1)$.
	\end{itemize}
	Thus, we have constructed a desired partition of $M(v_1)$ into $\kappa$ $\ell$-fair sets and can assign to $v_1$ and to each vertex $v\in V$ with $M(v)=\emptyset$ one of these sets, thereby obtaining a left-perfect $0$-fair many-to-one matching $M$.
\end{proof}

We now show an analogous statement for \textsc{Max-Min Fair Matching}. Unsurprisingly, while for MoV the relation between the number of occurrences of the most frequent color and the summed occurrences of other colors is important, here the relation between the number of occurrences of the most and least frequent color are decisive. 
In fact, our constructed solution here is slightly simpler as for MoV as we only need to distribute the vertices of each color as evenly as possible. 

\begin{restatable}{theorem}{cliqueMaxMin}\label{th:cliqueMaxMin}
	A \textsc{Max-Min Fair Matching} instance $\mathcal{I}=(G=(U \cupdot V,E),C,\col,\ell)$ with $G$ being a complete bipartite graph is a yes-instance if and only if $|U_{c_1}|\leq \ell k + |U_{c_{|C|}}|$.
	With the non-emptiness constraint, $\mathcal{I}$ is a yes-instance if and only if it additionally satisfies: $\ell > 0$ and $n\ge k$, or $\ell=0$ and $|U_{c_1}| \ge k$. 
\end{restatable}
\begin{proof}
    As a reminder, we assume that $V = \{ v_1, \dots, v_k \}$ and that $|U_{c_i}|\geq |U_{c_{i+1}}|$ for each $i\in [1,|C|-1]$.
    We first prove the theorem for $\smin=0$ and later extend the proof to the case $\smin=1$. 
	
	\proofsubparagraph*{\textbf{No size constraints}.}
	We consider both directions separately.
	Assume that there is a left-perfect $\ell$-fair many-to-one matching $M$ in $\mathcal{I}$, then as $M$ is $\ell$-fair, it holds that $|M(v)_{c_1}|\leq |M(v)_{c_{|C|}}|+\ell$ for each $v\in V$ and thus that $|U_{c_1}|=\sum_{v\in V} |M(v)_{c_1}|\leq \ell k + \sum_{v\in V} |M(v)_{c_{|C|}}|=\ell k + |U_{c_{|C|}}|$.
	
	We now prove the reverse direction: 
	If $\mathcal{I}$ fulfills $|U_{c_1}|\leq \ell k + |U_{c_{|C|}}|$, then there is a left-perfect $\ell$-fair many-to-one matching $M$ in $\mathcal{I}$. 
	Note that as $|U_{c_i}|\geq |U_{c_{i+1}}|$ for each $i\in [1,|C|-1]$, $\mathcal{I}$ also fulfills $|U_{c_j}|\leq \ell k + |U_{c_{i}}|$ for all $i, j\in [1,|C|]$.
	We construct $M$ as follows. 
	For each color $c\in C$, we execute the following procedure as long as there are still unassigned vertices in $U_c$.
	Set $i:=1$ and assign a vertex from $U_c$ not assigned yet to $v_i$ and increase $i$ by one if $i<k$ and set $i:=1$ otherwise.
	The constructed matching $M$ is clearly a left-perfect many-to-one matching so it remains to prove that $M(v_i)$ is $\ell$-fair for each $v_i\in V$. 
	For the sake of contradiction assume that there is a $v_{i^\star}\in V$ such that $M(v_{i^\star})$ is not $\ell$-fair, i.e., there are two colors $c',c''\in C$ with $|M(v_{i^\star})_{c'}|\geq|M(v_{i^\star})_{c''}|+\ell+1$. 
	By the construction of $M$ it holds for each color $c\in C$ that $|M(v_j)_c|\geq|M(v_{i^\star})_c|\geq |M(v_j)_c|-1$ for $j<i^\star$  and  $|M(v_{i^\star})_c|-1\leq |M(v_j)_c|\leq|M(v_{i^\star})_c|$ for $j>i^\star$. 
	Thus, from $|M(v_{i^\star})_{c'}|\geq|M(v_{i^\star})_{c''}|+\ell+1$ it follows that $|M(v_i)_{c'}|-|M(v_i)_{c''}|\geq \ell$ for all $i\in [k]$. 
	Consequently, $|U_{c'}|-|U_{c''}|=|M(v_{i^\star})_{c'}|-|M(v_{i^\star})_{c''}|+ \sum_{j\in [k]-i^*}|M(v_j)_{c'}|-|M(v_j)_{c''}|\geq \ell+ 1 + (k-1)\ell >\ell k$, a contradiction to our initial assumption on $\mathcal{I}$. 
	
	\proofsubparagraph*{\textbf{Non-emptiness constraint}.}
	We again consider both directions separately.
	Assume that there is a left-perfect $\ell$-fair many-to-one matching $M$ in $\mathcal{I}$ with $M(v)\neq \emptyset$ for all $v\in V$. Then ,as argued above since $M$ is $\ell$-fair, $|U_{c_1}|\leq \ell k + |U_{c_{|C|}}|$ needs to hold.
	As $M(v) \ne \emptyset$ for every $v \in V$, $n\geq k$ holds. 
	Moreover, if $\ell=0$, then there is at least one vertex of each color in $M(v)$ for every $v\in V$. 
	Thus, $|U_{c_1}|\geq k$ follows. 
    
    For the reverse direction, assume that $\mathcal{I}$ fulfills $|U_{c_1}|\leq \ell k + |U_{c_{|C|}}|$ and that $\ell > 0$ and $n\ge k$ or $\ell=0$ and $|U_{c_1}| \ge k$. 
    We make a case distinction based on $|U_{c_1}|$. 
    If $|U_{c_1}| \ge k$, then the procedure the case without the non-emptiness constraint produces a left-perfect $\ell$-fair many-to-one matching $M$ with $|M(v)|\geq 1$ for all $v\in V$. 
    If $|U_{c_1}| <k$, then it holds that $\ell>1$ and $|U_{c}| <k$ for each $c\in C$. 
    As $n\geq k$, it is clearly possible to construct a left-perfect many-to-one matching $M$ of vertices from $U$ to vertices from $V$ such that $|M(v)|\geq 1$ and $|M(v)_c|\leq 1$ for each $v\in V$ and $c\in C$. As $M$ is clearly $\ell$-fair, the statement follows.	
\end{proof}

\Cref{th:clique,th:cliqueMaxMin} imply that \textsc{Max-Min/MoV Fair Matching} on a complete bipartite graph are solvable in linear~time even with the non-emptiness constraint. 

\section{Conclusion}

In this work, we have investigated the (parameterized) computational complexity of the \textsc{Fair Matching} problem.
Two concrete directions of open questions are:
\begin{itemize}
	\item
	We have provided algorithms that solve \textsc{Fair Matching} even if we require that every vertex in the right side is matched to at least one vertex.
 	Can we extend our algorithms to handle arbitrary size constraints?
	In particular, does \textsc{Fair Matching} remain fixed-parameter tractable with respect to $k$?
	We have shown in \Cref{sec:kc} that \textsc{Fair Matching} is indeed FPT with respect to $k + |C|$ even for arbitrary size constraints.
	However, it does not seem straightforward to incorporate arbitrary size constraints in the ILPs given in \Cref{sec:k}.
	The complexity of \textsc{Fair Matching} on complete bipartite graphs (\Cref{se:cliques}) is also open when arbitrary size constraints are present.
	\item 
	Is \textsc{Fair Matching} solvable in $O^\star(2^k)$ time?
	Note that the ILP presented in \Cref{sec:k:mmm} (which solves \textsc{Max-Min Fair Matching} without the non-emptiness constraint) is an ILP where the constraint matrix involves only zeros and ones when $\ell = 0$.
	Can we exploit such a structure in the constraint matrix to obtain a faster algorithm?
\end{itemize}
For future research, it would also be natural to study other variants of the \textsc{Fair Matching} problem.
For instance, we may relax the left-perfect constraint studied in this work and consider a variant where the objective is to maximize the matching size under a fairness constraint.
One may also look into other fairness notions such as proportionality constraints~\cite{DBLP:journals/ior/NguyenV19}. 

\bibliographystyle{plain}
\bibliography{literature}

\appendix
\appendixText
\end{document}